\ifdefined\PPMReview
\documentclass[acmsmall,screen,review,anonymous]{acmart}
\else
\documentclass[acmsmall,screen,nonacm]{acmart}
\fi
\setcopyright{none}
\ifdefined\PPMReview
\renewcommand\footnotetextcopyrightpermission[1]{}
\makeatletter
\AtBeginDocument{%
  \apptocmd{\ps@standardpagestyle}{\fancyfoot{}}{}{\ClassError{ppm}{Cannot clear review footer}{}}
  \apptocmd{\ps@firstpagestyle}{\fancyfoot{}}{}{\ClassError{ppm}{Cannot clear first-page footer}{}}
  \pagestyle{standardpagestyle}}
\makeatother
\fi
\setcitestyle{citesep={,}}
\makeatletter
\patchcmd{\ccsdesc@parse}{\else; \fi}{\else. \fi}{}{\ClassError{ppm}{Cannot update CCS punctuation}{}}
\patchcmd{\@mkauthorsaddresses}{\unskip; }{\unskip. }{}{\ClassError{ppm}{Cannot update author-address punctuation}{}}
\makeatother
\hypersetup{hypertexnames=false}

\usepackage{amsmath,amsthm}
\usepackage{booktabs,array}
\usepackage{mathpartir}
\usepackage{algorithm}
\usepackage{algpseudocode}
\usepackage{tikz}
\usetikzlibrary{arrows.meta}
\definecolor{ppmx}{RGB}{239,181,107}
\definecolor{ppmy}{RGB}{181,157,216}
\definecolor{ppmz}{RGB}{114,177,232}
\definecolor{ppmi}{RGB}{238,238,238}
\definecolor{ppmclassical}{RGB}{215,241,191}
\tikzset{
 pauli/X/.style={fill=ppmx}, pauli/Y/.style={fill=ppmy},
 pauli/Z/.style={fill=ppmz}, pauli/I/.style={fill=ppmi},
 pauli/P/.style={fill=white},
 pics/ppmone/.style args={#1/#2}{code={
  \draw[pauli/#1] (-.25,.25) rectangle (.25,-.25);
  \node at (0,0) {$#1$};
  \draw[fill=white] (.25,.14) arc[start angle=90,end angle=-90,radius=.14] -- cycle;
  \draw (.27,-.035) -- (.36,.08);
  \node[font=\scriptsize] at (0,-.5) {$#2$};
 }},
 pics/ppmtwo/.style args={#1/#2/#3}{code={
  \fill[pauli/#1] (-.24,.25) rectangle (.24,-.375);
  \fill[pauli/#2] (-.24,-.375) rectangle (.24,-1.0);
  \draw (-.24,.25) rectangle (.24,-1.0);
  \node at (0,0) {$#1$}; \node at (0,-.75) {$#2$};
  \draw[fill=white] (.24,-.25) arc[start angle=90,end angle=-90,radius=.125] -- cycle;
  \draw (.26,-.4) -- (.34,-.31);
  \node[font=\scriptsize] at (0,-1.2) {$#3$};
 }}
}

\newcommand{\pwire}[2][P]{\mathord{\begin{tikzpicture}[baseline=(c.base)]
 \node[draw,pauli/#1,inner xsep=4pt,inner ysep=3pt] (c) {$#2$};
 \draw (c.west) -- ++(-.23,0) (c.east) -- ++(.23,0);
\end{tikzpicture}}}

\newcommand{\resetfeedback}{\mathord{\begin{tikzpicture}[
 baseline=(current bounding box.center),font=\small,>=Latex]
 \node[draw,fill=ppmclassical,inner sep=3pt] (r) at (0,0) {$r$};
 \node[draw,pauli/X,inner sep=3pt] (x) at (.95,.35) {$X$};
 \draw (-.45,0) -- (r.west);
 \draw[->] (r.east) -- (.4,0) |- (x.west);
 \draw (r.east) -- (.4,0) -- (.4,-.35) -- (1.5,-.35) -- (1.5,0);
 \draw (x.east) -| (1.5,0) -- (1.75,0);
 \node[font=\scriptsize] at (.47,.54) {$1$};
 \node[font=\scriptsize] at (.45,-.55) {$0$};
 \node[font=\scriptsize,fill=white,inner sep=1pt] at (.95,-.35) {$\cmd{skip}$};
\end{tikzpicture}}}
\newcommand{\pmeter}[3][P]{\mathord{\begin{tikzpicture}[baseline=(c.base)]
 \node[draw,pauli/#1,inner xsep=5pt,inner ysep=3pt] (c) {$#2$};
 \draw (c.west) -- ++(-.23,0) (c.east) -- ++(.33,0);
 \draw[fill=white] ([yshift=3pt]c.east) arc[start angle=90,end angle=-90,radius=3pt] -- cycle;
 \draw ([xshift=1pt,yshift=-1pt]c.east) -- ++(2pt,3pt);
 \node[anchor=north,font=\scriptsize,inner sep=1pt] at (c.south) {$#3$};
\end{tikzpicture}}}
\newcommand{\qassert}[2]{\mathord{\begin{tikzpicture}[baseline=(a.base)]
 \node[draw,pauli/#1,rounded corners=2pt,inner sep=2pt] (a) {$#2$};
\end{tikzpicture}}}
\newcommand{\htriple}[3]{\{#1\}\,\pwire{#2}\,\{#3\}}
\newcommand{\dtriple}[3]{\{#1\}\,#2\,\{#3\}}
\newcommand{\pchoice}[3]{\mathord{\begin{tikzpicture}[baseline=(current bounding box.center),font=\small,>=Latex]
 \node[draw,fill=ppmclassical,inner sep=3pt] (e) at (0,0) {$#1$};
 \node[draw,fill=white,inner sep=3pt] (t) at (1,.42) {$#2$};
 \node[draw,fill=white,inner sep=3pt] (f) at (1,-.42) {$#3$};
 \draw (-.55,0) -- (e.west);
 \draw[->] (e.east) -- (.45,0) |- (t.west);
 \draw[->] (e.east) -- (.45,0) |- (f.west);
 \node[font=\scriptsize,inner sep=1pt] at (.57,.62) {$1$};
 \node[font=\scriptsize,inner sep=1pt] at (.57,-.62) {$0$};
 \draw (t.east) -| (1.7,0) (f.east) -| (1.7,0);
 \draw (1.7,0) -- (2,0);
\end{tikzpicture}}}
\newenvironment{ppmcompiledrow}[4][.65]{%
 \def\ppmcompiledn{#2}%
 \pgfmathtruncatemacro{\ppmcompiledlast}{#2-1}%
 \pgfmathsetmacro{\ppmcompiledbottom}{-.42*#2}%
 \pgfmathsetmacro{\ppmcompiledfeedback}{\ppmcompiledbottom-.48}%
 \pgfmathsetmacro{\ppmcompiledassert}{\ppmcompiledbottom-#1}%
 \begin{tikzpicture}[x=1cm,y=1cm,font=\scriptsize,>=Latex,
   baseline=(current bounding box.center)]
  \foreach \j in {0,...,\ppmcompiledlast}{
   \pgfmathsetmacro{\ppmwirey}{-.42*(\j+1)}
   \draw (0,\ppmwirey) -- (12,\ppmwirey);
   \node[anchor=east,inner sep=2pt] at (-.08,\ppmwirey) {$q_\j$};
  }
  \node[anchor=west,inner sep=1pt] at (0,\ppmcompiledassert) {$#3$};
  \node[anchor=east,inner sep=1pt] at (12,\ppmcompiledassert) {$#4$};
}{%
 \end{tikzpicture}%
}

\newcommand{\ppmscope}[4]{%
 \begin{scope}[xshift=#1cm]
  \pgfmathsetmacro{\ppmcolumnbottom}{\ppmcompiledbottom-.21}
  \pgfmathsetmacro{\ppmmetery}{-.21*\ppmcompiledn}
  \node[anchor=east,inner sep=1pt] at (.87,0) {$|0\rangle_{\nu_{#2}}$};
  \node[inner sep=1pt] at (2,.43) {scope $#2$};
  \draw (.89,0) -- (3.78,0);
  \fill[pauli/X] (1.02,.2) rectangle (1.48,-.21);
  \node at (1.25,0) {$#3X$};
  \foreach \ppmfactor [count=\ppmindex from 1] in {#4}{
   \pgfmathsetmacro{\ppmwirey}{-.42*\ppmindex}
   \fill[pauli/\ppmfactor] (1.02,\ppmwirey+.21) rectangle (1.48,\ppmwirey-.21);
   \node at (1.25,\ppmwirey) {$\ppmfactor$};
  }
  \draw (1.02,.2) rectangle (1.48,\ppmcolumnbottom);
  \draw[fill=white] (1.48,\ppmmetery+.11)
   arc[start angle=90,end angle=-90,radius=.11cm] -- cycle;
  \draw (1.51,\ppmmetery-.03) -- (1.59,\ppmmetery+.06);
  \node[anchor=east,inner sep=1pt] at (1.16,\ppmcolumnbottom-.12) {$a$};
  \draw[pauli/Y] (1.93,.17) rectangle (2.31,-.17);
  \node at (2.12,0) {$Y$};
  \draw[fill=white] (2.31,.11)
   arc[start angle=90,end angle=-90,radius=.11cm] -- cycle;
  \draw (2.34,-.03) -- (2.42,.06);
  \node[anchor=south,inner sep=1pt] at (2.75,-.29) {$b$};
  \node[fill=white,inner sep=1pt] (ppmxor#2) at (3.12,\ppmcompiledfeedback) {$\oplus$};
  \draw[double,double distance=.45pt]
   (1.25,\ppmcolumnbottom) |- (ppmxor#2.west);
  \draw[double,double distance=.45pt]
   (2.12,-.17) -- (2.12,-.29) -- (3.64,-.29)
   |- (ppmxor#2.east);
  \def\ppmidentity{I}
  \foreach \ppmfactor [count=\ppmindex from 1] in {#4}{
   \edef\ppmcurrentfactor{\ppmfactor}
   \ifx\ppmcurrentfactor\ppmidentity
   \else
    \pgfmathsetmacro{\ppmwirey}{-.42*\ppmindex}
    \draw[double,double distance=.45pt,->]
     (ppmxor#2.north) -- (3.12,\ppmwirey-.16);
   \fi
  }
  \foreach \ppmfactor [count=\ppmindex from 1] in {#4}{
   \edef\ppmcurrentfactor{\ppmfactor}
   \ifx\ppmcurrentfactor\ppmidentity
   \else
    \pgfmathsetmacro{\ppmwirey}{-.42*\ppmindex}
    \node[draw,pauli/\ppmfactor,minimum width=.36cm,minimum height=.32cm,inner sep=1pt]
     at (3.12,\ppmwirey) {$\ppmfactor$};
   \fi
  }
  \draw (3.7,-.13) -- (3.7,.13) (3.78,-.13) -- (3.78,.13);
 \end{scope}%
}

\newcommand{\ppmreadout}[4]{%
 \pgfmathsetmacro{\ppmwirey}{-.42*(#2+1)}%
 \begin{scope}[shift={(#1,\ppmwirey)}]
  \draw[pauli/#3] (-.2,.17) rectangle (.2,-.17);
  \node at (0,0) {$#3$};
  \draw[fill=white] (.2,.11)
   arc[start angle=90,end angle=-90,radius=.11cm] -- cycle;
  \draw (.23,-.03) -- (.31,.06);
  \node[anchor=west,inner sep=1pt] at (.25,-.23) {$#4$};
 \end{scope}%
}

\newcommand{\ppmcorrection}[2]{%
 \def\ppmidentity{I}%
 \foreach \ppmfactor [count=\ppmindex from 1] in {#2}{
  \edef\ppmcurrentfactor{\ppmfactor}
  \ifx\ppmcurrentfactor\ppmidentity
  \else
   \pgfmathsetmacro{\ppmwirey}{-.42*\ppmindex}
   \node[draw,pauli/\ppmfactor,minimum width=.36cm,minimum height=.32cm,inner sep=1pt]
    at (#1,\ppmwirey) {$\ppmfactor$};
  \fi
 }%
}

\newtheorem{theorem}{Theorem}
\newcommand{\cmd}[1]{\mathsf{#1}}
\newcommand{\seq}{\mathbin{\triangleright}}
\newcommand{\WP}{\operatorname{wp}}
\newcommand{\tr}{\operatorname{tr}}
\newcommand{\Inst}{\mathcal{I}}
\newcommand{\triple}[3]{\{#1\}\,#2\,\{#3\}}
\makeatletter
\newcommand{\rulename}[2]{\begingroup
 \def\@currentlabel{\textsc{#1}}\phantomsection\label{#2}\textsc{#1}\endgroup}
\makeatother

\begin{document}
\title{Verified Pauli-Product Measurement Compilation for Fault-Tolerant Quantum Computation}
\ifdefined\PPMReview
\author{Anonymous Author(s)}
\renewcommand{\shortauthors}{Anonymous Authors}
\else
\author{John Zhuoyang Ye}
\orcid{0009-0007-7755-3112}
\affiliation{\institution{University of California Los Angeles}\city{Los Angeles}\country{USA}}
\email{yezhuoyang@cs.ucla.edu}
\author{Jens Palsberg}
\orcid{0009-0002-1149-0706}
\affiliation{\institution{University of California Los Angeles}\city{Los Angeles}\country{USA}}
\email{palsberg@cs.ucla.edu}
\renewcommand{\shortauthors}{John Zhuoyang Ye and Jens Palsberg}
\fi

\begin{abstract}
Fault-tolerant quantum compilers implement gates through Pauli-product measurements and
classical feedback. Verifying such a compiler requires reasoning about the quantum state
and the recorded outcomes together: a correction reads an outcome, a gate's internal
outcomes must be hidden from its caller, and an application's outcomes must survive
optimization. We introduce QPPM, a language for Pauli-product measurement programs, with a
weakest-precondition semantics over store-indexed matrix assertions, following quantum
Hoare logic with classical variables. Every program terminates, so the Hoare logic is a
syntactic presentation of weakest preconditions whose soundness and relative completeness
are immediate. Its value is that component specifications compose. All results are
mechanized in Lean.

We verify lowering of Clifford+$T$ circuits to QPPM with private outcomes and
supplied magic states, and an optimizer that eliminates measurements, fuses and reorders
rotations, and schedules logical rounds while preserving every outcome that subsequent
code can read. A compositional approximation relation bounds the effect of synthesis error
on specifications and on output trace distance, and a certificate-checked pipeline
selects, among synthesized circuits that meet the user's error limit, the one with the
lowest certified depth bound. On QASMBench, verified fusion and scheduling reduce logical
measurement depth by a median of 17.9\%, and complete source-to-QPPM allowances are
certified for all 35 valid selected small-suite inputs.
\end{abstract}

\maketitle
\ifdefined\PPMReview\else
\hypersetup{pdfauthor={John Zhuoyang Ye and Jens Palsberg}}
\fi

% !TEX root = main.tex
% Mathematical presentation shared by the author and review manuscripts.

\section{Introduction}
\label{sec:introduction}

Fault-tolerant quantum computation encodes logical qubits in quantum error-correcting codes.
A compiler must translate an algorithm's gates into operations on these encoded qubits.
For surface codes, lattice surgery implements logical gates through joint Pauli
measurements~\cite{horsman2012surgery,litinski2019game}. This measurement interface also
extends to quantum low-density parity-check (qLDPC) codes.

Code-surgery constructions for qLDPC codes support logical Pauli-product measurements
(PPMs)~\cite{cohen2022low,williamson2026gauging}, including joint measurements across
code blocks~\cite{cross2024surgery} and parallel measurements on logically disjoint
sets of qubits~\cite{cowtan2026parallel}. Despite different connectivity requirements and
physical costs, these constructions share a logical measurement interface.

This interface already connects algorithms to proposed architectures.
\citet{cain2026shor} estimate the resources for Shor's algorithm, including RSA--2048
factorization, on neutral-atom architectures using lifted-product codes.
In their compilation strategy, PPMs serve as an intermediate representation: Clifford
gates are absorbed into the measured Pauli operators, and code surgery implements the
measurements. Their time-efficient estimates assume parallel code surgery for layers
of logically disjoint PPMs. This use of PPMs motivates the interface in
Figure~\ref{fig:compilation-stages}.

\begin{figure}[htbp]
\centering
\begin{tikzpicture}[font=\small,>=Latex,
 stage/.style={draw,rounded corners=2pt,minimum height=.85cm,text width=2.5cm,align=center}]
 \node[stage] (source) at (0,0) {Source circuit\\rotations and readouts};
 \node[stage] (basis) at (4.1,0) {Clifford+$T$ circuit};
 \node[stage] (ppm) at (8.5,0) {QPPM program};
 \draw[->] (source) -- node[above=13pt,font=\scriptsize] {gate synthesis}
   node[below=7pt,font=\scriptsize] {allowance $\varepsilon$} (basis);
 \draw[->] (basis) -- node[above=13pt,font=\scriptsize] {exact lowering} (ppm);
 \draw[->] (ppm.north east) -- ++(0,.65) -| (ppm.north west);
 \node[font=\scriptsize] at (8.5,1.28) {exact optimization};
 \node[draw,fill=ppmclassical,align=center,text width=2.5cm] (magic) at (4.1,-1.55)
   {Magic states\\ideal inputs};
 \draw[->] (magic.east) -| (6.3,0);
 \node[stage,dashed,fill=black!4] (backend) at (8.5,-1.65) {Code surgery\\surface or qLDPC};
 \draw[->,dashed] (ppm.south) -- node[right,font=\scriptsize] {physical mapping} (backend.north);
 \node[anchor=west,align=left,font=\scriptsize,text width=3.6cm] at (-1.25,-1.55)
   {The source bound transfers through\\lowering and exact optimization.};
\end{tikzpicture}
\caption{Compilation stages. Synthesis error is certified per input, lowering to QPPM is
verified under explicit magic-state assumptions, and exact transformations preserve the
public specifications. Algorithm~\ref{alg:compile} specifies the solid arrows. The dashed
stage is code-specific physical mapping.}
\Description{Source circuits pass through synthesis to Clifford plus T and exact lowering to
QPPM. Ideal magic states are inputs to lowering. QPPM has an exact optimization loop.
A dashed arrow leads to surface-code or qLDPC code surgery, whose proof obligations are separate.
Algorithm 1 specifies the stages through optimized QPPM.}
\label{fig:compilation-stages}
\end{figure}
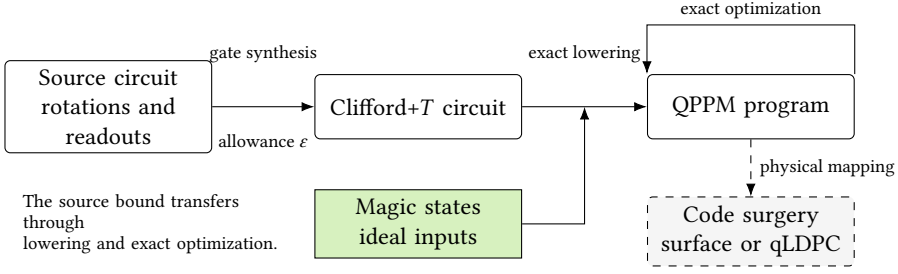

Existing lattice-surgery compilers translate circuits into measurement
instructions~\cite{watkins2024compiler,leblond2024realistic}. They do not provide a
machine-checked proof that this translation preserves an algorithm's behavior, leaving a
gap between algorithm verification and error-correcting implementations.
Such a proof must cover both quantum outputs and classical outcomes: a correction reads a
measurement result, a gate's private outcomes must be hidden from its caller, and public
outcomes must survive optimization.

We formalize this measurement interface in QPPM, a language for Pauli-product measurements
and classical feedback. Its Hoare logic composes specifications according to program structure: a verified
phase-gate specification can be reused without expanding the gate's internal measurements.

Quantum Hoare logic interprets assertions as matrices whose expectations describe quantum
states~\cite{ying2009hoare}. Following \citet{feng2021classical}, these matrices depend on the
classical store. An assertion can relate a measurement's output state to its recorded outcome.
Sequencing (\ref{rule:seq}) connects that specification to a later correction. We prove soundness
and relative completeness for QPPM's weakest-precondition semantics. A forward semantics gives
the output states and their probabilities.

Compiled-gate specifications hide private outcomes by summing over their possible values.
Optimizations preserve the outcomes available to a continuation. We verify measurement
elimination and rotation scheduling under these respective conditions.

Approximate synthesis changes a program's behavior. We bound the change in the expectation
of every effect, a matrix predicate whose expectation lies between zero and one.
These bounds compose: replacing a command with allowance $\varepsilon$ lowers its Hoare
guarantee by at most $\varepsilon$. Our compiler checks a finite collection of Clifford+$T$
approximations and, among those meeting the user's error limit, selects the implementation
with the smallest certified measurement-depth bound (Section~\ref{sec:budget-method}).

Our contributions are:
\begin{enumerate}
 \item We define QPPM and prove a sound and relatively complete Hoare logic for its
   measurements, classical control, and scoped ancillas (Sections~\ref{sec:core}--\ref{sec:hoare}).
 \item We verify compilation into QPPM and optimizations that preserve public quantum
   outputs and measurement results. Worked derivations show how their specifications
   compose (Sections~\ref{sec:compilation}--\ref{sec:optimization}).
 \item We prove compositional approximation bounds and their connection to Hoare
   specifications and classical--quantum trace distance. A separate implementation contract
   adds logical-noise bounds without extending the core calculus (Section~\ref{sec:approximation}).
\end{enumerate}
We mechanize these results in Lean and evaluate them on QASMBench.
Verified fusion and scheduling reduce logical measurement depth by 11.8--28.6\% in six
of eight Clifford circuits, against a parallel baseline with matched ancilla resources.
Our proofs concern typed source circuits and ideal logical measurements. Physical mapping
and decoding remain outside their scope.

\section{Background and compilation overview}
\label{sec:example}

\subsection{Pauli-product measurements}

A qubit has basis states $|0\rangle$ and $|1\rangle$, and an $n$-qubit register has
basis vectors indexed by $n$-bit strings. We use density matrices to include
probabilistic mixtures: a state $\rho$ is a positive semidefinite $2^n\times2^n$ complex
matrix with $\tr(\rho)=1$, a pure state $|\psi\rangle$ has density matrix
$|\psi\rangle\langle\psi|$, and a unitary gate $U$ acts as $\rho\mapsto U\rho U^\dagger$,
where $\dagger$ denotes conjugate transpose.

The Pauli matrices are
\[
 I=\begin{pmatrix}1&0\\0&1\end{pmatrix},\quad
 X=\begin{pmatrix}0&1\\1&0\end{pmatrix},\quad
 Y=\begin{pmatrix}0&-i\\i&0\end{pmatrix},\quad
 Z=\begin{pmatrix}1&0\\0&-1\end{pmatrix}.
\]
We form a \emph{signed Pauli product} $P=(-1)^\eta P_1\otimes\cdots\otimes P_n$
by choosing a sign bit $\eta$ and one Pauli matrix for each qubit. The tensor product
$\otimes$ places these operators on their respective qubits. An $I$ factor leaves a qubit
unchanged. Since $P$ is Hermitian and $P^2=I$, its possible measurement outcomes are $+1$
and $-1$. We encode them by a bit $b$, with eigenvalue $(-1)^b$.

Figure~\ref{fig:ppm} illustrates a measurement of $Z\otimes Z$, which returns the parity
of two qubits in the computational basis: outcome zero selects the span of
$|00\rangle,|11\rangle$ and outcome one the span of $|01\rangle,|10\rangle$, preserving
superpositions within the selected subspace, whereas measuring each qubit separately
would destroy such coherence.
We write $r:=\cmd{measure}\ P$ for measuring $P$ and recording its outcome in $r$.
The column in the figure denotes one joint measurement with one result bit. Both
quantum wires remain available afterwards. Lattice surgery implements such
measurements on encoded qubits~\cite{litinski2019game}.

\begin{figure}[htbp]
\centering
\begin{tikzpicture}[font=\small,>=Latex]
 \node at (1.5,.85) {$r:=\cmd{measure}\ (Z\otimes Z)$};
 \node[anchor=east] at (-.08,0) {$q_0$};
 \node[anchor=east] at (-.08,-.85) {$q_1$};
 \draw (0,0) -- (3,0);
 \draw (0,-.85) -- (3,-.85);
 \draw[fill=ppmz] (1,.28) rectangle (1.65,-1.13);
 \node[font=\Large] at (1.325,0) {$Z$};
 \node[font=\Large] at (1.325,-.85) {$Z$};
 \draw[fill=white] (1.65,-.24) arc[start angle=90,end angle=-90,radius=.185] -- cycle;
 \draw (1.68,-.46) -- (1.79,-.33);
 \draw[double,double distance=1.2pt,->] (1.85,-.425) -- (4.5,-.425)
   node[midway,above=3pt] {outcome bit $r$};
 \node[font=\scriptsize,align=center] at (1.5,-1.55) {one measurement of $Z\otimes Z$};
 \node[anchor=west,align=left] at (4.9,.03)
   {$r=0:\quad \operatorname{span}\{|00\rangle,|11\rangle\}$\\[7pt]
    $r=1:\quad \operatorname{span}\{|01\rangle,|10\rangle\}$};
 \node[anchor=west,font=\scriptsize,align=left] at (4.9,-1.12)
   {The outcome selects a subspace.\\Coherence within that subspace is retained.};
 \foreach \x/\p in {2.7/X,4.05/Y,5.4/Z,6.75/I} {
   \node[draw,pauli/\p,minimum width=.6cm,minimum height=.4cm] at (\x,-2.15) {$\p$};
 }
 \node[anchor=east,font=\scriptsize] at (2.15,-2.15) {Pauli factors:};
\end{tikzpicture}
\caption{Input qubits $q_0,q_1$ and classical location $r$. One outlined column denotes one joint measurement. In all circuit figures,
orange, purple, and blue identify $X,Y,Z$. Gray denotes $I$. A meter marks a measurement. The quantum wires continue after the
measurement, while the double line carries its classical result.}
\Description{Two quantum wires pass through a single measurement column labeled Z on both wires.
A legend associates X with orange, Y with purple, Z with blue, and I with gray.
One classical bit leaves the column. Zero selects the span of 00 and 11. One selects the
span of 01 and 10. The quantum wires continue after measurement.}
\label{fig:ppm}
\end{figure}
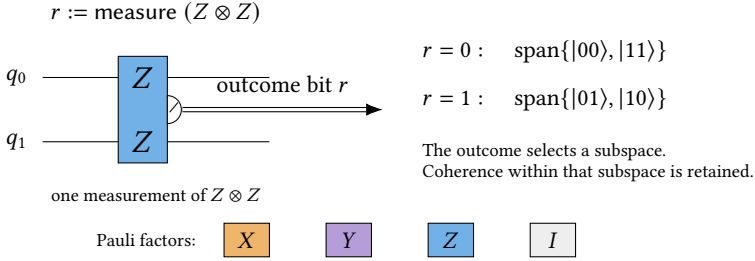

For a product $P$, define the projectors
\[
 \Pi^P_0=\frac{I+P}{2},\qquad \Pi^P_1=\frac{I-P}{2}.
\]
The output associated with bit $b$ is $\Pi^P_b\rho\Pi^P_b$, and its trace gives the
probability of that outcome. We retain this unnormalized matrix throughout the semantics:
it records the probability and the quantum state in one object. The projectors satisfy
\begin{equation}
 (\Pi^P_b)^\dagger=\Pi^P_b,\qquad
 (\Pi^P_b)^2=\Pi^P_b,\qquad
 \Pi^P_0\Pi^P_1=\Pi^P_1\Pi^P_0=0,\qquad
 \Pi^P_0+\Pi^P_1=I.
 \label{eq:projector-laws}
\end{equation}
We prove these identities from the Pauli matrices and use them in the measurement and
optimization proofs.

\subsection{Classical feedback and gate compilation}

A measurement program can read a stored outcome to choose a correction. In a gate
implementation, these corrections make the output implement the required gate after
summing over internal outcomes. Results returned to the caller must remain distinguishable.

We use the phase gate $S=\operatorname{diag}(1,i)$ as a running example. Consider an
application that measures $X$, applies $S$, and then measures $Y$. Since
$SXS^\dagger=Y$, the two measurements return the same bit. We verify this property
after replacing $S$ by measurements and feedback. The application reads its initial
and final results. The gate implementation's outcomes stay private.

To specify this property, we use assertions that relate a stored bit to a quantum
eigenspace, following \citet{feng2021classical}. Section~\ref{sec:derive-application}
composes the two measurement specifications with the compiled phase gate's specification.

The Clifford gates $H,S,\mathrm{CNOT}$ map Pauli products to Pauli products by
conjugation. Adding $T=\operatorname{diag}(1,e^{i\pi/4})$ permits approximation of
arbitrary single-qubit rotations~\cite{ross2016synthesis}.

When synthesis approximates a source rotation, we bound the change in observable
probabilities. Section~\ref{sec:approximation} composes these bounds along a circuit
and transfers its Hoare specification to the approximate implementation.

\section{QPPM: syntax and semantics}
\label{sec:core}
\subsection{Syntax}
\label{sec:syntax}

A QPPM command $c\in\mathcal{C}_{n,m}$ operates on $n$ data qubits and $m$ classical
bits. Its input and output registers have the same widths. Temporary qubits are introduced
by an ancilla scope. The syntax and metatheory are parameterized by $n$ and $m$:
\[
\begin{aligned}
 e &::= 0 \mid 1 \mid \cmd{read}\ r \mid \neg e \mid e\oplus e,\\
 c &::= \cmd{skip} \mid r:=e \mid r:=\cmd{measure}\ P \mid \cmd{correct}\ P\\
   &\quad\mid c\seq c \mid \cmd{if}\ e\ \cmd{then}\ c\ \cmd{else}\ c\\
   &\quad\mid \cmd{let}\ \nu=\cmd{ancilla}\ |b\rangle\ \cmd{in}\ c.
\end{aligned}
\]

Here $e$ is a Boolean expression, $r$ is a bit position, and $b\in\{0,1\}$ is an
initialization value. The operator $\oplus$ denotes exclusive-or. The product
$P$ contains one Pauli factor for each qubit in the current register.
The correction command applies $P$ by conjugation. The sequence $c\seq d$ executes $c$ then $d$.
Programs have declared input contexts $\Gamma=(q_0,\ldots,q_{n-1})$ for qubits and
$\Delta=(r_0,\ldots,r_{m-1})$ for classical locations, with distinct names in each.
These ordered lists name register positions, and $()$ denotes an empty register.
Every qubit reference must belong to the current $\Gamma$, and every bit read or write
must refer to $\Delta$. Measurement overwrites a declared location rather than
allocating one. The caller supplies the input state and initial classical values.
A scope binds a fresh name $\nu\notin\Gamma$ to a first qubit prepared in $|b\rangle$,
executes its body in $\mathcal C_{n+1,m}$ with context $(\nu,\Gamma)$, and discards $\nu$.
Existing qubits keep their names and move up one position inside the scope.
The name $\nu$ is local to the body. Classical writes remain available afterwards.
On a one-qubit register we omit the qubit subscript. Unmentioned tensor factors are $I$.

We abbreviate $r:=\cmd{measure}\ P$ by $M_P(r)$ and $\cmd{read}\ r$ by $r$ in guards
and assignments. Figure~\ref{fig:reset-syntax} uses these commands to reset a qubit to
$|0\rangle$ while returning a copy of the measurement record. Measurement supplies the
record. The conditional applies a correction only to outcome one. Assignment makes the
copy. We use this same program for the backward, forward, and Hoare derivations below.

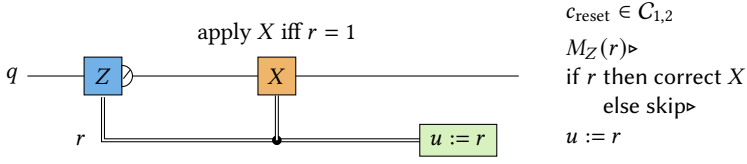
\begin{figure}[htbp]
\centering
\begin{tikzpicture}[font=\small,>=Latex]
 \draw (0,0) -- (6.5,0);
 \node[anchor=east] at (0,0) {$q$};
 \pic at (1,0) {ppmone={Z/}};
 \node[draw,pauli/X,minimum width=.5cm,minimum height=.5cm] (x) at (3.3,0) {$X$};
 \draw[double,double distance=.7pt] (1,-.28) -- (1,-.85) -- (5.7,-.85);
 \node[anchor=east] at (.9,-.85) {$r$};
 \fill (3.3,-.85) circle (1.8pt);
 \draw[double,double distance=.7pt] (3.3,-.85) -- (x.south);
 \node[draw,fill=ppmclassical,inner sep=4pt] at (5.7,-.85) {$u:=r$};
 \node at (3.3,.56) {apply $X$ iff $r=1$};
 \node[anchor=west,align=left] at (7,.05) {$c_{\mathrm{reset}}\in\mathcal C_{1,2}$\\[3pt]
   $M_Z(r)\seq $\\
   $\cmd{if}\ r\ \cmd{then}\ \cmd{correct}\ X$\\
   $\phantom{\cmd{if}\ r\ }\cmd{else}\ \cmd{skip}\seq $\\
   $u:=r$};
\end{tikzpicture}
\caption{Inputs: $\Gamma=(q)$, $\Delta=(r,u)$. Measure $Z$, correct outcome $1$,
and copy the result. The double wire carries $r$ to the correction and assignment.
When $r=0$, the data wire passes through unchanged.}
\label{fig:reset-syntax}
\Description{A Z measurement writes r. The data qubit then passes through an X correction
controlled by r, and a classical assignment copies r into u. The QPPM syntax is written
beside the circuit.}
\end{figure}

\subsection{Assertions and weakest preconditions}
\label{sec:assertions}

A classical store $s$ maps the $m$ bit positions to Boolean values. We use assertions
$A$ that map each store to a Hermitian matrix on the quantum register. For a quantum state
$\rho$, the value $\tr(\rho A(s))$ is its expectation.
An \emph{effect} satisfies $0\preceq A(s)\preceq I$ for every store, so this expectation
lies between zero and one on normalized states.
The mechanization admits arbitrary complex matrices, which is convenient for intermediate
calculations. Every assertion in this paper is Hermitian.

The classical locations of a program are divided into \emph{public} bits, which the
caller may read, and \emph{private} bits, which a component uses for its internal
measurement outcomes. Let $\pi$ map a store to its public part. A \emph{public
assertion} has the form $A\circ\pi$: it ignores the private bits. The contracts of
Section~\ref{sec:compilation} are stated for public assertions, so that a compiled
gate's internal outcomes are hidden from its caller, whereas the optimizer of
Section~\ref{sec:optimization} compares assertions on the full store because the
continuation may read any retained bit.

Two forms recur. A classical event $E$ gives the effect $[E]$, with $[E](s)=I$ if $E(s)$
holds and $0$ otherwise, whose expected value on a program output is the probability of
$E$. A quantum property that depends on a measurement result uses $A(s)=\Pi^P_{s(r)}$:
expectation one states that the output lies in the eigenspace recorded in $r$, and
multiplying by $[s(u)=s(r)]$ also requires two result bits to agree. Both forms appear
in the reset and phase-gate examples.

In the L\"owner order, $A\sqsubseteq B$ means that $B(s)-A(s)$ is
positive semidefinite for every store $s$. This order compares expectations on all positive
states. Effects are the assertions used for probability and approximation allowances.

We define $\WP(c,A)$ by reasoning backwards from the postcondition $A$ through $c$:
assignment substitutes the new bit value into $A$, measurement performs this substitution
for each outcome and applies the corresponding projectors, and an ancilla scope ignores
the temporary qubit in the postcondition and assumes its initial state in the precondition.

Write $s[r\mapsto b]$ for the store obtained by updating bit $r$ to $b$.
Let $E_b$ be the inclusion $|\psi\rangle\mapsto|b\rangle\otimes|\psi\rangle$. Write
$\mathsf{prep}_b^*(A)=E_b^\dagger A E_b$ and
$\mathsf{drop}^*(B)=\sum_{j=0}^1 E_j B E_j^\dagger=I_2\otimes B$,
acting pointwise on assertions. The weakest-precondition equations are:
\par\noindent
\begin{minipage}[t]{.48\linewidth}
\begin{equation}
 \WP(\cmd{skip},A)=A
 \tag{W-Skip}\label{wp:skip}
\end{equation}
\begin{equation}
 \begin{aligned}
  &\WP(r:=e,A)(s)\\
  &\quad=A(s[r\mapsto e(s)])
 \end{aligned}
 \tag{W-Asgn}\label{wp:assign}
\end{equation}
\begin{equation}
 \WP(\cmd{correct}\ P,A)(s)=P A(s)P
 \tag{W-Corr}\label{wp:correct}
\end{equation}
\end{minipage}\hfill
\begin{minipage}[t]{.48\linewidth}
\begin{equation}
 \WP(c\seq d,A)=\WP(c,\WP(d,A))
 \tag{W-Seq}\label{wp:seq}
\end{equation}
\begin{equation}
 \begin{aligned}
  &\WP(M_P(r),A)(s)\\
  &\quad=\sum_{j=0}^1\Pi^P_j A(s[r\mapsto j])\Pi^P_j
 \end{aligned}
 \tag{W-Meas}\label{wp:measure}
\end{equation}
\end{minipage}
\begin{equation}
 \WP(\cmd{if}\ e\ \cmd{then}\ c\ \cmd{else}\ d,A)(s)
  =\begin{cases}\WP(c,A)(s)&e(s)=1,\\\WP(d,A)(s)&e(s)=0,
   \end{cases}
 \tag{W-If}\label{wp:if}
\end{equation}
\begin{equation}
 \WP(\cmd{let}\ \nu=\cmd{ancilla}\ |b\rangle\ \cmd{in}\ c,A)
  =\mathsf{prep}_b^*(\WP(c,\mathsf{drop}^*(A))).
 \tag{W-Anc}\label{wp:ancilla}
\end{equation}

We prove that $\WP$ is linear and monotone, that $\WP(c,I)=I$, and that it maps effects
to effects. These properties justify \ref{rule:consequence} and the approximation
rules of Theorem~\ref{thm:approximation}.

\paragraph{Example: calculate the reset precondition.}
\label{sec:reset-wp}
For the program in Figure~\ref{fig:reset-syntax}, set
\[
 Q_0=\Pi^Z_0=\begin{pmatrix}1&0\\0&0\end{pmatrix},\qquad
 Q_1=\Pi^Z_1=\begin{pmatrix}0&0\\0&1\end{pmatrix},\qquad
 R(s)=Q_{s(r)}.
\]
The desired postcondition is $B_{\mathrm{reset}}(s)=[s(u)=s(r)]Q_0$:
the qubit is in $|0\rangle$ and the two records agree. We calculate backwards,
placing each command fragment beside its substitution. A schematic wire carries the
quantum register and classical store. A box denotes a command. A fork selects a branch.
In the expanded circuits, single wires are qubits and double wires carry bits.

\smallskip
\noindent\textit{1. Copy the record.}
Since $r\ne u$, Equation~\eqref{wp:assign} gives
\[
 \begin{array}{c@{\qquad}l}
  \pwire{u:=r}
  &\begin{aligned}
    \WP(u:=r,B_{\mathrm{reset}})(s)
       &=B_{\mathrm{reset}}(s[u\mapsto s(r)])\\
       &=[s(r)=s(r)]Q_0=Q_0.
   \end{aligned}
 \end{array}
\]
\noindent\textit{2. Correct the measured state.}
Write $g$ for the conditional in the figure. The two arms require different
input eigenspaces by \eqref{wp:if}, \eqref{wp:correct}, and \eqref{wp:skip}:
\[
 \begin{array}{c@{\qquad}l}
  \resetfeedback
  &\begin{aligned}
   \WP(g,Q_0)(s)
     &=\begin{cases}XQ_0X&s(r)=1,\\Q_0&s(r)=0\end{cases}\\
     &=\begin{cases}Q_1&s(r)=1,\\Q_0&s(r)=0\end{cases}
       =R(s).
   \end{aligned}
 \end{array}
\]
The correction equation $XQ_0X=Q_1$ can be checked entry by entry.

\noindent\textit{3. Measure and substitute each outcome.}
In \eqref{wp:measure}, the update inside $R$ selects the matching projector:
\[
 \begin{array}{c@{\qquad}l}
  \pmeter[Z]{Z}{r}
  &\begin{aligned}
   \WP(M_Z(r),R)(s)
     &=Q_0\underbrace{R(s[r\mapsto0])}_{Q_0}Q_0
       +Q_1\underbrace{R(s[r\mapsto1])}_{Q_1}Q_1\\
     &=Q_0^3+Q_1^3=Q_0+Q_1=I.
   \end{aligned}
 \end{array}
\]
Applying \eqref{wp:seq} twice yields
$\WP(c_{\mathrm{reset}},B_{\mathrm{reset}})=I$ for every initial store.
In particular, the postcondition has expectation one on every normalized input.
For a fresh qubit measured in $X$ and discarded, \eqref{wp:ancilla} and
\eqref{wp:measure} give
$\WP(\cmd{let}\ \nu=\cmd{ancilla}\ |0\rangle\ \cmd{in}\ M_{X_\nu}(r),[r=1])
=E_0^\dagger\Pi^X_1E_0=1/2$: the record is a fair bit.

\subsection{Forward semantics}
To describe execution, we assign a quantum output matrix to each final store. For input
store $s$, input matrix $\rho$, and final store $t$, write this matrix as
$\Inst_c(s,\rho,t)$. Together these matrices form the program's \emph{instrument}. On a
normalized input, the trace of each matrix is the probability of that final store. We prove
\[
 \rho\succeq0 \Longrightarrow \Inst_c(s,\rho,t)\succeq0,
 \qquad
 \sum_t\tr(\Inst_c(s,\rho,t))=\tr(\rho).
\]
The forward scope prepares $E_b\rho E_b^\dagger$, executes the body, and applies partial
trace separately to each final-store block, even when the ancilla is entangled with the
data. The following duality equation connects the execution semantics to weakest
preconditions:
\begin{equation}
 \begin{gathered}
 \begin{tikzpicture}[baseline=(current bounding box.center),font=\small,>=Latex]
  \node (input) at (0,0) {$(s,\rho)$};
  \node[draw,inner sep=4pt] (c) at (3.3,0) {$c$};
  \node (output) at (7.4,0) {$\{\Inst_c(s,\rho,t)\}_t$};
  \draw[->] (input.east) -- (c.west); \draw[->] (c.east) -- (output.west);
  \node (pre) at (0,-.72) {$\WP(c,A)(s)$};
  \node (post) at (7.4,-.72) {$\{A(t)\}_t$};
  \draw[->] (post.west) -- node[above,font=\scriptsize] {pull back the observation} (pre.east);
  \draw[dotted] (input.south) -- (pre.north) (output.south) -- (post.north);
 \end{tikzpicture}\\[4pt]
 \tr(\rho\,\WP(c,A)(s))
       =\sum_t\tr(\Inst_c(s,\rho,t)A(t)).
 \end{gathered}
 \label{eq:duality}
\end{equation}
Execution runs rightwards. Assertions are pulled back leftwards. The dotted pairings
evaluate the same expected value at the input and at the complete output distribution.
For a component with private bits, its \emph{public output} at public store $t$ sums the
blocks of all final stores with that public part:
\[
 \Inst^{\mathrm{public}}_c(s,\rho,t)=\sum_{u:\pi(u)=t}\Inst_c(s,\rho,u).
\]
By duality, the expectation of $A$ on this output equals that of $A\circ\pi$ on the full
output. Hiding bits is part of a public specification. An ancilla scope itself retains
all classical writes.

Duality is proved by a finite path construction in which every execution path
contributes one Kraus operator and one final store.

\paragraph{Example: execute the reset program.}
\label{sec:reset-forward}
Take the initial store $(r,u)=(0,0)$ and the state
$|+\rangle=(|0\rangle+|1\rangle)/\sqrt2$, with density matrix
$\rho_+=|+\rangle\langle+|=\frac12\left(\begin{smallmatrix}1&1\\1&1\end{smallmatrix}\right)$.
Each row below follows one outcome. The matrix at a boundary retains the probability
of reaching it. The pair below it is the current store $(r,u)$.
\[
\begin{array}{c@{\;}c@{\;}c@{\;}c@{\;}c@{\;}c@{\;}c}
 \rho_+&
 \pmeter[Z]{Z}{r=0}&
 \underbrace{Q_0/2}_{(0,0)}&
 \pwire{\cmd{skip}}&
 \underbrace{Q_0/2}_{(0,0)}&
 \pwire{u:=r}&
 \underbrace{Q_0/2}_{(0,0)}
 \\[12pt]
 \rho_+&
 \pmeter[Z]{Z}{r=1}&
 \underbrace{Q_1/2}_{(1,0)}&
 \pwire[X]{X}&
 \underbrace{Q_0/2}_{(1,0)}&
 \pwire{u:=r}&
 \underbrace{Q_0/2}_{(1,1)}
\end{array}
\]
The first boundary uses $Q_b\rho_+Q_b=Q_b/2$, the lower correction uses
$X(Q_1/2)X=Q_0/2$, and assignment changes only the store. The two nonzero final blocks,
$\Inst_{c_{\mathrm{reset}}}((0,0),\rho_+,(0,0))$ and
$\Inst_{c_{\mathrm{reset}}}((0,0),\rho_+,(1,1))$, both equal $Q_0/2$ and have trace
$1/2$, so the records remain distinguishable although their quantum outputs agree. For $B_{\mathrm{reset}}$, Equation~\eqref{eq:duality} reads
$\tr(\rho_+I)=1=\tr((Q_0/2)Q_0)+\tr((Q_0/2)Q_0)$, and for the event $[r=1]$ only the
second block contributes, giving probability $1/2$.

\section{Hoare logic}
\label{sec:hoare}

We define validity of a triple $\triple{A}{c}{B}$ by $A\sqsubseteq\WP(c,B)$. For
Hermitian assertions, duality means that executing $c$ establishes an expected value
of $B$ at least as large as the initial expected value of $A$.
We write $\vdash\triple{A}{c}{B}$ when the triple
is derivable using the following rules. Commands $c,d$, Boolean expression $e$, and assertions
$A,B,C,Q$ use the same declared contexts, except at a scope boundary.
Here $p$ ranges over assignment, measurement,
and correction, and $A\mathbin{\triangleright_e}B$ selects $A(s)$ when $e(s)=1$ and
$B(s)$ otherwise.

We draw the command inside its judgment:
\begin{equation}
 \htriple{A}{c}{B}\quad\text{denotes}\quad\triple{A}{c}{B},
 \qquad
 \pmeter{P}{r}\quad\text{denotes}\quad r:=\cmd{measure}\ P.
 \label{eq:diagram-notation}
\end{equation}
Joining command boxes means sequencing, and braces mark assertions at their boundaries.
Colored rounded boxes inside assertions denote the displayed matrices, without applying
an operation.
\begin{mathpar}
 \inferrule*[lab=\rulename{H-Skip}{rule:skip}]{}{\htriple{A}{\cmd{skip}}{A}}
 \and
 \inferrule*[lab=\rulename{H-Prim}{rule:primitive}]{}{ \htriple{\WP(p,B)}{p}{B}}
 \quad(p\text{ primitive})
 \and
 \inferrule*[lab=\rulename{H-Seq}{rule:seq}]{\htriple{A}{c}{B}\\\htriple{B}{d}{C}}
             {\dtriple{A}{\pwire{c}\!\pwire{d}}{C}}
 \and
 \inferrule*[lab=\rulename{H-If}{rule:if}]{\htriple{A}{c}{Q}\\\htriple{B}{d}{Q}}
   {\htriple{A\mathbin{\triangleright_e}B}
     {\cmd{if}\ e\ \cmd{then}\ c\ \cmd{else}\ d}{Q}}
 \and
 \inferrule*[lab=\rulename{H-Anc}{rule:ancilla}]{\htriple{A}{c}{\mathsf{drop}^*(B)}}
   {\htriple{\mathsf{prep}_b^*(A)}{\cmd{let}\ \nu=\cmd{ancilla}\ |b\rangle\ \cmd{in}\ c}{B}}
 \and
 \inferrule*[lab=\rulename{H-Cons}{rule:consequence}]{A\sqsubseteq A'\\\htriple{A'}{c}{B'}\\B'\sqsubseteq B}
   {\htriple{A}{c}{B}}
\end{mathpar}
In \ref{rule:ancilla}, the premise uses context $(\nu,\Gamma)$ and the conclusion uses
$\Gamma$. The body postcondition ignores the temporary qubit. Compressing the body
precondition by $\mathsf{prep}_b^*$ accounts for its initialization.

\paragraph{Example: derive the reset triple.}
\label{sec:reset-hoare}
We now use the rules to prove the specification whose precondition we calculated above.
Rules \ref{rule:primitive} and \ref{rule:skip} supply the premises of \ref{rule:if}:
\[
 \frac{\dtriple{Q_1}{\pwire[X]{X}}{Q_0}
       \qquad\htriple{Q_0}{\cmd{skip}}{Q_0}}
      {\dtriple{Q_1\mathbin{\triangleright_r}Q_0}{\resetfeedback}{Q_0}}.
\]
The correction premise uses $XQ_0X=Q_1$.
The guard reads $r$, so $Q_1\mathbin{\triangleright_r}Q_0$ selects
$Q_1$ on outcome one and $Q_0$ on outcome zero.
Rule \ref{rule:primitive}, instantiated with \eqref{wp:measure} and \eqref{wp:assign},
gives the other two triples:
\[
 \dtriple{I}{\pmeter[Z]{Z}{r}}{R}
 \qquad
 \htriple{Q_0}{u:=r}{B_{\mathrm{reset}}}.
\]
Two uses of \ref{rule:seq} join the equal assertion boundaries:
\begin{equation}
 \{I\}\;\pmeter[Z]{Z}{r}\;\{R\}\;
 \resetfeedback\;\{Q_0\}\;
 \pwire{u:=r}\;\{B_{\mathrm{reset}}\}.
 \label{eq:reset-proof}
\end{equation}
This is a derivation of $\vdash\triple{I}{c_{\mathrm{reset}}}{B_{\mathrm{reset}}}$.
To keep only the classical guarantee, put $E(s)\equiv s(u)=s(r)$.
At an agreeing store, $[E](s)-B_{\mathrm{reset}}(s)=I-Q_0=Q_1\succeq0$.
At a disagreeing store, the difference is zero. Hence $B_{\mathrm{reset}}\sqsubseteq[E]$,
and \ref{rule:consequence} gives
\[
 \frac{I\sqsubseteq I\qquad
       \htriple{I}{c_{\mathrm{reset}}}{B_{\mathrm{reset}}}
       \qquad B_{\mathrm{reset}}\sqsubseteq[E]}
      {\htriple{I}{c_{\mathrm{reset}}}{[E]}}.
\]
Thus the records agree with probability one.

The rules are a syntactic presentation of the weakest-precondition equations: each
primitive rule instantiates one equation, and the structural rules follow the
recursion of $\WP$. Since every QPPM command terminates, the following theorem is
therefore immediate, and partial and total correctness coincide.
\begin{theorem}[Soundness and relative completeness]
\label{thm:hoare}
For every QPPM command $c$ and assertions $A,B$,
$\vdash\triple{A}{c}{B}$ if and only if $A\sqsubseteq\WP(c,B)$.
\end{theorem}
\begin{proof}
Soundness is induction on the derivation, using monotonicity of $\WP$ and positivity
of the preparation map. For completeness, structural induction on $c$ derives
$\vdash\triple{\WP(c,B)}{c}{B}$, and \ref{rule:consequence} supplies any valid precondition:
\[
 \frac{A\sqsubseteq\WP(c,B)\qquad\htriple{\WP(c,B)}{c}{B}}
      {\htriple{A}{c}{B}}.
\]
\end{proof}
Completeness is relative to proofs of the matrix-order side conditions. The value of
the rules lies in composition: a derived triple for a component is reused at every
call site without recomputing its weakest precondition. In the soundness case for
\ref{rule:seq}, monotonicity and \eqref{wp:seq} give
$A\sqsubseteq\WP(c,B)\sqsubseteq\WP(c,\WP(d,C))$. For \ref{rule:ancilla}, positive
compression by $E_b$ preserves the matrix order, and \eqref{wp:ancilla} gives the
conclusion.

Define exact equivalence $c\equiv d$ by equality of their weakest preconditions for every
matrix assertion. The theorem immediately transports all Hoare specifications between
equivalent commands. The instrument development also proves equality of their output
matrices for each final store. We use this equivalence as the correctness condition for QPPM optimizations.

\subsection{Verification of the phase-gate application}
\label{sec:derive-application}

We now verify the application introduced in Section~\ref{sec:example}, with inputs
$\Gamma=(d)$ and $\Delta=(a,b,r,t)$. Locations $r,t$ are public. Locations $a,b$ are
private to the phase implementation $C_S$ of Figure~\ref{fig:gate-gallery}.
The application is
\[
 c_{\mathrm{app}}=M_X(r)\seq C_S\seq M_Y(t).
\]
We want the final bits to agree and the quantum state to lie in the corresponding $Y$
eigenspace. Define
\[
 A_X(s)=\qassert{X}{\Pi^X_{s(r)}},\qquad
 A_Y(s)=\qassert{Y}{\Pi^Y_{s(r)}},\qquad
 B(s)=[s(t)=s(r)]\,\qassert{Y}{\Pi^Y_{s(r)}}.
\]
Our specification is
\begin{equation}
 \triple{I}{c_{\mathrm{app}}}{B}.
 \label{eq:running}
\end{equation}
Figure~\ref{fig:application-proof} places each intermediate assertion below its position
in the circuit and the corresponding triples beside it.

\begin{figure}[htbp]
\centering
\begin{tikzpicture}[font=\small,>=Latex]
 \node[anchor=east] at (0,0) {$d$};
 \draw (.05,0) -- (6,0);
 \pic at (1.1,0) {ppmone={X/r}};
 \node[draw,fill=white,minimum width=.9cm,minimum height=.5cm] at (3,0) {$C_S$};
 \pic at (4.9,0) {ppmone={Y/t}};
 \foreach \x/\a in {.35/I,2/A_X,4/A_Y,5.7/B} {
   \draw[densely dotted] (\x,-.18) -- (\x,-.85);
   \node at (\x,-1.03) {$\{\a\}$};
 }
 \node[font=\scriptsize] at (1.1,.65) {measurement};
 \node[font=\scriptsize] at (3,.65) {compiled phase};
 \node[font=\scriptsize] at (4.9,.65) {measurement};
 \node[anchor=west,align=left,text width=5.3cm] at (6.5,-.12)
   {$r:=\cmd{measure}\ X\seq \ C_S\seq \ t:=\cmd{measure}\ Y$\\[7pt]
    $\{I\}\ M_X(r)\ \{A_X\}$\\[3pt]
    $S^\dagger\qassert{Y}{\Pi^Y_{s(r)}}S=\qassert{X}{\Pi^X_{s(r)}}$\\[3pt]
    $\{A_Y\}\ M_Y(t)\ \{B\}$};
 \node[anchor=west,font=\scriptsize] at (6.5,-1.5)
   {Two uses of \ref{rule:seq} give $\{I\}\,c_{\mathrm{app}}\,\{B\}$.};
\end{tikzpicture}
\caption{Inputs: $\Gamma=(d)$, $\Delta=(a,b,r,t)$. Assertions under the wire
mark the intermediate specifications. The result bits $r,t$ are public. The phase
component's private bits are hidden from these assertions.}
\Description{An X measurement, compiled S component, and Y measurement form a one-qubit
circuit. The assertions I, A X, A Y, and B mark successive points. Beside the circuit,
measurement triples and a colored projector identity justify the application triple.}
\label{fig:application-proof}
\end{figure}
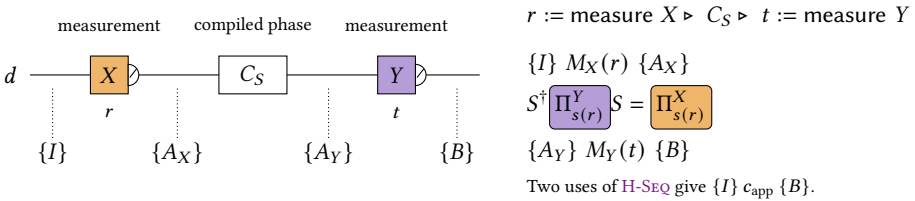

\paragraph{Step 1: establish $A_X$.}
Step 3 of the reset calculation, with $X$ in place of $Z$, establishes the eigenspace
named by $r$: for $A_P(s)=\Pi^P_{s(r)}$, \eqref{wp:measure} gives
$\WP(M_P(r),A_P)(s)=\sum_b\Pi^P_b\,\Pi^P_b\,\Pi^P_b=I$.

\paragraph{Step 2: transport the projector through the component.}
The component contract established in Section~\ref{sec:compilation} says that public
postconditions are pulled back by $S^\dagger(\cdot)S$. The relevant matrix calculation is
$SXS^\dagger=\operatorname{diag}(1,i)\,X\,\operatorname{diag}(1,-i)
=\left(\begin{smallmatrix}0&-i\\i&0\end{smallmatrix}\right)=Y$.
Conjugation distributes over $I\mathbin{\pm}X$, so
$S\Pi^X_bS^\dagger=\Pi^Y_b$. Unitarity then gives
$S^\dagger\Pi^Y_bS=\Pi^X_b$. As the component preserves public stores, the label $s(r)$
is unchanged. Its weakest precondition for $A_Y$ is exactly $A_X$, proving
$\triple{A_X}{C_S}{A_Y}$. Both assertions depend only on the public bit $r$ and the data qubit.

\paragraph{Step 3: establish agreement and the final eigenspace.}
A second measurement of the same product copies the recorded eigenspace: since $t\ne r$,
\eqref{wp:measure} gives
$\WP(M_Y(t),B)(s)=\sum_b\Pi^Y_b\,[b=s(r)]\,\Pi^Y_{s(r)}\,\Pi^Y_b=\Pi^Y_{s(r)}=A_Y(s)$.
The precondition is exactly $A_Y$, so there is no intervening matrix inequality to guess. Two uses of \ref{rule:seq}, shown in Figure~\ref{fig:application-proof}, prove
Equation~\eqref{eq:running}.

\paragraph{Step 4: extract an ordinary probability guarantee.}
Let $E(s)$ be $s(t)=s(r)$. Since $0\preceq\Pi^Y_{s(r)}\preceq I$, we have
$B\sqsubseteq[E]$, so \ref{rule:consequence} gives $\{I\}\,c_{\mathrm{app}}\,\{[E]\}$.
Duality on a positive trace-one input yields
$1\leq\sum_{u:E(u)}\tr(\Inst_{c_{\mathrm{app}}}(s,\rho,u))$, and total trace
preservation gives the reverse inequality, so agreement has probability one.

\paragraph{Step 5: reuse the proof after transformation.}
Exact optimization preserves this triple, and a certificate
$\mathcal{O}(c_{\mathrm{app}})\approx_\varepsilon d$ transfers the effect postcondition
$B$ to $d$ with precondition $(1-\varepsilon)I$, so $\Pr_d[E]\geq1-\varepsilon$.

\subsection{An incorrect phase-gate implementation}
\label{sec:incorrect}

Delete the phase component and call the resulting program $d=M_X(r)\seq M_Y(t)$.
For the event $E$, \eqref{wp:seq} and \eqref{wp:measure} give the exact precondition:
\[
 \frac{\WP(d,[E])=\sum_{b=0}^1\Pi^X_b\Pi^Y_b\Pi^X_b
              =\tfrac12\sum_b\Pi^X_b=\tfrac12 I}
      {\dtriple{\tfrac12 I}{\pmeter[X]{X}{r}\!\pmeter[Y]{Y}{t}}{[E]}}
 \qquad I\not\preceq\tfrac12 I.
\]
Every normalized input has agreement probability $1/2$. Soundness therefore rules out
the desired precondition $I$, regardless of the proof-search strategy.

The same event witnesses a necessary approximation allowance. If a certificate related the
optimized correct program to $d$ with allowance $\varepsilon$, specification transfer would
give $1-\varepsilon\leq1/2$, hence $\varepsilon\geq1/2$. The event proves only a lower
bound on the uniform error. A matching upper bound would require a proof for every effect.

\section{Verified compilation}
\label{sec:compilation}

Compilation replaces each source gate by a measurement program whose private outcomes
are hidden from the caller, who observes only the data register and the public bits.
Let $\pi$ drop the two private bits $a,b$. A component $c$ \emph{implements} a source
weakest-precondition transformer $W$ when
\begin{equation}
 \forall A.\quad \WP(c,A\circ\pi)=(W A)\circ\pi.
 \label{eq:component-contract}
\end{equation}
The equation quantifies over every initial value of the private bits, so sequential
components can reuse the same two bits. For a unitary component, $W(A)(s)=U^\dagger A(s)U$.
For a Clifford circuit with matrix $U$, composing the contracts of its gates gives the
public output $U\rho U^\dagger$ with the public store unchanged.

\subsection{Implementing a universal gate set}
\label{sec:gate-gallery}
Figure~\ref{fig:gate-gallery} gives measurement implementations of $H,S,\mathrm{CNOT}$.
Equation~\eqref{eq:injection-proof} gives $T$ injection.
Each diagram shows individual Pauli measurements and corrections beside its QPPM program.
Each row declares its input registers. The caller supplies $d$ (or $c,t$).
The distinct locations $a,b$ are private. Each binder introduces the fresh qubit $\nu$ shown
on its circuit wire and discards it after the correction. Sequential scopes reuse the
name for distinct qubits. Subscripts locate Pauli factors in the current context, and
all others are $I$. Scopes retain classical writes. The output specification hides $a,b$.

\begin{figure}[htbp]
\centering\small\setlength{\jot}{1pt}
\begin{tabular}{@{}m{.44\linewidth}@{\hspace{.025\linewidth}}m{.535\linewidth}@{}}
\multicolumn{2}{@{}l@{}}{\textbf{(a) $S$}\quad Inputs: $\Gamma=(d)$, $\Delta=(a,b)$}\\[3pt]
\begin{tikzpicture}[x=.75cm,y=.75cm,font=\small,>=Latex,
 baseline=(current bounding box.center)]
 \node[anchor=east] at (.45,0) {$|0\rangle_\nu$};
 \node[anchor=east] at (.45,-.75) {$d$};
 \draw (.5,0) -- (5.1,0); \draw (.5,-.75) -- (5.1,-.75);
 \pic at (1.3,0) {ppmtwo={X/Z/a}};
 \pic at (2.65,0) {ppmone={Y/b}};
 \node[draw,pauli/Z,minimum height=.4cm] (z) at (4.1,-.75) {$Z$};
 \node[fill=white,inner sep=1pt] (xor) at (4.1,-1.6) {$\oplus$};
 \draw[double,double distance=.6pt] (1.3,-1.4) |- (xor.west);
 \draw[double,double distance=.6pt] (2.65,-.68) -- (2.65,-1.6);
 \draw[double,double distance=.6pt,->] (xor.north) -- (z.south);
 \draw (5,-.15) -- (5,.15) (5.1,-.15) -- (5.1,.15);
\end{tikzpicture}
&
$\begin{aligned}
C_S\equiv{}&\cmd{let}\ \nu=\cmd{ancilla}\ |0\rangle\ \cmd{in}\ (M_{X_\nu Z_d}(a)\seq M_{Y_\nu}(b)\seq \\
 &\cmd{if}\ a\oplus b\ \cmd{then}\ \cmd{correct}\ Z_d\\
 &\cmd{else}\ \cmd{skip})
\end{aligned}$\\[8pt]
\multicolumn{2}{@{}l@{}}{\textbf{(b) $H$}\quad Inputs: $\Gamma=(d)$, $\Delta=(a,b)$}\\[3pt]
\begin{tikzpicture}[x=.53cm,y=.7cm,font=\scriptsize,>=Latex,
 baseline=(current bounding box.center)]
 \node[anchor=east] at (.15,-.75) {$d$};
 \draw (.2,-.75) -- (10.6,-.75);
 % Repeated scopes start further right so that each |0> label clears the discard bars.
 \foreach \off/\axis in {0/Z,3.8/X,7.6/Z} {
  \begin{scope}[xshift=\off*.53cm]
   \node[anchor=east] at (.55,0) {$|0\rangle_\nu$};
   \draw (.6,0) -- (2.96,0);
   \pic at (1,0) {ppmtwo={X/\axis/a}};
   \pic at (1.85,0) {ppmone={Y/b}};
   \node[draw,pauli/\axis,inner sep=2pt] (corr) at (2.6,-.75) {$\axis$};
   \node[fill=white,inner sep=.5pt] (xor) at (2.6,-1.65) {$\oplus$};
   \draw[double,double distance=.5pt] (1,-1.4) |- (xor.west);
   \draw[double,double distance=.5pt] (1.85,-.68) -- (1.85,-1.65);
   \draw[double,double distance=.5pt,->] (xor.north) -- (corr.south);
   \draw (2.88,-.13) -- (2.88,.13) (2.96,-.13) -- (2.96,.13);
  \end{scope}
 }
\end{tikzpicture}
&
$\begin{aligned}
C_H\equiv{}&\cmd{let}\ \nu=\cmd{ancilla}\ |0\rangle\ \cmd{in}\ (M_{X_\nu Z_d}(a)\seq M_{Y_\nu}(b)\seq \\
 &\cmd{if}\ a\oplus b\ \cmd{then}\ \cmd{correct}\ Z_d\ \cmd{else}\ \cmd{skip})\seq \\
 &\cmd{let}\ \nu=\cmd{ancilla}\ |0\rangle\ \cmd{in}\ (M_{X_\nu X_d}(a)\seq M_{Y_\nu}(b)\seq \\
 &\cmd{if}\ a\oplus b\ \cmd{then}\ \cmd{correct}\ X_d\ \cmd{else}\ \cmd{skip})\seq \\
 &\cmd{let}\ \nu=\cmd{ancilla}\ |0\rangle\ \cmd{in}\ (M_{X_\nu Z_d}(a)\seq M_{Y_\nu}(b)\seq \\
 &\cmd{if}\ a\oplus b\ \cmd{then}\ \cmd{correct}\ Z_d\ \cmd{else}\ \cmd{skip})
\end{aligned}$\\[8pt]
\multicolumn{2}{@{}l@{}}{\textbf{(c) CNOT}\quad Inputs: $\Gamma=(c,t)$, $\Delta=(a,b)$ (control $c$, target $t$)}\\[3pt]
\begin{tikzpicture}[x=.53cm,y=.7cm,font=\scriptsize,>=Latex,
 baseline=(current bounding box.center)]
 \node[anchor=east] at (.15,-.75) {$c$};
 \node[anchor=east] at (.15,-1.5) {$t$};
 \draw (.2,-.75) -- (10.6,-.75) (.2,-1.5) -- (10.6,-1.5);
 \foreach \off/\pc/\pt/\sgn in {0/Z/I/{},3.8/I/X/{},7.6/Z/X/{-}} {
  \begin{scope}[xshift=\off*.53cm]
   \node[anchor=east] at (.4,0) {$|0\rangle_\nu$};
   \draw (.45,0) -- (2.96,0);
   % Leave room for the signed Pauli label in the final measurement.
   \fill[ppmx] (.6,.25) rectangle (1.4,-.375);
   \fill[pauli/\pc] (.6,-.375) rectangle (1.4,-1.125);
   \fill[pauli/\pt] (.6,-1.125) rectangle (1.4,-1.75);
   \draw (.6,.25) rectangle (1.4,-1.75);
   \node at (1,0) {$\sgn X$};
   \node at (1,-.75) {$\pc$}; \node at (1,-1.5) {$\pt$};
   \draw[fill=white] (1.4,-.6) arc[start angle=90,end angle=-90,radius=.125] -- cycle;
   \draw (1.42,-.75) -- (1.5,-.66);
   \node at (1,-1.95) {$a$};
   \pic at (1.85,0) {ppmone={Y/b}};
   \node[fill=white,inner sep=.5pt] (xor) at (2.6,-2.35) {$\oplus$};
   \draw[double,double distance=.5pt] (1,-2.12) |- (xor.west);
   \draw[double,double distance=.5pt] (1.85,-.68) -- (1.85,-2.35);
   \ifdim\off pt=0pt
    \node[draw,pauli/Z,inner sep=2pt] (corr) at (2.6,-.75) {$Z$};
    \draw[double,double distance=.5pt,->] (xor.north) -- (corr.south);
   \else
    \node[draw,pauli/X,inner sep=2pt] (corr) at (2.6,-1.5) {$X$};
    \draw[double,double distance=.5pt,->] (xor.north) -- (corr.south);
    \ifdim\off pt=7.6pt
     \node[draw,pauli/Z,inner sep=2pt] (corrz) at (2.6,-.75) {$Z$};
     \draw (corrz.south) -- (corr.north);
    \fi
   \fi
   \draw (2.88,-.13) -- (2.88,.13) (2.96,-.13) -- (2.96,.13);
  \end{scope}
 }
\end{tikzpicture}
&
$\begin{aligned}
C_{\mathrm{CNOT}}\equiv{}&\cmd{let}\ \nu=\cmd{ancilla}\ |0\rangle\ \cmd{in}\ (M_{X_\nu Z_c}(a)\seq M_{Y_\nu}(b)\seq \\
 &\cmd{if}\ a\oplus b\ \cmd{then}\ \cmd{correct}\ Z_c\ \cmd{else}\ \cmd{skip})\seq \\
 &\cmd{let}\ \nu=\cmd{ancilla}\ |0\rangle\ \cmd{in}\ (M_{X_\nu X_t}(a)\seq M_{Y_\nu}(b)\seq \\
 &\cmd{if}\ a\oplus b\ \cmd{then}\ \cmd{correct}\ X_t\ \cmd{else}\ \cmd{skip})\seq \\
 &\cmd{let}\ \nu=\cmd{ancilla}\ |0\rangle\ \cmd{in}\ (M_{-X_\nu Z_cX_t}(a)\seq M_{Y_\nu}(b)\seq \\
 &\cmd{if}\ a\oplus b\ \cmd{then}\ \cmd{correct}\ Z_cX_t\ \cmd{else}\ \cmd{skip})
\end{aligned}$
\end{tabular}
\caption{Clifford gates using only Pauli measurements and conditional corrections.
Here $M_P(r)$ is $r:=\cmd{measure}\ P$. Every meter records the indicated bit.
An $\oplus$ node controls a correction when the two outcomes differ. Double end bars
discard an ancilla. Each following $|0\rangle$ begins a fresh scope. A gray $I$ factor
is drawn only inside a joint measurement column.
The output specification hides $a,b$.}
\Description{Explicit measurement sequences implement S, H, and CNOT. S measures
XZ and then Y on a zero ancilla, with a parity-controlled Z correction. H repeats
this with data factors Z, X, Z and three fresh ancillas. CNOT uses data factors
Z on control, X on target, then a negative joint XZX measurement, with Y readouts
and parity-controlled corrections.}
\label{fig:gate-gallery}
\end{figure}

Measuring $X_\nu P$ and then $Y_\nu$, and correcting $P$ when the outcomes differ,
implements the matrix $U(P)=\Pi^P_0+i\Pi^P_1$ after discarding $\nu$.
The identities $U(Z)=S$, $U(Z)U(X)U(Z)=e^{i\pi/4}H$, and
$U(-Z_cX_t)U(X_t)U(Z_c)=i\,\mathrm{CNOT}_{c,t}$ verify the Clifford sequences.
Global phases cancel in conjugation, including the sign of the final CNOT correction.
For $T$, supply $|M\rangle=(|0\rangle+e^{i\pi/4}|1\rangle)/\sqrt2$ on wire $\mu$.
The conditional branch performs the explicit $S$ sequence.

\begin{samepage}
\begin{theorem}[Clifford+$T$ implementation by PPM]
\label{thm:gate-compilation}
Every finite circuit over $H,S,\mathrm{CNOT},T$, with distinct control and target
wires for each CNOT, has a QPPM implementation with the same channel on its data
register. The implementation assumes one ideal $|M\rangle$ input per $T$ gate.
It uses basis-state ancillas, Pauli-product measurements, and conditional Pauli
corrections, and hides its private outcomes and consumed magic-state qubits.
\end{theorem}
\begin{proof}
Each diagram implements its gate after summing over its private outcomes.
Section~\ref{sec:derive-phase} proves the two-measurement contract, and
Section~\ref{sec:magic-resource} proves the injection contract. The displayed
matrix identities give $H$ and CNOT. Induction on the source circuit composes
these contracts. Each holds for every initial value of the private bits.
Successive implementations can therefore reuse them.
\end{proof}
\end{samepage}

With the standard universality of Clifford+$T$~\cite{boykin1999universal}, which we take
from the literature, this gives universal quantum computation with PPM and supplied
magic states.

\subsection{Clifford+T compilation with magic states}
\label{sec:magic-resource}

For data context $\Gamma$, let $P$ be a signed product on $\Gamma$.
Injection takes quantum inputs $(\mu,\Gamma)$, with $\mu\notin\Gamma$, and distinct private
locations $a,b$ in $\Delta$. The map $\pi$ removes $a,b$ from the store. Set
$\omega=e^{i\pi/4}$, $|M\rangle=(|0\rangle+\omega|1\rangle)/\sqrt2$, and
$T_P=\Pi^P_0+\omega\Pi^P_1$.
The isometry $V\psi=|M\rangle\otimes\psi$ supplies this magic state on $\mu$. Let $C_P$ be the
displayed QPPM body of the standard injection construction~\cite{litinski2019game}.
Inside the local scope, the context is $(\nu,\mu,\Gamma)$. Products on the data
act as identity on $\nu,\mu$. For a public assertion $B$, put
$D(B)(s)=I\otimes B(\pi(s))$. The circuit and its contract are
\begin{equation}
\begin{gathered}
\begin{gathered}
\begin{tikzpicture}[x=.55cm,y=.7cm,
 baseline=(current bounding box.center),font=\scriptsize,>=Latex]
 \node[anchor=east] at (.35,0) {$|M\rangle_\mu$};
 \node[anchor=east] at (.35,-1.5) {data};
 \draw (.4,0) -- (9.65,0) (.4,-1.5) -- (9.65,-1.5);
 % Joint Z_mu tensor P measurement; the temporary wire does not yet exist.
 \fill[ppmz] (.71,.25) rectangle (1.19,-.75);
 \fill[white] (.71,-.75) rectangle (1.19,-1.75);
 \draw (.71,.25) rectangle (1.19,-1.75);
 \node at (.95,0) {$Z$}; \node at (.95,-1.5) {$P$};
 \draw[fill=white] (1.19,-.625) arc[start angle=90,end angle=-90,radius=.125] -- cycle;
 \draw (1.21,-.775) -- (1.29,-.685);
 \node[font=\scriptsize] at (.95,-1.95) {$a$};
 % The entire scoped branch is selected before its measurements overwrite a,b.
 \draw[dashed,rounded corners=2pt] (2.05,-.3) rectangle (6.65,-2.25);
 \node[fill=white,inner sep=2pt,font=\scriptsize] at (4.35,-.3)
  {$a=1$};
 \draw[double,double distance=.7pt,->] (1.14,-1.95) -- (2.02,-1.95);
 \node[anchor=east] at (2.95,-.75) {$|0\rangle_\nu$};
 \draw (3,-.75) -- (6.35,-.75);
 \pic at (3.35,-.75) {ppmtwo={X/P/a}};
 \pic at (4.65,-.75) {ppmone={Y/b}};
 \node[draw,fill=white,inner sep=2pt] (innerp) at (5.9,-1.5) {$P$};
 \node[fill=white,inner sep=1pt] (xor) at (5.9,-2.12) {$\oplus$};
 \draw[double,double distance=.5pt] (3.35,-2.07) |- (xor.west);
 \draw[double,double distance=.5pt] (4.65,-1.42) -- (4.65,-2.12);
 \draw[double,double distance=.5pt,->] (xor.north) -- (innerp.south);
 \draw (6.25,-.9) -- (6.25,-.6) (6.35,-.9) -- (6.35,-.6);
 \pic at (7.6,0) {ppmone={X/b}};
 \node[draw,fill=white,inner sep=2pt] (p) at (8.95,-1.5) {$P$};
 \draw[double,double distance=.5pt,->] (7.6,-.67) -- (7.6,-2.08) -| (p.south);
 \node[fill=white,inner sep=1pt,font=\scriptsize] at (8.2,-2.08) {$b=1$};
\end{tikzpicture}
\end{gathered}
\quad
\text{\small\setlength{\jot}{1pt}$\begin{aligned}
 C_P\equiv{}&M_{Z_\mu P}(a)\seq \\
 &(\cmd{if}\ a\ \cmd{then}\ \cmd{let}\ \nu=\cmd{ancilla}\ |0\rangle\ \cmd{in}\\
 &\quad(M_{X_\nu P}(a)\seq M_{Y_\nu}(b)\seq \\
 &\quad\cmd{if}\ a\oplus b\ \cmd{then}\ \cmd{correct}\ P\ \cmd{else}\ \cmd{skip})\\
 &\quad\cmd{else}\ \cmd{skip})\seq \\
 &M_{X_\mu}(b)\seq \\
 &\cmd{if}\ b\ \cmd{then}\ \cmd{correct}\ P\ \cmd{else}\ \cmd{skip}
\end{aligned}$}\\[3pt]
 V^\dagger\WP(C_P,D(B))(s)V=T_P^\dagger B(\pi(s))T_P.
\end{gathered}
\label{eq:injection-proof}
\end{equation}
The dashed branch reads its guard before overwriting $a,b$. It is skipped when $a=0$.
Double end bars discard the fresh ancilla $\nu$.
Once the inner phase gadget is summed to the unitary $U(P)$ (its own eight branches are
in Section~\ref{sec:derive-phase}), the injection has eight effective Kraus branches,
indexed by the joint outcome $a$, the readout $b$, and the basis index $j$ of the consumed
qubit $\mu$, which $D(B)$ ignores. Each branch is a scalar multiple of $T_P$ with squared
magnitude $1/8$. Summing them proves the equation and the public output
$T_P\rho T_P^\dagger$. Negating $P$ gives the inverse action. Setting $P=Z$ gives $T$.
Induction extends the contract to a finite Clifford+$T$ circuit $c$ with public readouts.
If $c$ uses $k$ injections, supply distinct inputs $\mu_0,\ldots,\mu_{k-1}$ before its
data register, set $V_k\psi=|M\rangle^{\otimes k}\otimes\psi$ and
$D_k(B)(s)=I\otimes B(\pi(s))$, and let $W_c$ compose the source instruction
transformers. The compiler leaves each consumed qubit idle and keeps temporary basis
ancillas scoped. Writing $\mathcal L(c)$ for the compiled QPPM program, we prove, for
every assertion $B$ and store $s$,
\begin{equation}
 V_k^\dagger\,\WP(\mathcal L(c),D_k(B))(s)\,V_k=W_c(B)(\pi(s)).
 \label{eq:circuit-contract}
\end{equation}
Hoare transfer is one application of Theorem~\ref{thm:hoare}: if $A\sqsubseteq W_c(B)$,
then $Q=\WP(\mathcal L(c),D_k(B))$ satisfies $\vdash\triple{Q}{\mathcal L(c)}{D_k(B)}$
and $A(\pi(s))\preceq V_k^\dagger Q(s)V_k$ for every $s$.
Replacing $|M\rangle$ by $|0\rangle$ changes the transformer to $B\mapsto\sum_a\Pi^P_a B\Pi^P_a$.
Thus the magic-state premise is necessary.

\subsection{Correctness of the phase measurement sequence}
\label{sec:derive-phase}

Fix data inputs $\Gamma$, private locations $a\ne b$, and a signed product $P$ on
$\Gamma$. The sequence of Figure~\ref{fig:gate-gallery}(a), with $P$ in place of $Z_d$,
binds a fresh zero ancilla $\nu$, measures $X\otimes P$ into $a$ and $Y\otimes I$ into
$b$, applies $I\otimes P$ when the outcomes differ, and discards $\nu$. Let
$x,y\in\{0,1\}$ be the recorded values, $j\in\{0,1\}$ the partial-trace index, and
$C_{xy}=I\otimes P^{x\oplus y}$. The data-register Kraus operator of one branch is
\begin{equation}
 \begin{gathered}
 \pwire{E_0}\!\pwire{\Pi^{X\otimes P}_x}\!\pwire{\Pi^{Y\otimes I}_y}\!\pwire{C_{xy}}\!\pwire{E_j^\dagger}\\[3pt]
 K_{xyj}=E_j^\dagger C_{xy}\Pi^{Y\otimes I}_y\Pi^{X\otimes P}_x E_0.
 \end{gathered}
 \label{eq:phase-kraus}
\end{equation}
Boxes are linear maps composed left to right. The matrix product runs right to left.
Set $e=x\oplus y$. Expanding the two projectors with the $X$ and $Y$ matrices and using
$P^2=I$ gives
\[
 K_{xy0}=\tfrac14 P^e\bigl(I-i(-1)^eP\bigr)=\tfrac{1-i(-1)^e}{4}\,U(P),\qquad
 K_{xy1}=i(-1)^y K_{xy0}.
\]
Thus $K_{xyj}=\alpha_{xyj}U(P)$ with $\alpha_{xyj}=[i(-1)^y]^j(1-i(-1)^{x\oplus y})/4$.
Writing $u=(1-i)/4$ and $v=(1+i)/4$, the eight coefficients are
\[
\begin{array}{cc|c|cc}
 x & y & \text{correction} & \alpha_{xy0} & \alpha_{xy1}\\\hline
 0&0&I&u&iu\\ 0&1&P&v&-iv\\ 1&0&P&v&iv\\ 1&1&I&u&-iu
\end{array}
\]
and every entry has squared magnitude $1/8$. Summing the branches therefore gives
\[
 \sum_{x,y,j}K_{xyj}^\dagger B K_{xyj}=U(P)^\dagger B U(P),\qquad
 \dtriple{S^\dagger B S}{\pwire{C_S}}{B}\quad(P=Z).
\]
Here $B$ ignores the two private bits. Both are overwritten before the guard reads them,
and public bits are preserved, so the contract holds for every initial value of the
private bits. The proof retains the off-diagonal blocks of joint assertions, so the
contract also holds when the data qubit is entangled with an idle qubit of the caller.

\section{Verified optimization}
\label{sec:optimization}

The QPPM optimizer, denoted $\mathcal{O}$, combines measurement-run optimization, adjacent
repetition and empty-scope cleanup, and lowering to unsigned products. A negative measurement
becomes an unsigned measurement followed by a flip of its result bit. The sign of a correction
can be removed because its two conjugating factors cancel the global sign.

The nonlocal rewrite follows a pending measurement through a sequence of measurements.
It can cross a measurement of $Q$ when $P$ and $Q$ commute and the intervening destination
does not overwrite the pending result. A later measurement of the same signed $P$ can then
copy the saved result. The commutation test is executable and its matrix interpretation is
proved. The intervening measurements need not commute with each other.

In Figure~\ref{fig:measurement-rewrite}, with inputs $\Gamma=(q_0,q_1)$ and
$\Delta=(r_0,r_1,r_2,r_3)$, the first measurement of $Z\otimes I$ moves past
$I\otimes X$ and $I\otimes Z$, which anticommute with each other but commute with
$Z\otimes I$. The two identical measurements are then adjacent, so the last is replaced
by a copy. The result has three measurements and retains all four classical outcomes.

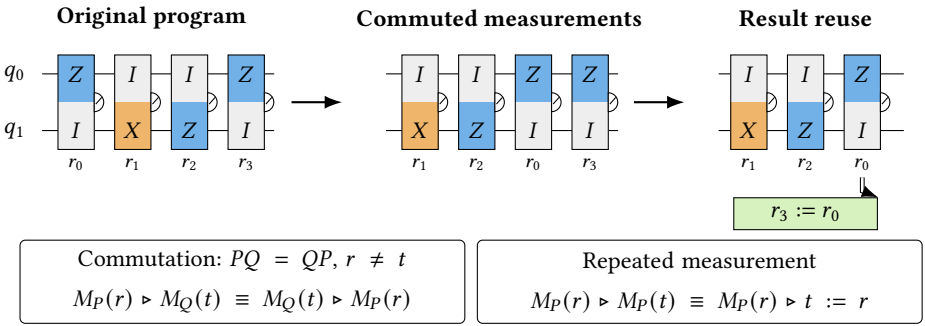
\begin{figure}[htbp]
\centering
\begin{tikzpicture}[font=\small,>=Latex]
 \node[font=\small\bfseries] at (1.7,.75) {Original program};
 \node[font=\small\bfseries] at (6.3,.75) {Commuted measurements};
 \node[font=\small\bfseries] at (10.35,.75) {Result reuse};
 \node[anchor=east] at (.15,0) {$q_0$};
 \node[anchor=east] at (.15,-.75) {$q_1$};
 \draw (.25,0) -- (3.32,0);
 \draw (.25,-.75) -- (3.32,-.75);
 \pic at (.7,0) {ppmtwo={Z/I/r_0}};
 \pic at (1.45,0) {ppmtwo={I/X/r_1}};
 \pic at (2.2,0) {ppmtwo={I/Z/r_2}};
 \pic at (2.95,0) {ppmtwo={Z/I/r_3}};
 \draw[->,thick] (3.55,-.375) -- (4.22,-.375);
 \begin{scope}[xshift=4.6cm]
   \draw (.2,0) -- (3.27,0);
   \draw (.2,-.75) -- (3.27,-.75);
   \pic at (.65,0) {ppmtwo={I/X/r_1}};
   \pic at (1.4,0) {ppmtwo={I/Z/r_2}};
   \pic at (2.15,0) {ppmtwo={Z/I/r_0}};
   \pic at (2.9,0) {ppmtwo={Z/I/r_3}};
 \end{scope}
 \draw[->,thick] (8.08,-.375) -- (8.75,-.375);
 \begin{scope}[xshift=9cm]
   \draw (.2,0) -- (2.65,0);
   \draw (.2,-.75) -- (2.65,-.75);
   \pic at (.6,0) {ppmtwo={I/X/r_1}};
   \pic at (1.35,0) {ppmtwo={I/Z/r_2}};
   \pic at (2.1,0) {ppmtwo={Z/I/r_0}};
   \node[draw,fill=ppmclassical,minimum width=1.9cm,minimum height=.4cm]
     (copy) at (1.35,-1.85) {$r_3:=r_0$};
   \draw[double,double distance=1pt,->] (2.1,-1.37) -- (2.1,-1.55) -- (copy.north east);
 \end{scope}
 \node[draw,rounded corners=2pt,align=center,text width=5.65cm,minimum height=1.1cm]
   at (2.9,-2.75)
   {Commutation: $PQ=QP$, $r\ne t$\\[4pt]
    $M_P(r)\seq M_Q(t)\equiv M_Q(t)\seq M_P(r)$};
 \node[draw,rounded corners=2pt,align=center,text width=5.65cm,minimum height=1.1cm]
   at (8.95,-2.75)
   {Repeated measurement\\[4pt]
    $M_P(r)\seq M_P(t)\equiv M_P(r)\seq t:=r$};
\end{tikzpicture}
\caption{Measurement elimination in QPPM. Each outlined column is one measurement and the
label below it names the destination bit. All four destinations are distinct. The first
$Z\otimes I$ measurement moves past $I\otimes X$ and $I\otimes Z$, then its result replaces
the final measurement. The green box is a classical assignment. The two rules preserve
both the quantum output and every final bit.}
\Description{Three two-qubit circuits show four Pauli-product measurements, the same
measurements after commuting the first past the middle two, and three measurements followed
by a classical copy. Boxes below state commutation for commuting products with distinct
result destinations and the repeated-measurement rule.}
\label{fig:measurement-rewrite}
\end{figure}

With inputs $\Gamma=(q)$ and $\Delta=(r,t)$, the following attempted rewrite copies a result
that the middle measurement has overwritten:
\[
 \underbrace{\pmeter[I]{I}{r}\!\pmeter[I]{-I}{r}\!\pmeter[I]{I}{t}}_{r=1,\ t=0}
 \quad\not\equiv\quad
 \underbrace{\pmeter[I]{I}{r}\!\pmeter[I]{-I}{r}\!\pwire{t:=r}}_{r=1,\ t=1}.
\]
The scanner stops at such overwrites and at noncommuting products.
Assignments, corrections, conditionals, and
scopes are scan barriers. The optimizer recursively visits conditional arms and scope bodies.

\begin{theorem}[Exact optimization]
\label{thm:optimizer}
For every $c$, $\mathcal{O}(c)\equiv c$. The result uses only unsigned products and
has no more syntactic measurement nodes than $c$. It has the same forward output matrix
for each input store, input matrix, and final store, and the same derivable Hoare triples.
\end{theorem}
\begin{proof}
The rewrites in Figure~\ref{fig:measurement-rewrite} show that scanning preserves the
pending measurement followed by its continuation, $\mathsf{scan}(P,r,d)\equiv M_P(r)\seq d$,
using projector commutation, record substitution, and repetition elimination. Sequence
reassociation and structural induction lift the invariant to whole programs, and
composition with the cleanup and sign-lowering passes proves the result. Count
inequalities are proved separately by induction. Duality gives equality of every output
block. Theorem~\ref{thm:hoare} gives equivalence of the derivable Hoare specifications.
\end{proof}
The count adds both arms of a conditional.

\subsection{Rotation fusion and scheduling}
\label{sec:rotation-schedule}

For the phase unitary $U(P)=(I+P)/2+i(I-P)/2$, we prove $U(P)^2=P$ and $U(P)^4=I$, so
a quarter-turn count $k\in\{0,1,2,3\}$ expands to skip, the two-measurement gadget, a
Pauli correction, or the inverse gadget. We optimize these source rotations before
expanding them to QPPM programs. Their component specifications hold for every signed
product and register width.
With data inputs $\Gamma=(q_0,q_1)$, put $P=Z\otimes I$ and $Q=I\otimes Z$. The
following diagram places each source unitary on its data wire. Commutation moves the
first box. Adding turns modulo four removes the upper pair.
\begin{equation}
\begin{tikzpicture}[baseline=(current bounding box.center),font=\small,>=Latex]
 \foreach \yy/\lab in {0/{\text{before}},-1.12/{\text{swap}},-2.24/{\text{fuse}}} {
  \node[anchor=east] at (-.4,\yy-.24) {$\lab$};
  \draw (0,\yy) -- (5,\yy);
  \draw (0,\yy-.48) -- (5,\yy-.48);
 }
 \node[draw,pauli/Z] at (.8,0) {$U(Z)$};
 \node[draw,pauli/Z] at (2.4,-.48) {$U(Z)$};
 \node[draw,pauli/Z] at (4,0) {$U(Z)^3$};
 \node[draw,pauli/Z] at (.8,-1.6) {$U(Z)$};
 \node[draw,pauli/Z] at (2.4,-1.12) {$U(Z)$};
 \node[draw,pauli/Z] at (4,-1.12) {$U(Z)^3$};
 \node[draw,pauli/Z] at (.8,-2.72) {$U(Z)$};
 \node[anchor=west,align=left] at (5.5,-.24)
  {$P=Z\otimes I,\ Q=I\otimes Z$\\$U(P)\seq U(Q)\seq U(P)^3$};
 \node[anchor=west,align=left] at (5.5,-1.36)
  {$PQ=QP$\\$U(Q)\seq U(P)\seq U(P)^3$};
 \node[anchor=west,align=left] at (5.5,-2.48)
  {$U(P)^4=I$\\$U(Q)$};
 \draw[->] (4.8,-.64) -- (4.8,-.89);
 \draw[->] (4.8,-1.76) -- (4.8,-2.01);
\end{tikzpicture}
\label{eq:rotation-proof}
\end{equation}
With two private ancillas, the parallel baseline groups the first two disjoint rotations,
costing four logical rounds overall. The result costs two. Lean checks
this derivation and both costs.

The scheduling checker replays each proposed edit after verifying its commutation and
classical-dependency conditions, then compares the result with the proposed schedule.
Induction proves equality of public weakest preconditions, and composing gadget
contracts transfers every public Hoare specification and output block. Private
measurement results may differ. Rotation motion stops at readouts, and swapping readouts
requires distinct destinations.

For the two disjoint gadgets, allocation binds fresh $\nu_0,\nu_1$ and private locations
$\Delta=(a_0,b_0,a_1,b_1)$. The labels $J_i$, $R_i$, and $F_i$ denote the joint measurement into
$a_i$, the ancilla readout into $b_i$, and the conditional correction of gadget $i$.
Each wire groups a gadget's data, ancilla, and private bits. Dashed columns mark
simultaneous logical rounds after three guarded exchanges:
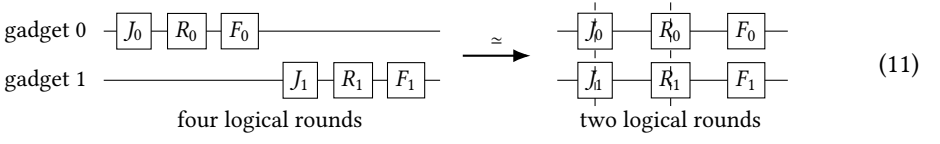
\begin{equation}
\begin{tikzpicture}[baseline=(current bounding box.center),font=\small,>=Latex]
 \foreach \y/\i in {0/0,-.65/1} {
  \node[anchor=east] at (0,\y) {$\text{gadget }\i$};
  \draw (.1,\y) -- (4.55,\y);
  \draw (6.1,\y) -- (9.15,\y);
 }
 \foreach \x/\y/\t in {.5/0/J_0,1.2/0/R_0,1.9/0/F_0,2.7/-.65/J_1,3.4/-.65/R_1,4.1/-.65/F_1,
  6.6/0/J_0,7.6/0/R_0,8.6/0/F_0,6.6/-.65/J_1,7.6/-.65/R_1,8.6/-.65/F_1} {
  \node[draw,fill=white,inner sep=3pt] at (\x,\y) {$\t$};
 }
 \draw[->,thick] (4.85,-.325) -- node[above,font=\scriptsize] {$\simeq$} (5.75,-.325);
 \draw[dashed] (6.6,.38) -- (6.6,-1.03) (7.6,.38) -- (7.6,-1.03);
 \node at (2.3,-1.2) {four logical rounds};
 \node at (7.6,-1.2) {two logical rounds};
\end{tikzpicture}
\label{eq:schedule-proof}
\end{equation}
$J_1$ moves left across $F_0$ and $R_0$. Then $R_1$ moves left across $F_0$.
Quantum commutation and independent result bits justify each exchange. The checker verifies disjoint supports,
distinct private bits, and the ancilla budget, and we prove that the allocated program
has the macros' public weakest precondition. We charge two logical rounds per odd turn,
one per readout, and zero for corrections. Routing and hardware timing remain outside
this model.

\subsection{Compilation of a complete circuit}
\label{sec:two-qubit-compilation}

The source has inputs $\Gamma=(q_0,q_1)$ and $\Delta=()$.
Figure~\ref{fig:two-qubit-compilation} applies $S$ to each qubit, CNOT from $q_0$ to $q_1$,
$R_\theta=\operatorname{diag}(1,e^{i\theta})$ to $q_0$, and $X$ to $q_1$.
Here $\theta=\pi/2+\delta$ and $\delta=\pi/2^{26}$. We round $R_\theta$ to $S$,
then compile and optimize. Let $U_\theta$ be its circuit matrix and
$W_\theta(B)=U_\theta^\dagger B U_\theta$ the
source transformer and $W_0$ the rounded circuit's transformer.

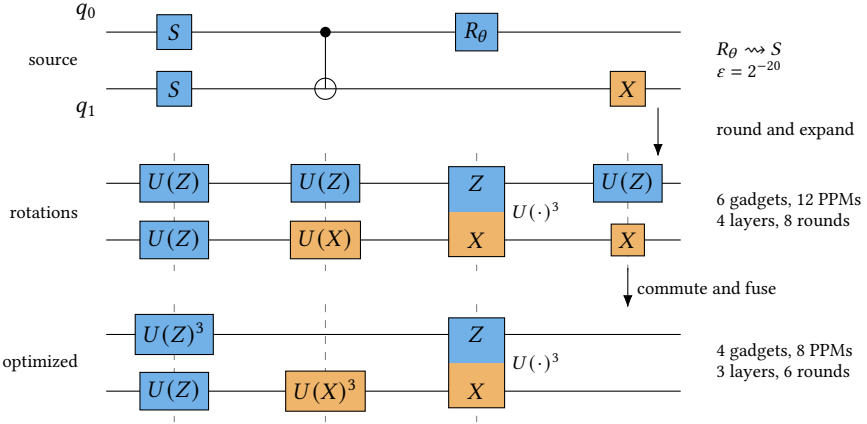
\begin{figure}[htbp]
\centering
\begin{tikzpicture}[font=\small,>=Latex]
 \foreach \yy/\lab in {0/source,-2.0/rotations,-4.0/optimized} {
   \node[anchor=east,font=\scriptsize] at (-.15,\yy-.375) {\lab};
   \draw (.1,\yy) -- (7.7,\yy) (.1,\yy-.75) -- (7.7,\yy-.75);
 }
 \node[anchor=east] at (.1,.27) {$q_0$};
 \node[anchor=east] at (.1,-1.03) {$q_1$};
 \foreach \yy in {0,-.75} {
   \node[draw,pauli/Z,minimum size=.46cm] at (1,\yy) {$S$};
 }
 \fill (3,0) circle (2pt);
 \draw (3,0) -- (3,-.75);
 \draw (3,-.75) circle (.14);
 \node[draw,pauli/Z,minimum height=.46cm] at (5,0) {$R_\theta$};
 \node[draw,pauli/X,minimum size=.46cm] at (7,-.75) {$X$};
 \node[anchor=west,align=left,font=\scriptsize] at (8.05,-.37)
   {$R_\theta\rightsquigarrow S$\\$\varepsilon=2^{-20}$};
 \draw[->] (7.4,-1.0) -- (7.4,-1.65);
 \node[anchor=west,font=\scriptsize] at (8.05,-1.3) {round and expand};
 % The baseline uses four two-round layers.
 \foreach \xx in {1,3,5,7} {
   \draw[dashed,gray] (\xx,-1.6) -- (\xx,-3.16);
 }
 \foreach \xx in {1,3,5} {
   \draw[dashed,gray] (\xx,-3.6) -- (\xx,-5.16);
 }
 \foreach \xx/\yy/\axis/\txt in {
   1/-2.0/Z/{U(Z)},1/-2.75/Z/{U(Z)},
   3/-2.0/Z/{U(Z)},3/-2.75/X/{U(X)},
   7/-2.0/Z/{U(Z)},7/-2.75/X/{X},
   1/-4.0/Z/{U(Z)^3},1/-4.75/Z/{U(Z)},
   3/-4.75/X/{U(X)^3}} {
   \node[draw,pauli/\axis,inner sep=3pt] at (\xx,\yy) {$\txt$};
 }
 \foreach \yy in {-2,-4} {
   \fill[ppmz] (4.63,\yy+.22) rectangle (5.37,\yy-.375);
   \fill[ppmx] (4.63,\yy-.375) rectangle (5.37,\yy-.97);
   \draw (4.63,\yy+.22) rectangle (5.37,\yy-.97);
   \node at (5,\yy) {$Z$};
   \node at (5,\yy-.75) {$X$};
   \node[anchor=west,font=\scriptsize,fill=white,inner sep=1pt] at (5.43,\yy-.375)
     {$U(\cdot)^3$};
 }
 \node[anchor=west,align=left,font=\scriptsize] at (8.05,-2.37)
   {6 gadgets, 12 PPMs\\4 layers, 8 rounds};
 \node[anchor=west,align=left,font=\scriptsize] at (8.05,-4.37)
   {4 gadgets, 8 PPMs\\3 layers, 6 rounds};
 \draw[->] (7,-3.12) -- node[right,font=\scriptsize] {commute and fuse} (7,-3.65);
\end{tikzpicture}
\caption{Source unitaries, rotations after certified rounding, and optimization. A split
blue--orange box denotes $U(Z\otimes X)^3$. Each dashed column costs two logical
rounds. The unmarked $X$ correction costs zero. The baseline already runs disjoint
rotations in parallel. Figure~\ref{fig:two-qubit-ppm} shows the individual Pauli measurements.}
\Description{Two data wires carry S gates, a control-zero target-one CNOT, a phase rotation
on qubit zero, and X on qubit one. Rounding the phase to S gives six odd Pauli turns.
Commuting and fusing gives an inverse Z turn on qubit zero parallel to an S on qubit one,
then an inverse X turn on qubit one, then an inverse ZX turn on both wires.}
\label{fig:two-qubit-compilation}
\end{figure}

\paragraph{Step 1: certify the approximation.}
The bound $|e^{i\theta}-i|\leq\delta<2^{-24}$ gives the conservative allowance
$\varepsilon=16\cdot2^{-24}=2^{-20}<9.54\times10^{-7}$: expanding
$(S\otimes I)^\dagger B(S\otimes I)-(R_\theta\otimes I)^\dagger B(R_\theta\otimes I)$ in
matrix units gives sixteen terms, each of operator norm at most $2^{-24}$ for an effect
$B$. All other gates are exact, so
$-\varepsilon I\sqsubseteq W_\theta(B)-W_0(B)\sqsubseteq\varepsilon I$
for every two-qubit effect $B$, including observations on entangled inputs.

\paragraph{Step 2: expand the gates into measurement gadgets.}
Put $P=Z\otimes I$, $Q=I\otimes Z$, and $R=I\otimes X$, so $PR=Z\otimes X$.
The S gates become $U(P)$ or $U(Q)$, CNOT becomes
$U(P)\seq U(R)\seq U(PR)^3$, and X becomes the correction $R$.
Each odd turn uses two measurements and feedback through the phase gadget.
The expanded circuit therefore has twelve measurements.

\paragraph{Step 3: commute, fuse, and allocate.}
Since $P$ commutes with $Q,R,PR$, move its three turns together and fuse them into
$U(P)^3$. Since $R$ commutes with $PR$, move the final correction next to $U(R)$
and fuse them into $U(R)^3$. This gives
\[
 \underbrace{U(Q)\seq U(P)^3}_{\text{disjoint supports}}\seq U(R)^3\seq U(PR)^3.
\]
Six gadgets become four. The first two run in parallel. The others share $q_1$
and run in order (Figure~\ref{fig:two-qubit-ppm}), with at most two
ancillas and four private bits. Against the parallel baseline, measurements fall
from twelve to eight and depth from eight to six logical rounds, a 25\% reduction.
Pauli feedback costs zero logical rounds.

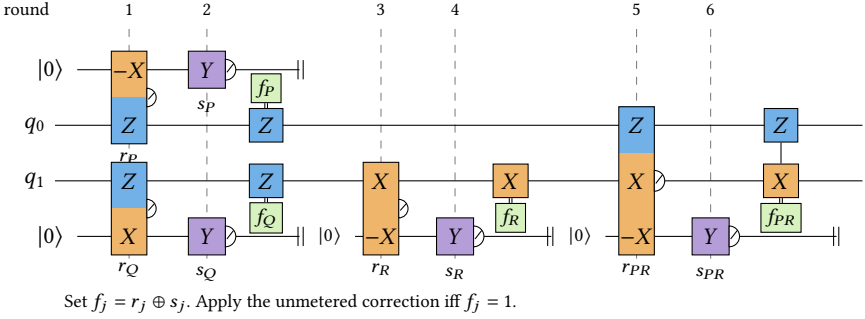
\begin{figure}[htbp]
\centering
\begin{tikzpicture}[font=\small,>=Latex,x=.98cm,y=.98cm]
 % Dashed columns are simultaneous primitive measurement rounds.
 \foreach \xx/\rr in {1.65/1,2.7/2,5.05/3,6.05/4,8.5/5,9.5/6} {
   \draw[dashed,gray] (\xx,.55) -- (\xx,-2.55);
   \node[font=\scriptsize] at (\xx,.82) {\rr};
 }
 \node[anchor=east,font=\scriptsize] at (.7,.82) {round};
 \draw (.65,-.75) -- (11.55,-.75) (.65,-1.5) -- (11.55,-1.5);
 \node[anchor=east] at (.65,-.75) {$q_0$};
 \node[anchor=east] at (.65,-1.5) {$q_1$};
 \node[anchor=east] at (.9,0) {$|0\rangle$};
 \node[anchor=east] at (.9,-2.25) {$|0\rangle$};
 \draw (.95,0) -- (3.95,0) (.95,-2.25) -- (3.95,-2.25);
 % U(P)^3 on the upper pair; U(Q) on the lower pair.
 \pic at (1.65,0) {ppmtwo={X/Z/r_P}};
 \node[fill=ppmx,inner sep=0pt] at (1.65,0) {$-X$};
 \pic at (1.65,-1.5) {ppmtwo={Z/X/r_Q}};
 \pic at (2.7,0) {ppmone={Y/s_P}};
 \pic at (2.7,-2.25) {ppmone={Y/s_Q}};
 \node[draw,pauli/Z,minimum size=.43cm] (cp) at (3.5,-.75) {$Z$};
 \node[draw,pauli/Z,minimum size=.43cm] (cq) at (3.5,-1.5) {$Z$};
 \node[draw,fill=ppmclassical,inner sep=2pt,font=\scriptsize] (fp) at (3.5,-.25) {$f_P$};
 \node[draw,fill=ppmclassical,inner sep=2pt,font=\scriptsize] (fq) at (3.5,-2) {$f_Q$};
 \draw[double,double distance=.6pt] (fp.south) -- (cp.north);
 \draw[double,double distance=.6pt] (fq.north) -- (cq.south);
 \foreach \yy in {0,-2.25} {
   \draw (3.93,\yy-.13) -- (3.93,\yy+.13) (4,\yy-.13) -- (4,\yy+.13);
 }
 % Reinitialize the lower ancilla for U(R)^3.
 \node[anchor=east,font=\scriptsize] at (4.65,-2.25) {$|0\rangle$};
 \draw (4.7,-2.25) -- (7.35,-2.25);
 \pic at (5.05,-1.5) {ppmtwo={X/X/r_R}};
 \node[fill=ppmx,inner sep=0pt] at (5.05,-2.25) {$-X$};
 \pic at (6.05,-2.25) {ppmone={Y/s_R}};
 \node[draw,pauli/X,minimum size=.43cm] (cr) at (6.8,-1.5) {$X$};
 \node[draw,fill=ppmclassical,inner sep=2pt,font=\scriptsize] (fr) at (6.8,-2) {$f_R$};
 \draw[double,double distance=.6pt] (fr.north) -- (cr.south);
 \draw (7.3,-2.38) -- (7.3,-2.12) (7.37,-2.38) -- (7.37,-2.12);
 % Final three-body measurement for U(PR)^3.
 \node[anchor=east,font=\scriptsize] at (8.03,-2.25) {$|0\rangle$};
 \draw (8.08,-2.25) -- (11.2,-2.25);
 \fill[ppmz] (8.26,-.5) rectangle (8.74,-1.125);
 \fill[ppmx] (8.26,-1.125) rectangle (8.74,-2.5);
 \draw (8.26,-.5) rectangle (8.74,-2.5);
 \node at (8.5,-.75) {$Z$};
 \node at (8.5,-1.5) {$X$};
 \node at (8.5,-2.25) {$-X$};
 \draw[fill=white] (8.74,-1.36) arc[start angle=90,end angle=-90,radius=.14] -- cycle;
 \draw (8.76,-1.535) -- (8.85,-1.42);
 \node[font=\scriptsize] at (8.5,-2.7) {$r_{PR}$};
 \pic at (9.5,-2.25) {ppmone={Y/s_{PR}}};
 \node[draw,pauli/Z,minimum size=.43cm] (czx) at (10.45,-.75) {$Z$};
 \node[draw,pauli/X,minimum size=.43cm] (cxx) at (10.45,-1.5) {$X$};
 \draw (czx.south) -- (cxx.north);
 \node[draw,fill=ppmclassical,inner sep=2pt,font=\scriptsize] (fzx) at (10.45,-2) {$f_{PR}$};
 \draw[double,double distance=.6pt] (fzx.north) -- (cxx.south);
 \draw (11.16,-2.38) -- (11.16,-2.12) (11.23,-2.38) -- (11.23,-2.12);
 \node[anchor=west,font=\scriptsize] at (.65,-3.15)
   {Set $f_j=r_j\oplus s_j$. Apply the unmetered correction iff $f_j=1$.};
\end{tikzpicture}
\caption{Inputs: $\Gamma=(q_0,q_1)$, $\Delta=(r_Q,s_Q,r_P,s_P)$, all bits private.
Eight measurements occupy six rounds.
The first two rounds use two fresh ancillas. Subsequent pairs each prepare and discard
one ancilla. Double end bars denote discard. The later labels are aliases:
$(r_R,s_R)=(r_{PR},s_{PR})=(r_Q,s_Q)$. The expression $f_j$ abbreviates a guard, not a stored bit.
A minus sign changes the measured product.
Unsigned lowering implements it by flipping that measurement's recorded bit.}
\Description{Two data wires and two temporary ancillas implement four gadgets.
In rounds one and two, negative XZ and ZX measurements on disjoint pairs are followed
by Y readouts on the ancillas and controlled Z corrections on the data.
Rounds three and four measure negative XX and read out Y, correcting X on qubit one.
Rounds five and six measure negative ZXX on the two data qubits and a fresh ancilla,
read out Y, and correct ZX. Each correction is controlled by the parity of its pair of records.}
\label{fig:two-qubit-ppm}
\end{figure}

\paragraph{Step 4: retain the specification and allowance.}
For the allocated program $d$, let $\pi$ forget the four private bits.
Exact rewriting and allocation give $\WP(d,B\circ\pi)=W_0(B)\circ\pi$.
Thus every effect $B$ and precondition $A\sqsubseteq W_\theta(B)$ give
\[
 \dtriple{(A-2^{-20}I)\circ\pi}{\pwire{d}}{B\circ\pi}.
\]
For every normalized input and initial private store, public effect expectations
therefore differ from the source by at most $2^{-20}$. Optimization adds zero error.

\section{Approximate compilation}
\label{sec:approximation}

An approximate implementation may change the probability of an observation. We bound this
change uniformly over all effects. For commands of the same type, write $c\approx_\varepsilon d$ when $\varepsilon\geq0$ and,
for every effect $B$,
\[
 \WP(c,B)\sqsubseteq\WP(d,B)+\varepsilon I
 \quad\text{and}\quad
 \WP(d,B)\sqsubseteq\WP(c,B)+\varepsilon I.
\]
This relation bounds observations on the entire declared quantum
register and the retained classical records, uniformly over the initial classical store.
It is defined on the declared register only: a certified subroutine must be recertified
when it is placed inside a larger register, and the per-rotation certificates of
Section~\ref{sec:source-errors} are therefore computed on the full register of each
circuit. The following theorem uses this relation to transfer an exact Hoare
specification to an approximate implementation.

\begin{theorem}[Composition and specification transfer]
\label{thm:approximation}
Sequential approximation allowances add. Conditional approximations with the same guard use the
maximum of their branch allowances. Applying the same ancilla preparation and scope to two
approximated bodies preserves their allowance. If $c\approx_\varepsilon d$,
$\vdash\triple{A}{c}{B}$, and $B$ is an effect, then
$\vdash\triple{A-\varepsilon I}{d}{B}$.
\end{theorem}
\begin{proof}
Composition uses effect preservation, monotonicity, and identity preservation:
$\WP(d',B)$ is an effect, $\WP(c,\cdot)$ is monotone, and $\WP(c,\delta I)=\delta I$,
so the two allowances add. The conditional selects one branch in each input store.
For scopes, \eqref{wp:ancilla} lifts effects and preserves order. Hoare transfer follows
from soundness, the first approximation inequality, and completeness.
\end{proof}
For an event guaranteed by an exact triple with precondition $I$, this yields probability
at least $1-\varepsilon$ on positive trace-one inputs, and exact optimization preserves
the relation, so applying $\mathcal{O}$ to both programs does not change the valid allowances.

\paragraph{Trace-distance interpretation.}
Let $\Delta^{(t)}=\Inst_c(s,\rho,t)-\Inst_d(s,\rho,t)$ for a positive trace-one input, and let
$\|\cdot\|_1$ be the spectral trace norm. Then
\begin{equation}
 c\approx_\varepsilon d
 \quad\Longrightarrow\quad
 \frac12\sum_t\|\Delta^{(t)}\|_1\leq\varepsilon.
 \label{eq:distance}
\end{equation}
The proof identifies the half-sum with the largest absolute expectation difference over
store-indexed effects. Total trace preservation gives $\sum_t\tr(\Delta^{(t)})=0$. Individual
blocks need not have trace zero.

Section~\ref{sec:source-errors} applies the composition theorem to synthesized rotations.

\subsection{Adding logical implementation error}
\label{sec:logical-noise}

Synthesis error compares ideal programs. A noisy implementation needs a second bound.
A \emph{transformer} on the source register type is additive, order preserving, and fixes
every scalar identity. The duals of trace-preserving channels, and the weakest-precondition
maps of QPPM commands, are transformers. A \emph{location} is an ideal transformer
$F_i$, an implemented transformer $\widetilde F_i$ that includes its classical records,
and a rate $p_i\geq0$ such that, for every effect $B$,
\[
 F_i(B)\sqsubseteq\widetilde F_i(B)+p_iI,
 \qquad \widetilde F_i(B)\sqsubseteq F_i(B)+p_iI.
\]
The contract concerns the complete declared register. A measured failure frequency or
a code distance alone does not establish it.

\begin{theorem}[Synthesis and logical noise]
\label{thm:logical-noise}
Let $\ell_1,\ldots,\ell_N$ be a sequential list of locations on the source register type,
with rates $p_1,\ldots,p_N$, and let $W$ be the source transformer. Suppose that for every
effect $B$ the ideal composite $F_1\circ\cdots\circ F_N$ is within $\varepsilon$ of $W$ in
both inequalities above. If $A\sqsubseteq W(B)$ and $B$ is an effect, then
$A-(\varepsilon+\sum_i p_i)I\sqsubseteq(\widetilde F_1\circ\cdots\circ\widetilde F_N)(B)$.
\end{theorem}
\begin{proof}
Replace one location at a time:
\[
 \pwire{F_1}\,\pwire{F_2}
 \ \xrightarrow{\ p_1\ }\
 \pwire{\widetilde F_1}\,\pwire{F_2}
 \ \xrightarrow{\ p_2\ }\
 \pwire{\widetilde F_1}\,\pwire{\widetilde F_2}.
\]
The suffix preserves effects. The prefix preserves order and scalar allowances.
Induction adds the local rates, then the synthesis allowance.
\end{proof}

The compiled program enters through Equation~\eqref{eq:circuit-contract}, whose left
side is a transformer on the source register type. The ideal $F_i$ can therefore be the
compressed transformers of compiled segments and $\varepsilon$ their certified allowance.
For $N_L$ logical, $N_T$ magic-state, and $N_I$ idle locations with uniform rates
$p_L,p_T,p_I$, the total allowance is $\varepsilon+p_LN_L+p_TN_T+p_IN_I$. The rates
need not describe independent faults, because each contract must cover the state
reaching its location. In the two-qubit example, with only measurements noisy at rate
$p_L$, the 12-to-8 measurement reduction changes the allowance from $2^{-20}+12p_L$ to
$2^{-20}+8p_L$. A physical backend must also count routing, preparation, and idle
locations and establish their contracts before it can determine whether reducing depth
reduces error. An $X$ fault with probability $p$ after an ideal command is a location of
rate $p$: its transformer $B\mapsto(1-p)B+pXBX$ satisfies both inequalities because
$(I-B)+XBX$ and $B+(I-XBX)$ are positive for every effect $B$.

\section{Evaluation}
\label{sec:evaluation}

We use complete QASMBench circuits to evaluate lowering, depth optimization, and
approximation. A lowering certificate (Section~\ref{sec:qasmbench}) starts from a compiled
gate sequence. A depth certificate (Section~\ref{sec:depth-evaluation}) checks a parallel
schedule. A selection certificate (Section~\ref{sec:budget-sweep}) checks an allowance
against the requested limit. A source certificate (Section~\ref{sec:source-errors})
also accounts for synthesis. Lean checks the certificates. External tools parse sources,
synthesize gate sequences, and propose schedules. Magic states are ideal throughout.
Every experiment includes \emph{mutation tests}: deliberately corrupted certificates, such
as an altered gate, an understated allowance, or an invalid swap, which the checker must
reject.
Appendices~\ref{app:cat-case-study} and \ref{app:grover-case-study} derive full
compilation proofs for \texttt{cat\_state\_n4} and \texttt{grover\_n2}, using the
Hoare rules at every target step.

\subsection{Compilation under a requested error limit}
\label{sec:budget-method}

Algorithm~\ref{alg:compile} specifies the solid arrows in Figure~\ref{fig:compilation-stages}.
Its input is a typed circuit $c$, an error limit $\varepsilon_{\mathrm{user}}\geq0$,
an ancilla budget $a\geq1$, and a finite ordered collection $\mathcal S$ of synthesis proposals.
Each proposal supplies Clifford+$T$ words and local error certificates, which the checks
compose over the whole source before verifying the emitted target. A failed check
discards the proposal. A requested synthesis tolerance alone is insufficient.

\begin{algorithm}[htbp]
\caption{Certified compilation and depth selection}
\label{alg:compile}
\begin{algorithmic}[1]
\Require $c$, $\varepsilon_{\mathrm{user}}$, $a$, finite ordered collection $\mathcal S$
\Ensure A QPPM implementation, a public-effect allowance, and a depth upper bound, or failure
\State $\mathcal V\gets[\,]$
\For{each proposal $w\in\mathcal S$}
  \State Check $c\xrightarrow{\varepsilon_w}c_w$ by local matrix bounds and their composition
  \State \textbf{continue} if the check fails or $\varepsilon_w>\varepsilon_{\mathrm{user}}$
  \State Lower $c_w$ to $d_w$ using the verified gadgets and their magic-state contracts
  \State Check the complete instruction sequence, public interface, and magic-state inputs
  \State \textbf{continue} if the check fails
  \State $\mathsf{sd}_w\gets$ the proved serial depth bound of $d_w$
  \State Propose exact rewrites and a schedule under budget $a$
  \State Replay the rewrites and check allocation, feedback order, and the claimed depth
  \If{this check succeeds}
    \State $(d_w,\mathsf{sd}_w)\gets$ the checked optimized target and its depth bound
  \EndIf
  \State Append $(d_w,\varepsilon_w,\mathsf{sd}_w)$ and its certificates to $\mathcal V$
\EndFor
\State \Return failure if $\mathcal V=[\,]$
\State \Return the first element with least $\mathsf{sd}_w$
\end{algorithmic}
\end{algorithm}

Every returned target meets the error limit by Theorem~\ref{thm:approximation}, exact
rewrites preserve it, and the selected depth is the smallest certified bound among the
supplied proposals. If no proposal passes the checks, compilation fails.
The artifact implements the checks in separate synthesis, lowering, and scheduling
drivers, with Ross--Selinger synthesis proposals. The evaluation exercises the schedule
step on Clifford circuits (Section~\ref{sec:depth-evaluation}) and the selection step on
$T$-containing circuits with the serial bound alone (Section~\ref{sec:budget-sweep}).
Parallel scheduling of $T$ injections is future work.
The optional physical backend in the figure needs the contracts of
Section~\ref{sec:logical-noise} in addition to this algorithm.

The greedy proposal algorithm for fusion and parallel scheduling scans backwards from
each turn to an equal axis across commuting turns, adds exponents modulo four, and then
fills layers of disjoint supports in source order up to the ancilla budget. It does not
prove a minimum. All variants use the same layer costs and resource limit.

\subsection{Lowering the Qiskit outputs}
\label{sec:qasmbench}

We select the 39 canonical small-suite QASMBench inputs~\cite{li2022qasmbench} at revision
\texttt{357b942} after excluding the three error-correction circuits. The selection
is by application domain, independent of compilation results. Qiskit 2.3.1 compiles valid inputs to Clifford+$T$ using Ross--Selinger
synthesis~\cite{ross2016synthesis}, optimization level~1, seed~0, and requested local
tolerance $10^{-8}$. Four sources reference undeclared registers and fail parsing.
Eight outputs contain only Clifford gates and Z readouts. Their expanded QPPM targets
support the depth study below. The magic-state compiler handles 19 further outputs
containing T gates, with one ideal magic state per T or inverse T. For each, we check the
emitted instructions and apply the compilation theorem to obtain public Hoare
specifications and static measurement bounds. The remaining eight
outputs contain barriers, reset, or classical control, which this frontend does not
accept: it takes only the gate set of Figure~\ref{fig:gate-gallery} with Z readouts,
whereas the typed-source frontend of Section~\ref{sec:source-errors} expands reset into
a private measurement with a conditional correction and register guards into nested
conditionals, and so covers those inputs. Eight
mutation tests are rejected. A separate audit of the 56 large-suite
inputs yields 23 complete source certificates and eleven depth certificates. At four
ancillas, the 280-qubit Bernstein--Vazirani circuit falls from 1,069 to 715 logical rounds
(33.1\%) with zero approximation error.

\subsection{Logical measurement depth}
\label{sec:depth-evaluation}

\begin{table}[htbp]
\centering\footnotesize\setlength{\tabcolsep}{4pt}
\begin{tabular}{@{}lrrrr@{}}
\toprule
Input & Baseline & Commute & Fuse+commute & Reduction\\
\midrule
\texttt{bb84\_n8}  & 33  & 33  & 33  & 0.0\%\\
\texttt{cat\_state\_n4}  & 19  & 19  & 15  & 21.1\%\\
\texttt{deutsch\_n2}  & 17  & 17  & 15  & 11.8\%\\
\texttt{grover\_n2}  & 27  & 27  & 23  & 14.8\%\\
\texttt{hs4\_n4}  & 27  & 27  & 21  & 22.2\%\\
\texttt{iswap\_n2}  & 29  & 29  & 21  & 27.6\%\\
\texttt{lpn\_n5}  & 21  & 19  & 15  & 28.6\%\\
\texttt{qrng\_n4}  & 7  & 7  & 7  & 0.0\%\\
\bottomrule
\end{tabular}
\caption{Certified logical rounds with the same budget of four private gadget
ancillas, for all eight complete Clifford circuits in the selected application suite.
Inputs are the Qiskit optimization-level-1 outputs. The 19 T-containing lowering
certificates are not included.}
\label{tab:qasmbench}
\end{table}

This experiment fixes synthesis error at zero and varies the ancilla budget.
The eight expanded Clifford circuits contain 337 measurements in total. Measurement
elimination removes two of them, both readouts of \texttt{bb84\_n8}.
Table~\ref{tab:qasmbench} instead evaluates the rotation optimizer of
Section~\ref{sec:rotation-schedule} under the scheduler of Section~\ref{sec:budget-method}:
the baseline schedules disjoint macros in parallel while preserving source order on
overlapping supports, Commute relaxes that order when the checked swap permits it, and
Fuse+commute also combines equal-axis turns.
With matched budgets 1, 2, and 4, Lean checks 72 schedules across three variants.
Certificates check source expansion, macro and primitive replay, scoped allocation, disjoint logical rounds,
depth, and public weakest-precondition equality. Seventeen mutation tests are rejected.

At budget four, fusion plus commutation improves six cases by 11.8--28.6\%. Two remain
unchanged. The median reduction across all eight is 17.9\%. Commutation alone changes
one row, so most of the reduction comes from combining rotations before their
expansion into measurements. Across the eight circuits, fusion reduces the expanded
measurement count from 337 to 281 before readout cleanup.
For learning parity with noise, budgets 1, 2, and 4 give baseline/optimized depths
$43/31$, $23/17$, and $21/15$. The optimizer does not reduce the number of magic states.
The magic-input counts reported below are the T-counts of the emitted circuits.
Repeating the compilation at Qiskit optimization level~3 produces the same gate
sequences for these eight inputs. All 24 level-3 schedules are certified, and their
depths equal the corresponding level-1 depths in Table~\ref{tab:qasmbench}.
The comparison measures the contribution of our passes under one scheduler. We do not
compare with an unverified PPM scheduler, so the baseline is our own pipeline. The depth
model counts logical measurements and excludes routing and magic-state preparation
latency.

\subsection{Varying the requested error limit}
\label{sec:budget-sweep}

We hold the complete typed expansion of \texttt{qft\_n4} fixed and vary
$\varepsilon_{\mathrm{user}}$ over $10^{-2},10^{-4},10^{-6}$.
Three candidate circuits use local synthesis tolerances $10^{-5},10^{-7},10^{-9}$ for
the non-exact phases. Each certificate covers all 41 source steps and four public
readouts. The checker establishes the whole-circuit allowance $\widehat\varepsilon$, the inequality
$\widehat\varepsilon\leq\varepsilon_{\mathrm{user}}$, and the serial depth bound
$\overline{\mathsf{sd}}=6n_H+6n_{\mathrm{CX}}+2n_S+4n_T+n_M$, a gate-count formula
(two logical rounds per phase gadget, four per injection, one per readout) rather than a
parallel schedule.
Candidates use one basis ancilla at a time. Ideal magic inputs are counted in the table.

\begin{table}[htbp]
\centering\small
\begin{tabular}{@{}rrrr@{}}
\toprule
Requested limit & Certified allowance & Serial depth bound & Magic inputs\\
\midrule
$10^{-2}$ & $<1.525\times10^{-3}$ & 5,236 & 471\\
$10^{-4}$ & $<5.859\times10^{-6}$ & 7,540 & 672\\
$10^{-6}$ & $<5.434\times10^{-8}$ & 9,184 & 828\\
\bottomrule
\end{tabular}
\caption{Selection under a requested error limit for the same QASMBench QFT circuit. Each row
selects the smallest certified serial-depth bound among the feasible candidates.
These are upper bounds for serial execution, not the parallel depths of
Table~\ref{tab:qasmbench}.}
\label{tab:budget-sweep}
\end{table}

Relaxing the limit from $10^{-6}$ to $10^{-2}$ lowers the selected depth bound by
43\%. At $10^{-10}$, none of these candidates meets the requested limit, so the compiler
returns no result. Other synthesis proposals might still meet it.
For a requested limit of $10^{-3}$, the coarsest candidate is infeasible and the middle
one is selected because $7{,}540<9{,}184$. Its allowance bounds the change in the
probability of any event on the four public readouts by $5.859\times10^{-6}$. This
experiment measures precision selection under serial bounds. Parallel scheduling of $T$
injections is future work.

\subsection{Source approximation allowances}
\label{sec:source-errors}

We bound compilation error for all 35 valid typed source expansions in the selected suite, assuming ideal
magic states. Their checked outputs are synthesized separately from the Qiskit outputs above.

Eight Ross--Selinger certificates cover phases $\pm\pi/2^k$ for $k=3,4,5,6$, with 204--210 gates
including 81--84 T or inverse-T gates. Lean checks exact $\mathbb{Q}(\sqrt{2},i)$ matrices,
rational enclosures, and Taylor remainders, and derives an effect allowance below
$1.263\times10^{-7}$ for each rotation on a selected wire of the full circuit register,
including when it is entangled with the other wires.

The four-qubit QFT expansion contains 18 phase rotations. Its certificate adds the local
synthesis allowances and preserves their sum through exact lowering, giving
$\varepsilon_{\mathrm{QFT}}=9\varepsilon_{\pi/4}+6\varepsilon_{\pi/8}+3\varepsilon_{\pi/16}
<1.132\times10^{-6}$. Eighteen inputs have nonzero allowances. Table~\ref{tab:source-errors}
reports the ten for which a numerical cross-check is available. The other eight
(\texttt{dnn\_n8}, \texttt{qaoa\_n6}, \texttt{inverseqft\_n4}, \texttt{ipea\_n2},
\texttt{pea\_n5}, \texttt{qpe\_n9}, \texttt{hhl\_n7}, and \texttt{ising\_n10}) have
certified allowances between $1.264\times10^{-7}$ and $4.057\times10^{-5}$.
Appendix~\ref{app:ising-case-study} derives the Ising guarantee step by step.
Let $c$ be the typed source, $d$ its QPPM target, $k$ the number of ideal magic inputs,
$\widetilde s$ the store $s$ extended by two zero private bits, and $\varepsilon$ the
certified allowance. For every effect $B$ and store $s$, the certificate proves
\[
 -\varepsilon I\;\sqsubseteq\;
 V_k^\dagger\,\WP(d,D_k(B))(\widetilde s)\,V_k-W_c(B)(s)
 \;\sqsubseteq\;\varepsilon I.
\]
This is the approximate form of Equation~\eqref{eq:circuit-contract}. Theorem~\ref{thm:hoare}
then turns a source specification $A\sqsubseteq W_c(B)$ into a target triple
$\vdash\triple{Q}{d}{D_k(B)}$ whose compressed precondition dominates $A-\varepsilon I$.
\begin{table}[htbp]
\centering
\small
\begin{tabular}{@{}lrrrrrr@{}}
\toprule
Input & Typed steps & Output instr. & Magic inputs & Allowance $\varepsilon$ & Numeric $d$ & $\varepsilon/d$\\
\midrule
\texttt{wstate\_n3} & 33 & 246 & 93 & $<1.001\times10^{-7}$ & $4.88\times10^{-9}$ & 20\\
\texttt{linearsolver\_n3} & 22 & 469 & 172 & $<4.152\times10^{-7}$ & $9.97\times10^{-9}$ & 42\\
\texttt{quantumwalks\_n2} & 13 & 3,171 & 1,233 & $<1.299\times10^{-6}$ & $1.70\times10^{-8}$ & 76\\
\texttt{qaoa\_n3} & 18 & 1,267 & 493 & $<8.117\times10^{-7}$ & $1.80\times10^{-8}$ & 45\\
\texttt{basis\_change\_n3} & 36 & 4,766 & 1,807 & $<1.783\times10^{-6}$ & $2.32\times10^{-8}$ & 77\\
\texttt{basis\_trotter\_n4} & 1,630 & 73,084 & 28,140 & $<3.199\times10^{-5}$ & $2.75\times10^{-7}$ & 116\\
\texttt{vqe\_n4} & 94 & 6,848 & 2,628 & $<3.258\times10^{-6}$ & $4.01\times10^{-8}$ & 81\\
\texttt{dnn\_n2} & 228 & 17,572 & 6,710 & $<6.605\times10^{-6}$ & $3.24\times10^{-8}$ & 204\\
\texttt{qft\_n4} & 41 & 1,921 & 756 & $<1.132\times10^{-6}$ & $2.34\times10^{-8}$ & 48\\
\texttt{variational\_n4} & 58 & 1,648 & 646 & $<8.368\times10^{-7}$ & $1.70\times10^{-8}$ & 49\\
\bottomrule
\end{tabular}
\caption{Nonzero approximation allowances for the typed source expansions with at most
four qubits and no reset or classical control. Output instructions are the emitted
Clifford+$T$ instructions lowered to QPPM. The magic-input count is the T-count, with one
ideal magic state per injection. The numeric distance $d$ is the phase-free operator-norm
distance between the source and output unitaries, computed in floating point.}
\label{tab:source-errors}
\end{table}
The inverse-QFT certificate includes six guarded phases and four public readouts,
with each guard checked in both branches. The PEA certificate checks words prefixed by T or inverse
T directly against $\pm3\pi/8$. Exact identities reduce $X$, $Y$, and general single-qubit
rotations to phases, and macro expansion (SWAP in basis-Trotter, Qiskit's legacy
square-root-$X$ convention in VQE) adds no error. Decimal radian literals denote exact
rationals and $\pi$ the mathematical constant. HHL uses both. Exported instruction
sequences are compared with the source expansion. Parsing and macro inlining remain
outside the proof.

How tight are the allowances? For the ten inputs with at most four qubits and no reset
or classical control, we computed numerically (Qiskit, NumPy) the distance
$d=\min_\phi\|U_c-e^{i\phi}U_d\|_{\mathrm{op}}$ between the source and the certified
Clifford+$T$ unitaries, readouts removed. The effect error that the certificate bounds
lies between about $d$ and $2d$ for small $d$, and the certified allowance exceeds $d$
by a factor of 20 (\texttt{wstate\_n3}) to 204 (\texttt{dnn\_n2}). The slack has two
sources: each per-rotation allowance applies a uniform bound to sixteen matrix-unit
terms, costing up to a factor of 16 (Section~\ref{sec:two-qubit-compilation}), and the
certificate adds local allowances along the circuit although the errors of successive
rotations do not add coherently, so the ratio grows with the number of synthesized
rotations. The numeric distance is a floating-point cross-check, not a certificate.

The other seventeen inputs compile exactly. Twelve use Clifford+$T$ directly and three
(adder, SAT, and Simon) admit exact macro expansion. These fifteen (seven with T gates)
emit 603 instructions with zero error, and Simon's two barriers are proved erasable.
The remaining two are Bell and Shor. Shor preserves two resets, four register guards,
and three readouts. Appendix~\ref{app:guarded-source} derives its feedback contract.
IPEA composes approximation with three resets, eleven guards, and four readouts.
Mutation tests reject altered gates, guards, magic-input counts, and understated
allowances in every certificate.

Appendix~\ref{app:adder-case-study} works through the four-qubit adder and its
measurement bound.

\subsection{Mechanization and reproduction}
\label{sec:artifact}

The Lean audit permits only propositional extensionality, classical choice, and quotient
soundness across the compilation, matrix, and spectral results, and it rejects proof
holes and kernel bypasses. Table~\ref{tab:trust} separates the trusted components from the verified ones. A result
is \emph{certified} when the Lean kernel checks its certificate for one input. The
theorems of Sections~\ref{sec:core}--\ref{sec:approximation} are proved once for all inputs.

\begin{table}[htbp]
\centering\small
\begin{tabular}{@{}ll@{}}
\toprule
Trusted (unverified) & Verified in Lean\\
\midrule
QASM parser and macro expansion & Optimizer: elimination, fusion, scheduling\\
Qiskit Clifford+$T$ output & Gadget lowering and magic-state contracts\\
Ross--Selinger synthesis proposals & Replay checker for proposed rewrites\\
Scheduler proposals & Synthesis arithmetic: exact matrices, enclosures, remainders\\
& Selection check against the error limit\\
\bottomrule
\end{tabular}
\caption{Trust boundary of the pipeline. A proposal from the left column is accepted only
when its certificate is checked by the right column.}
\label{tab:trust}
\end{table}

The default build audits 1,476 theorems against these three assumptions with Lean and
its mathematical library at version 4.29.1. An isolated rebuild of all 48 source files
reuses only pinned third-party caches and checks source hashes. Reproduction has been
tested on one Windows x86-64 machine (Intel Core i5-12400) with two prover workers per check.

Of 7,739 physical lines, 317 define the syntax, semantics, and logic, with seven commands
and eight rule constructors. Supplemental certificates are checked and axiom-audited
separately.

\subsection{Limitations}

The parallel depth study excludes $T$ injections. The QFT selection instead uses
serial depth upper bounds. Neither includes layout, code distance, or magic-state
production. Depth selection compares only the supplied proposals and schedules and does
not establish a global minimum. We do not compare against an unverified PPM scheduler,
so our baseline is our own pipeline. The optimizer does not reduce the number of magic
states. Parsing and macro expansion are compared with exported sequences but remain unverified.
The logical-noise theorem
composes supplied contracts. It does not certify physical error rates for a
surface-code or lifted-product backend.

\section{Related work}
\label{sec:related}

\paragraph{Quantum program logics.}
\citet{dhondt2006weakest} introduced quantum weakest preconditions as duals of completely
positive maps. \citet{ying2009hoare} proved Hoare-logic soundness and relative completeness,
and \citet{feng2021classical} added the store-indexed assertions we adopt. QPPM specializes
these logics to terminating measurement programs to verify compilation and optimization.
\citet{unruh2019ghost} adds ghost variables entangled with program variables. We express
the magic-state premise through compression $V_k$ of the compiled precondition.
Relational quantum Hoare logic~\cite{barthe2020relational} uses couplings, and
aQHL~\cite{zhou2019applied} gives error-robustness rules. Our uniform effect bound
$\approx_\varepsilon$ weakens Hoare triples by a scalar and composes through classical
control. Unlike the predicate-constrained diamond distance of \citet{hung2019robustness},
it covers only the declared register and store. Its reference-system extension remains
open. CoqQ mechanizes local and parallel reasoning~\cite{zhou2022coqq}, and CoqQLR uses
distribution formulas~\cite{wu2025local}. Neither hides reusable result bits behind a
component contract.

\paragraph{Verified compilers and circuit languages.}
VOQC verifies gate-level optimizations against matrix semantics~\cite{hietala2021verified},
and SQIR proves algorithms correct in the same framework~\cite{hietala2021proving}.
QWIRE embeds a typed circuit language in Coq~\cite{paykin2017qwire}, Qbricks automates
proofs of parametric circuit families~\cite{chareton2021qbricks}, and CertiQ verifies
Qiskit passes by SMT solving~\cite{shi2019certiq}. These developments
treat gates as unitaries. QPPM transformations are verified at the measurement level,
where feedback and private outcomes are part of the semantics. The measurement
calculus~\cite{danos2007measurement} gives compositional semantics for measurement-based
computation. QPPM restricts to lattice-surgery operations and adds Hoare
specifications. \citet{booth2026stabiliser} give fully abstract semantics for
stabiliser programs. QPPM's complex-matrix instruments additionally support $T$
injection and approximate synthesis.

\paragraph{Pauli-based and lattice-surgery compilation.}
\citet{bravyi2016trading} execute Clifford+$T$ circuits as adaptive Pauli-product
measurements on magic states. \citet{litinski2019game} gives a lattice-surgery instruction
set and the rotation fusion and commutation rules used in Section~\ref{sec:rotation-schedule}.
We verify these rules and check proposed schedules on complete circuits.
\citet{paler2019clifford} schedule $T$ gates using queueing models, and
\citet{beverland2022surface} route operations along edge-disjoint paths. The Lattice
Surgery Compiler~\cite{watkins2024compiler} and \citet{leblond2024realistic} place
surface-code operations and estimate resources. Our logical depth model excludes their
layout and routing costs. \citet{wakizaka2025connectivity} verify path availability.
CircLS lowers PPMs to physical circuits and checks them by simulation~\cite{zhang2026circls}.
These checks address physical implementation. Our proof establishes source-to-PPM
semantic preservation. \citet{cowtan2019phase} synthesize phase gadgets for shallow
circuits. Our rotation fusion handles commuting equal-axis turns and is proved exact.
$T$-count reduction~\cite{amy2014polynomial,heyfron2018efficient} reduces magic-state
consumption, which our optimizer leaves unchanged.

\section{Conclusion}

QPPM's Hoare logic verifies measurement-based compilation and optimization, transfers
specifications through approximate synthesis and logical noise, and certifies complete
QASMBench circuits. Physical backends, broader source coverage, and parameterized
circuits remain future work.

\section*{AI-Assistance Statement}
OpenAI Codex was used to generate and revise Lean definitions and proofs, investigate
literature, and draft and revise the manuscript. Lean checked the formal development,
including the proof-dependency audit in Section~\ref{sec:artifact}.

\section*{Data-Availability Statement}
Artifact:
\url{https://anonymous.4open.science/r/PPMCompilation-FE38/}.
The artifact contains the Lean development, compiler passes, pinned benchmark inputs,
certificates, proof audit, dependency versions, and instructions for reproducing the
proofs and evaluation.
We plan to submit the artifact for Artifact Evaluation when the paper becomes eligible.

\label{ppm-main-end}
\bibliography{main}
\clearpage
\appendix
% Worked instances of the notation and rules whose full calculations exceed the main text.
\section{A fresh qubit: execution, precondition, and scope rule}
\label{app:scope-example}

With inputs $\Gamma=()$ and $\Delta=(r)$, the program
$f=\cmd{let}\ \nu=\cmd{ancilla}\ |0\rangle\ \cmd{in}\ M_{X_\nu}(r)$ returns one bit and no quantum
register. Its input state is the scalar $1$, viewed as a $1\times1$ density matrix.
The scope changes the quantum width from zero to one for the duration of its body.
Inside this one-qubit scope, $X_\nu=X$ and $Z_\nu=Z$.
The inclusion matrices and the measurement projectors are
\[
 E_0=\begin{pmatrix}1\\0\end{pmatrix},\qquad
 E_1=\begin{pmatrix}0\\1\end{pmatrix},\qquad
 \Pi^X_0=\tfrac12\begin{pmatrix}1&1\\1&1\end{pmatrix},\qquad
 \Pi^X_1=\tfrac12\begin{pmatrix}1&-1\\-1&1\end{pmatrix}.
\]

\paragraph{Step 1: prepare, measure, and discard.}
Preparation gives $E_0\,1\,E_0^\dagger=Q_0$. The two paths are
\[
 \begin{array}{c@{\;}c@{\;}c@{\;}c@{\;}c}
 Q_0&\pmeter[X]{X}{r=0}&
 \displaystyle\tfrac14\begin{pmatrix}1&1\\1&1\end{pmatrix}&
 \pwire{\mathsf{drop}}&\underbrace{\tfrac12}_{r=0}\\[12pt]
 Q_0&\pmeter[X]{X}{r=1}&
 \displaystyle\tfrac14\begin{pmatrix}1&-1\\-1&1\end{pmatrix}&
 \pwire{\mathsf{drop}}&\underbrace{\tfrac12}_{r=1}
 \end{array}
\]
For a one-qubit matrix $H$, discarding the only qubit gives
$E_0^\dagger H E_0+E_1^\dagger H E_1=H_{00}+H_{11}$.
Thus both displayed matrices yield the scalar $1/2$. These are the two final-store
blocks $\Inst_f(s,1,s[r\mapsto b])$, for $b=0,1$.
The qubit disappears, but the record distinguishes the branches.

\paragraph{Step 2: calculate backwards from a recorded one.}
Let $F(s)=[s(r)=1]$ be the scalar event. Lifting through discard gives
$\mathsf{drop}^*(F)(s)=I_2 F(s)$.
By \eqref{wp:measure} and \eqref{wp:ancilla}, the body updates $r$ before evaluating this assertion:
\[
 \begin{array}{c@{\qquad}l}
 \pmeter[X]{X}{r}
 &\begin{aligned}
  \WP(M_X(r),\mathsf{drop}^*(F))(s)
   &=\Pi^X_0(0I_2)\Pi^X_0+\Pi^X_1(1I_2)\Pi^X_1\\
   &=\Pi^X_1.
 \end{aligned}\\[10pt]
 \pwire{|0\rangle}
 &\begin{aligned}
  \mathsf{prep}_0^*(\Pi^X_1)
   &=E_0^\dagger\Pi^X_1 E_0\\
   &=\begin{pmatrix}1&0\end{pmatrix}
      \tfrac12\begin{pmatrix}1&-1\\-1&1\end{pmatrix}
      \begin{pmatrix}1\\0\end{pmatrix}=\tfrac12 .
 \end{aligned}
 \end{array}
\]
Consequently $\WP(f,F)=1/2$, matching the lower forward branch.
For an arbitrary scalar assertion $A$, the same calculation replaces $0,1$ by
$A(s[r\mapsto0]),A(s[r\mapsto1])$ and gives Equation~\eqref{eq:fair-wp}.

\paragraph{Step 3: apply the scope rule.}
Rule \ref{rule:primitive} supplies the premise of \ref{rule:ancilla}. The compression just calculated
supplies the conclusion's precondition:
\[
 \frac{\dtriple{\Pi^X_1}{\pmeter[X]{X}{r}}{\mathsf{drop}^*(F)}}
      {\dtriple{\tfrac12}{
       \underbrace{\pwire{|0\rangle}\!\pmeter[X]{X}{r}\!\pwire{\mathsf{drop}}}
                   _{\cmd{let}\ \nu=\cmd{ancilla}\ |0\rangle\ \cmd{in}\ M_{X_\nu}(r)}}{F}}.
\]
The boxes for preparation and discard show the scope boundary, rather than additional
commands in the grammar. The triple gives a lower bound of $1/2$ for recording one.
The exact weakest-precondition calculation gives equality.

The initialization bit also has an observable role. Replacing $X$ by $Z$ gives
$\mathsf{prep}_b^*(\Pi^Z_1)=b$ for $b\in\{0,1\}$.
Thus $\cmd{let}\ \nu=\cmd{ancilla}\ |0\rangle\ \cmd{in}\ M_{Z_\nu}(r)$ always records zero, whereas
$\cmd{let}\ \nu=\cmd{ancilla}\ |1\rangle\ \cmd{in}\ M_{Z_\nu}(r)$ always records one.

\paragraph{Names and register positions.}
The binder makes the qubit in the diagram available by name inside its body.
With inputs $\Gamma=(d)$ and $\Delta=(r)$, consider
\[
 c=\cmd{let}\ \nu=\cmd{ancilla}\ |0\rangle\ \cmd{in}\ M_{X_\nu Z_d}(r).
\]
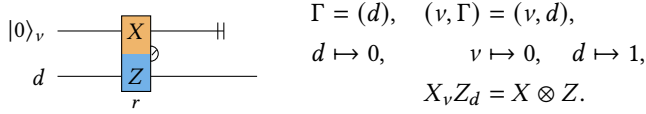
\[
\begin{tikzpicture}[x=.8cm,y=.8cm,font=\small,>=Latex,
 baseline=(current bounding box.center)]
 \node[anchor=east] at (.3,0) {$|0\rangle_\nu$};
 \node[anchor=east] at (.3,-.75) {$d$};
 \draw (.35,0) -- (3.1,0) (.35,-.75) -- (3.65,-.75);
 \pic at (1.65,0) {ppmtwo={X/Z/r}};
 \draw (3,-.14) -- (3,.14) (3.1,-.14) -- (3.1,.14);
\end{tikzpicture}
\qquad
\begin{aligned}
 \Gamma&=(d),& (\nu,\Gamma)&=(\nu,d),\\
 d&\mapsto0,&\nu&\mapsto0,\quad d\mapsto1,\\
 &&X_\nu Z_d&=X\otimes Z.
\end{aligned}
\]
Names in a context are distinct, and each bound name is fresh for the surrounding context.
Tensor factors follow context order. Omitted factors are $I$.
If an inner scope binds a fresh $\mu$, the context becomes $(\mu,\nu,d)$:
the same product $X_\nu Z_d$ then denotes $I\otimes X\otimes Z$.
The inner scope discards $\mu$ and restores $(\nu,d)$. The outer scope discards $\nu$
and restores $(d)$. Thus a name continues to refer to the same qubit when its position changes.

\section{Signs and classical expressions}
\label{app:bits}

With inputs $\Gamma=(q)$ and $\Delta=(r)$, take the signed product $P=-Z$,
so its sign bit is $\eta=1$ and its single factor is $Z$.
Its projectors exchange those of $Z$:
$\Pi^{-Z}_b=(I+(-1)^b(-Z))/2=\Pi^Z_{1-b}$.
Unsigned lowering therefore has the following concrete instance:
\[
 \pmeter[Z]{-Z}{r}\quad\equiv\quad
 \pmeter[Z]{Z}{r}\!\pwire{r:=\neg r}.
\]
On input $Q_0=|0\rangle\langle0|$, the left circuit records one with output $Q_0$.
The right circuit first records zero, then changes the record to one. For any
postcondition $A$, \eqref{wp:seq}, \eqref{wp:assign}, and \eqref{wp:measure} give
\[
 \WP(M_Z(r)\seq r:=\neg r,A)(s)
 =\sum_{b=0}^1Q_b A(s[r\mapsto1-b])Q_b
 =\sum_{a=0}^1\Pi^{-Z}_a A(s[r\mapsto a])\Pi^{-Z}_a.
\]
The second equality changes the index to $a=1-b$. In \eqref{wp:correct}, the sign cancels
without changing the store: $(-X)A(s)(-X)=XA(s)X$. Hence
$\pwire[X]{-X}\equiv\pwire[X]{X}$ as correction commands.

Literal, read, negation, and exclusive-or expressions can all be checked on this short
classical program, with inputs $\Gamma=()$ and $\Delta=(r,t,u)$.
Each annotation shows the updated store:
\[
 \underbrace{\pwire{r:=0}}_{r=0}\quad
 \underbrace{\pwire{t:=1}}_{r=0,\ t=1}\quad
 \underbrace{\pwire{r:=\neg\,\cmd{read}\ r}}_{r=1,\ t=1}\quad
 \underbrace{\pwire{u:=(\cmd{read}\ r)\oplus(\cmd{read}\ t)}}
             _{r=1,\ t=1,\ u=0}.
\]
For the last assignment and event $[u=0]$, \ref{rule:primitive} and \eqref{wp:assign} give
$[s(r)\oplus s(t)=0]$, which is true at its displayed input store.
In the phase gadget with inputs $\Gamma=(d)$, $\Delta=(a,b)$, the same exclusive-or expression controls a quantum correction:
the outcomes $(a,b)=(0,1)$ give $a\oplus b=1$, so the circuit applies $Z$.
The outcomes $(1,1)$ give zero, so it skips the correction.

\section{Shor's circuit: register guards and qubit reuse}
\label{app:guarded-source}

The QASMBench input \texttt{shor\_n5} uses iterative phase estimation. It reuses one
qubit for three readouts and chooses later phase gates from earlier outcomes.
Declare quantum inputs $\Gamma=(q_0,\ldots,q_4)$, public locations
$\Delta=(c_0,\ldots,c_4)$, and private locations $u,v\notin\Delta$.
Only $c_0,c_1,c_2$ are written. The unused public bits still participate in the
source's register comparisons. The proof therefore covers arbitrary initial values
of all five bits.

\paragraph{Step 1: expose the finite source.}
Figure~\ref{fig:guarded-source} shows the three stages in execution order.
The joint blocks abbreviate
\[
 \begin{aligned}
 U_2&=\mathrm{CNOT}_{4,2}\seq\mathrm{CNOT}_{4,0},\\
 U_3&=\mathrm{CSWAP}_{4,1,0}\seq\mathrm{CSWAP}_{4,2,1}
       \seq\mathrm{CSWAP}_{4,3,2}\\
    &\qquad\seq\mathrm{CNOT}_{4,3}\seq\mathrm{CNOT}_{4,2}
       \seq\mathrm{CNOT}_{4,1}\seq\mathrm{CNOT}_{4,0}.
 \end{aligned}
\]
Each controlled swap expands into its standard Clifford+$T$ sequence. The 25 source
statements become 73 typed steps. The complete certificate retains both resets,
all four phase guards, and all three readouts.

\begin{figure}[htbp]
\centering
\begin{tikzpicture}[x=.96cm,y=.67cm,font=\small,>=Latex,
 gate/.style={draw,fill=white,minimum width=.5cm,minimum height=.4cm,inner sep=3pt},
 joint/.style={gate,minimum height=1.1cm,minimum width=.7cm}]
 \foreach \y/\stage in {0/1,-2.5/2,-5/3}{
   \node[anchor=east] at (-.1,\y) {$q_0,\ldots,q_3$};
   \node[anchor=east] at (-.1,\y-1) {$q_4$};
   \draw (0,\y) -- (11.5,\y) (0,\y-1) -- (11.5,\y-1);
   \draw (.1,\y-.1) -- (.25,\y+.1);
   \node[font=\scriptsize,anchor=east] at (-.1,\y+.55) {stage \stage};
 }
 \node[gate,fill=ppmx] at (1,0) {$X_0$};
 \node[gate] at (1,-1) {$H$};
 \node[gate] at (2.4,-1) {$H$};
 \pic at (4,-1) {ppmone={Z/c_0}};
 \node[gate] at (7,-1) {$\mathrm{reset}_0$};
 \node[gate] at (1,-3.5) {$H$};
 \node[joint] at (2.4,-3) {$U_2$};
 \node[gate,fill=ppmz] at (4.2,-3.5) {$S$};
 \node[fill=ppmclassical,font=\scriptsize,inner sep=2pt] at (4.2,-2.9) {$w=1$};
 \node[gate] at (5.8,-3.5) {$H$};
 \pic at (7.4,-3.5) {ppmone={Z/c_1}};
 \node[gate] at (9.6,-3.5) {$\mathrm{reset}_0$};
 \node[gate] at (1,-6) {$H$};
 \node[joint] at (2.4,-5.5) {$U_3$};
 \foreach \x/\value/\phase in {4.2/3/ST,5.8/2/S,7.4/1/T}{
   \node[gate,fill=ppmz] at (\x,-6) {$\phase$};
   \node[fill=ppmclassical,font=\scriptsize,inner sep=2pt] at (\x,-5.4) {$w=\value$};
 }
 \node[gate] at (9,-6) {$H$};
 \pic at (10.6,-6) {ppmone={Z/c_2}};
\end{tikzpicture}
\caption{The complete five-qubit Shor source, read one stage at a time from top
 to bottom. A four-qubit bus carries the data register. Qubit $q_4$ is measured and
 reset between stages. Green conditions test the full public register
 $w=c_0+2c_1+4c_2+8c_3+16c_4$. Blue phase gates are the exact words for the source
 angles. The joint blocks $U_2,U_3$ are defined in the text.}
\Description{Three consecutive circuit stages on a four-qubit data bus and one
 reused phase-estimation qubit. The first stage applies X on data qubit zero,
 two Hadamards on the reused qubit, a Z measurement into c0, and reset. The second
 applies H, U2, S guarded by register value one, H, measurement into c1, and reset.
 The third applies H, U3, phases ST, S, T guarded by register values three, two,
 one respectively, H, and measurement into c2.}
\label{fig:guarded-source}
\end{figure}
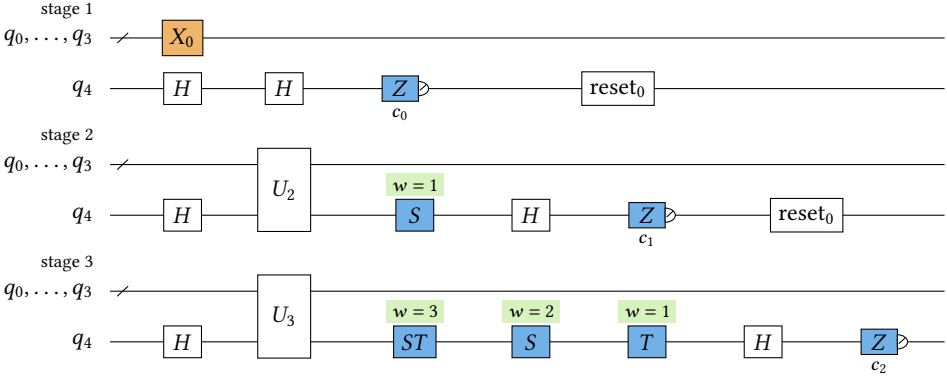

\paragraph{Step 2: lower reset without erasing public outcomes.}
Put $Q_b=\Pi_b^{Z_{q_4}}$ and $X_4=X_{q_4}$. A reset uses the private bit $u$:
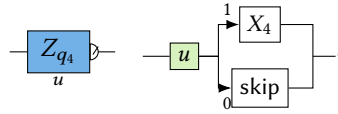
\[
 \pmeter[Z]{Z_{q_4}}{u}\quad
 \pchoice{u}{X_4}{\cmd{skip}}.
\]
For an assertion $B$ that reads only $\Delta$, substitution removes the private
record. Measurement and conditional give
\[
 \WP(\mathrm{reset}_0(q_4),B)(s)
 =Q_0B(s)Q_0+Q_1X_4B(s)X_4Q_1.
\]
These are the two Kraus terms for replacing $q_4$ by $\lvert0\rangle$.
In particular, the first reset leaves $c_0$ available to the next stage. The
second leaves both $c_0,c_1$ available. This is the reset derivation of
Section~\ref{sec:reset-wp}, instantiated on a reused algorithm qubit.

\paragraph{Step 3: derive a whole-register guard.}
Write $w(s)=\sum_{j=0}^4 2^j s(c_j)$, interpreting each Boolean as zero or one,
and $R_\theta=\operatorname{diag}(1,e^{i\theta})$ on $q_4$.
The last guarded phase, $G_1$, executes $R_{\pi/4}=T$ only when $w=1$.
Its core guard is five nested bit tests. A true $c_0$ continues to the next test,
each of $c_1,\ldots,c_4$ must be false, and every other branch skips.
Repeated use of \eqref{wp:if} gives
\[
 W_{G_1}(B)(s)=
 \begin{cases}
 T^\dagger B(s)T &s(c_0)=1,\ s(c_1)=\cdots=s(c_4)=0,\\
 B(s)&\text{otherwise}.
 \end{cases}
\]
Here $T$ is the gate matrix on $q_4$. Its exact measurement implementation has
the same public weakest precondition after resource compression. The required
magic input is the same as for an unconditional injection.

The preceding guards execute $R_{\pi/2}=S$ when $w=2$ and
$R_{3\pi/4}=ST$ when $w=3$. None writes a public location. Their conditions are
disjoint, so sequencing gives a calculation for every initial store:
\[
 \begin{array}{c|c|c}
 w(s)&\text{executed phase}&W_{G_3\seq G_2\seq G_1}(B)(s)\\\hline
 0&\pwire{\cmd{skip}}&B(s)\\[3pt]
 1&\pwire[Z]{T}&T^\dagger B(s)T\\[3pt]
 2&\pwire[Z]{S}&S^\dagger B(s)S\\[3pt]
 3&\pwire[Z]{ST}&(ST)^\dagger B(s)(ST)\\[3pt]
 4,\ldots,31&\pwire{\cmd{skip}}&B(s)
 \end{array}
\]
The last row checks why the unused high bits cannot be dropped from the guards.
The earlier $w=1$ guard is derived in the same way, with $S$ in place of $T$.

\paragraph{Step 4: compose the complete contract.}
The exact phase words produce 74 instructions and require $k=23$ ideal magic
inputs. Let $c$ be the complete typed source and $d$ its QPPM target. For any
public postcondition $B$, set $A=W_c(B)$. Composition gives
\[
 \vdash\dtriple{Q}{\pwire{d}}{D_k(B)},\qquad
 V_k^\dagger Q(\sigma)V_k=A(\pi(\sigma)),\qquad\varepsilon=0.
\]
The equality covers the whole circuit, including both resets and all feedback
branches. It preserves the source specification for every input state and public
store. Parsing and controlled-swap expansion are compared with the exported source
and remain outside the proof. This case contributes a source certificate, with no
parallel-depth claim.

\section{An arithmetic source: the four-qubit adder}
\label{app:adder-case-study}

The QASMBench input \texttt{adder\_n4} is already expressed in Clifford+$T$.
Declare $\Gamma=(q_0,q_1,q_2,q_3)$ and $\Delta=(r_0,r_1,r_2,r_3)$.
Figure~\ref{fig:adder-case-study} contains all 27 source instructions, including
the two initial X gates and four final readouts.

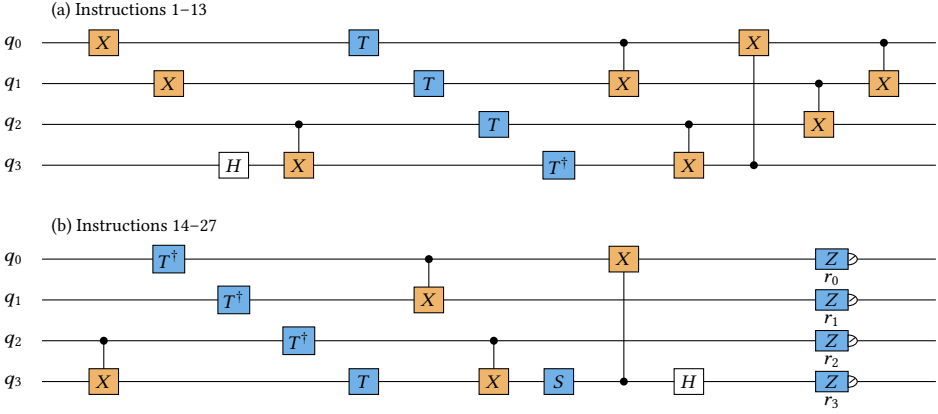
\begin{figure}[htbp]
\centering
\begin{tikzpicture}[x=.86cm,y=.54cm,font=\scriptsize,
 gate/.style={draw,fill=white,minimum width=.39cm,minimum height=.34cm,inner sep=2pt}]
 \node[anchor=east] at (.1,-0) {$q_0$};
 \draw (.25,-0) -- (14,-0);
 \node[anchor=east] at (.1,-1) {$q_1$};
 \draw (.25,-1) -- (14,-1);
 \node[anchor=east] at (.1,-2) {$q_2$};
 \draw (.25,-2) -- (14,-2);
 \node[anchor=east] at (.1,-3) {$q_3$};
 \draw (.25,-3) -- (14,-3);
 \node[anchor=west] at (.25,0.8) {(a) Instructions 1--13};
 \node[gate,fill=ppmx] at (1.2,-0) {$X$};
 \node[gate,fill=ppmx] at (2.2,-1) {$X$};
 \node[gate,fill=white] at (3.2,-3) {$H$};
 \draw (4.2,-2) -- (4.2,-3);
 \fill (4.2,-2) circle (1.5pt);
 \node[gate,fill=ppmx] at (4.2,-3) {$X$};
 \node[gate,fill=ppmz] at (5.2,-0) {$T$};
 \node[gate,fill=ppmz] at (6.2,-1) {$T$};
 \node[gate,fill=ppmz] at (7.2,-2) {$T$};
 \node[gate,fill=ppmz] at (8.2,-3) {$T^\dagger$};
 \draw (9.2,-0) -- (9.2,-1);
 \fill (9.2,-0) circle (1.5pt);
 \node[gate,fill=ppmx] at (9.2,-1) {$X$};
 \draw (10.2,-2) -- (10.2,-3);
 \fill (10.2,-2) circle (1.5pt);
 \node[gate,fill=ppmx] at (10.2,-3) {$X$};
 \draw (11.2,-3) -- (11.2,-0);
 \fill (11.2,-3) circle (1.5pt);
 \node[gate,fill=ppmx] at (11.2,-0) {$X$};
 \draw (12.2,-1) -- (12.2,-2);
 \fill (12.2,-1) circle (1.5pt);
 \node[gate,fill=ppmx] at (12.2,-2) {$X$};
 \draw (13.2,-0) -- (13.2,-1);
 \fill (13.2,-0) circle (1.5pt);
 \node[gate,fill=ppmx] at (13.2,-1) {$X$};
 \node[anchor=east] at (.1,-5.3) {$q_0$};
 \draw (.25,-5.3) -- (14,-5.3);
 \node[anchor=east] at (.1,-6.3) {$q_1$};
 \draw (.25,-6.3) -- (14,-6.3);
 \node[anchor=east] at (.1,-7.3) {$q_2$};
 \draw (.25,-7.3) -- (14,-7.3);
 \node[anchor=east] at (.1,-8.3) {$q_3$};
 \draw (.25,-8.3) -- (14,-8.3);
 \node[anchor=west] at (.25,-4.5) {(b) Instructions 14--27};
 \draw (1.2,-7.3) -- (1.2,-8.3);
 \fill (1.2,-7.3) circle (1.5pt);
 \node[gate,fill=ppmx] at (1.2,-8.3) {$X$};
 \node[gate,fill=ppmz] at (2.2,-5.3) {$T^\dagger$};
 \node[gate,fill=ppmz] at (3.2,-6.3) {$T^\dagger$};
 \node[gate,fill=ppmz] at (4.2,-7.3) {$T^\dagger$};
 \node[gate,fill=ppmz] at (5.2,-8.3) {$T$};
 \draw (6.2,-5.3) -- (6.2,-6.3);
 \fill (6.2,-5.3) circle (1.5pt);
 \node[gate,fill=ppmx] at (6.2,-6.3) {$X$};
 \draw (7.2,-7.3) -- (7.2,-8.3);
 \fill (7.2,-7.3) circle (1.5pt);
 \node[gate,fill=ppmx] at (7.2,-8.3) {$X$};
 \node[gate,fill=ppmz] at (8.2,-8.3) {$S$};
 \draw (9.2,-8.3) -- (9.2,-5.3);
 \fill (9.2,-8.3) circle (1.5pt);
 \node[gate,fill=ppmx] at (9.2,-5.3) {$X$};
 \node[gate,fill=white] at (10.2,-8.3) {$H$};
 \pic at (12.4,-5.3) {ppmone={Z/r_0}};
 \pic at (12.4,-6.3) {ppmone={Z/r_1}};
 \pic at (12.4,-7.3) {ppmone={Z/r_2}};
 \pic at (12.4,-8.3) {ppmone={Z/r_3}};
\end{tikzpicture}
\caption{The complete four-qubit adder source. Panel (b) continues panel (a).
Orange boxes apply X, including the targets of CNOTs. Blue phase gates rotate
about Z. All four public outputs are retained by the compiled QPPM program.}
\Description{Two consecutive four-wire panels show the complete adder input.
The circuit contains two X gates, two Hadamards, ten CNOTs, one S gate, four T
gates, four inverse-T gates, and four final Z measurements.}
\label{fig:adder-case-study}
\end{figure}

\paragraph{Step 1: count the component contracts.}
The source has two X gates, two H gates, ten CNOTs, one S gate, eight T or inverse-T
gates, and four Z readouts. The X gates compile to corrections. Each H and CNOT
uses six measurements, each S uses two, and each injection uses four. Therefore
\[
 \pwire{\mathrm{Adder}_4}\xrightarrow{\quad\varepsilon=0\quad}
 \pwire{d_{\mathrm{add}}},\qquad
 6\cdot2+6\cdot10+2\cdot1+4\cdot8+4=110.
\]
The certificate checks this bound for the complete target. It supplies eight ideal
magic inputs and two private classical locations $u,v\notin\Delta$.
Each basis ancilla has a local scope.

\paragraph{Step 2: compose assertions at gate boundaries.}
Write $g_1,\ldots,g_{27}$ for the displayed instructions and let $B$ be any public
postcondition. Set $A_{27}=B$ and $A_{i-1}=W_{g_i}(A_i)$.
For example, the last readout, the final H, and the preceding CNOT give successive
backward calculations
\[
 \begin{aligned}
 A_{26}(s)&=\sum_b\Pi_b^{Z_{q_3}}B(s[r_3\mapsto b])\Pi_b^{Z_{q_3}},\\
 A_{22}&=H_3^\dagger A_{23}H_3,\\
 A_{21}&=\mathrm{CNOT}_{3,0}^\dagger A_{22}\mathrm{CNOT}_{3,0}.
 \end{aligned}
\]
The three intervening readouts define $A_{25},A_{24},A_{23}$ by the same substitution
rule. Every local lowering has exactly the corresponding public transformer.
Sequencing joins these assertion boundaries from right to left. With the resource
embedding $V_8$ and assertion lift $D_8$, the complete result is
\[
 \vdash\dtriple{Q}{\pwire{d_{\mathrm{add}}}}{D_8(B)},\qquad
 V_8^\dagger Q(\sigma)V_8=A_0(\pi(\sigma)).
\]
There is no approximation term: this source needs no angle synthesis. The theorem
preserves every source postcondition for every data input. It is a compilation
correctness result. A functional proof that the source implements addition would
be a separate source-level specification.

\section{Approximation through classical feedback}
\label{app:inverseqft}

With inputs $\Gamma=(q_0,q_1,q_2,q_3)$ and $\Delta=(r_0,r_1,r_2,r_3)$,
the QASMBench source in Figure~\ref{fig:inverseqft-source} combines phase
synthesis with measurement-dependent control. Consider its guarded $\pi/8$
phase $R_{\pi/8}=\operatorname{diag}(1,e^{i\pi/8})$ on $q_3$.
Let $W$ be its synthesized word on the same wire, and let $\varepsilon_8$
be its checked effect allowance. The conditional rule compares these paths:
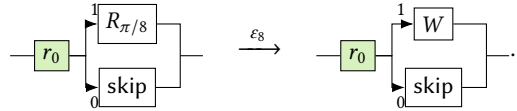
\[
 \pchoice{r_0}{R_{\pi/8}}{\cmd{skip}}
 \quad\xrightarrow{\ \varepsilon_8\ }\quad
 \pchoice{r_0}{W}{\cmd{skip}}.
\]
For $s(r_0)=0$, both weakest preconditions are $B(s)$, so their difference
is zero. For $s(r_0)=1$, they are $R_{\pi/8}^{\dagger}B(s)R_{\pi/8}$ and
$W^{\dagger}B(s)W$. The local certificate bounds this difference by
$\varepsilon_8 I$ in both directions. Thus the same allowance holds for
every store, without assuming a probability for either branch.
No gate in the Clifford+$T$ word changes a public bit. Testing the same guard
before each gate therefore makes the same branch selection throughout the word.

\begin{figure}[htbp]
\centering
\begin{tikzpicture}[x=.94cm,y=.86cm,font=\small,
 gate/.style={draw,fill=white,minimum width=.42cm,minimum height=.4cm,inner sep=2pt},
 phase/.style={gate,fill=ppmz}]
 \foreach \y/\q in {0/0,-1/1,-2/2,-3/3}{
   \draw (0,\y) -- (11.2,\y);
   \node[anchor=east] at (-.1,\y) {$q_\q$};
   \node[gate] at (.4,\y) {$H$};
 }
 \draw[densely dashed] (.95,.28) -- (.95,-3.28);
 \foreach \x/\y in {1.5/0,3.5/-1,6.2/-2,9.8/-3}{
   \node[gate] at (\x,\y) {$H$};
 }
 \foreach \x/\y/\r in {2.1/0/0,4.1/-1/1,6.8/-2/2,10.5/-3/3}{
   \pic at (\x,\y) {ppmone={Z/r_\r}};
 }
 \foreach \x/\y/\g/\r in {2.8/-1/S/0,4.8/-2/T/0,5.5/-2/S/1,
                              8.3/-3/T/1,9.1/-3/S/2}{
   \node[phase] at (\x,\y) {$\g$};
   \node[font=\scriptsize,fill=ppmclassical,inner sep=2pt] at (\x,\y+.48) {$r_\r$};
 }
 \node[phase] at (7.5,-3) {$R_{\pi/8}$};
 \node[font=\scriptsize,fill=ppmclassical,inner sep=2pt] at (7.5,-2.52) {$r_0$};
\end{tikzpicture}
\caption{The complete source named \texttt{inverseqft\_n4} in QASMBench.
A green bit label means that the phase executes when that bit is one.
The source's $\pi/2$ and $\pi/4$ phases are displayed as S and T. The
$\pi/8$ phase is replaced by a certified word of 210 Clifford+$T$ gates.
All six guards and all four public readouts occur in the checked expansion.}
\Description{Four wires start with H gates and a barrier. Each wire then ends
with H followed by a Z measurement. Before that final H, the second wire has
an S guarded by the first outcome. The third has T and S guarded by the first
two outcomes. The fourth has a pi-over-eight phase, T, and S guarded by the
first three outcomes. Operations appear in source order from left to right.}
\label{fig:inverseqft-source}
\end{figure}
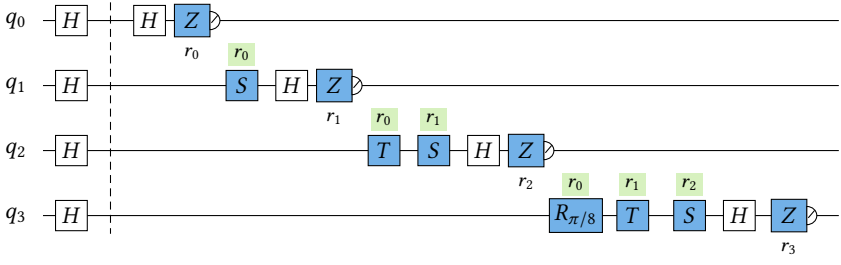

The source has three $\pi/2$ phases, two $\pi/4$ phases, and one $\pi/8$
phase. The first two kinds use S and T, respectively. The analytic checker
assigns them conservative allowances $\varepsilon_S<1.612\times10^{-11}$ and
$\varepsilon_T<1.601\times10^{-11}$. The synthesized word has
$\varepsilon_8<1.263\times10^{-7}$. The exact gates receive a tiny nonzero allowance
because the checker encloses $\pi$ by rationals and bounds a Taylor remainder of the
phase, rather than recognizing $S$ and $T$ symbolically. Pulling the final effect back through
the measurements and composing the six guarded replacements gives
\[
 \varepsilon_{\mathrm{invQFT}}
   =3\varepsilon_S+2\varepsilon_T+\varepsilon_8
   <1.264\times10^{-7}.
\]
All comparisons and the sum use exact rational allowances in the certificate.
The source's 19 steps emit 227 Clifford+$T$ instructions. Their PPM target
reserves 85 ideal magic-state inputs and has a bound of 988 primitive
measurements. When a guard is false, the corresponding magic states stay
idle. Marginalizing them leaves the data unchanged. The proof covers these
skipped injections as well as the executed ones. Lowering adds no approximation
error, and the ordinary Hoare precondition loses only
$\varepsilon_{\mathrm{invQFT}}I$. This case does not add a parallel-depth result.

\section{A benchmark proof from source gates to measurements}
\label{app:cat-case-study}

We prove compilation correctness for the complete QASMBench input
\texttt{cat\_state\_n4} from Table~\ref{tab:qasmbench}. The proof first treats an
arbitrary public postcondition. It then derives the benchmark's cat-state readout
property. Every source instruction and every compiled measurement scope appears
below. The proof uses the weakest-precondition equations and the Hoare rules of
Section~\ref{sec:hoare}.

\paragraph{Step 1: state the source and goal.}
The input contexts are $\Gamma=(q_0,q_1,q_2,q_3)$ and
$\Delta=(r_0,r_1,r_2,r_3)$. The source at the benchmark revision used in
Section~\ref{sec:qasmbench} is
\[
 \begin{aligned}
 c={}&H_0\seq\mathrm{CNOT}_{01}\seq\mathrm{CNOT}_{12}
       \seq\mathrm{CNOT}_{23}\seq\\
     &M_{Z_0}(r_0)\seq M_{Z_1}(r_1)\seq M_{Z_2}(r_2)\seq M_{Z_3}(r_3).
 \end{aligned}
\]
Subscripts name data wires. Each CNOT's first index is its control.
The four input qubits are supplied by the caller. We impose the usual all-zero
input only when deriving the concrete output probabilities in Step~6.
Figure~\ref{fig:cat-case-study} places an assertion at every instruction boundary.

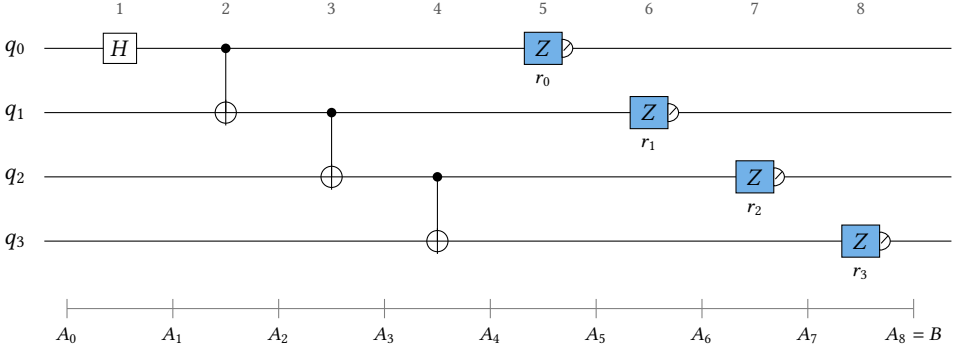
\begin{figure}[htbp]
\centering
\begin{tikzpicture}[x=1cm,y=.85cm,font=\small,>=Latex,
 gate/.style={draw,fill=white,minimum width=.44cm,minimum height=.4cm,inner sep=2pt}]
 \foreach \y/\j in {0/0,-1/1,-2/2,-3/3}{
   \draw (0,\y) -- (12,\y);
   \node[anchor=east] at (-.12,\y) {$q_\j$};
 }
 \node[gate] at (1,0) {$H$};
 \foreach \x/\c/\t in {2.4/0/-1,3.8/-1/-2,5.2/-2/-3}{
   \draw (\x,\c) -- (\x,\t);
   \fill (\x,\c) circle (1.8pt);
   \draw (\x,\t) circle (.14cm);
   \draw (\x,\t-.2) -- (\x,\t+.2);
 }
 \foreach \x/\y/\j in {6.6/0/0,8/-1/1,9.4/-2/2,10.8/-3/3}{
   \pic at (\x,\y) {ppmone={Z/r_\j}};
 }
 \foreach \x/\j in {1/1,2.4/2,3.8/3,5.2/4,6.6/5,8/6,9.4/7,10.8/8}{
   \node[font=\scriptsize,text=black!65] at (\x,.62) {\j};
 }
 \draw[gray] (.3,-4.05) -- (11.5,-4.05);
 \foreach \x/\j in {.3/0,1.7/1,3.1/2,4.5/3,5.9/4,7.3/5,8.7/6,10.1/7}{
   \draw[gray] (\x,-3.9) -- (\x,-4.2);
   \node[anchor=north,font=\scriptsize,inner sep=2pt] at (\x,-4.25) {$A_\j$};
 }
 \draw[gray] (11.5,-3.9) -- (11.5,-4.2);
 \node[anchor=north,font=\scriptsize,inner sep=2pt] at (11.5,-4.25) {$A_8=B$};
\end{tikzpicture}
\caption{The complete \texttt{cat\_state\_n4} source with public results $r_0,r_1,r_2,r_3$.
Steps 1--4 apply $H$ to $q_0$, then CNOTs from $q_0$ to $q_1$, from $q_1$ to $q_2$,
and from $q_2$ to $q_3$. Steps 5--8 measure $Z$ on $q_0,q_1,q_2,q_3$ in that order.
Each blue meter records its result in the displayed bit. The bottom strip places the
full-register assertions $A_0,\ldots,A_8$ at the boundaries of these eight steps.
The final assertion is $A_8=B$.}
\Description{Four quantum wires carry a Hadamard on qubit zero followed by a chain of
three CNOT gates. Four Z measurements occupy separate columns in source order and write
public bits zero through three. A bottom strip aligns nine assertion boundaries with
the eight operations, from A zero before the Hadamard to A eight equal to B after the
last measurement.}
\label{fig:cat-case-study}
\end{figure}

Let $B$ be any Hermitian assertion on the four data qubits and public store.
The target also has distinct private bits $a,b$, disjoint from $\Delta$.
Let $\pi$ forget these two bits. We will calculate $A_0=W_c(B)$ and derive
\begin{equation}
 \WP(d,B\circ\pi)=A_0\circ\pi,
 \qquad \vdash\triple{A_0\circ\pi}{d}{B\circ\pi}.
 \label{eq:cat-goal}
\end{equation}
The equality proves correctness for every public observation, including joint
quantum and classical assertions. It holds for every initial value of $a,b$.
There are no magic-state inputs or approximation allowances in this case.

\paragraph{Step 2: calculate all eight source boundaries.}
Write $Q_x^{(j)}=\Pi_x^{Z_j}$ and $C_{ij}=\mathrm{CNOT}_{ij}$ for the source
matrix. All omitted tensor factors are identities. Starting from $A_8=B$, four
uses of \eqref{wp:measure} give
\begin{align*}
 A_7(s)&=\sum_w Q_w^{(3)}B(s[r_3\mapsto w])Q_w^{(3)},\\
 A_6(s)&=\sum_{z,w}Q_z^{(2)}Q_w^{(3)}
    B(s[r_2\mapsto z,r_3\mapsto w])Q_w^{(3)}Q_z^{(2)},\\
 A_5(s)&=\sum_{y,z,w}Q_y^{(1)}Q_z^{(2)}Q_w^{(3)}
    B(s[r_1\mapsto y,r_2\mapsto z,r_3\mapsto w])
    Q_w^{(3)}Q_z^{(2)}Q_y^{(1)},\\
 A_4(s)&=\sum_{x,y,z,w}Q_x^{(0)}Q_y^{(1)}Q_z^{(2)}Q_w^{(3)}
    B(s[r_0\mapsto x,r_1\mapsto y,r_2\mapsto z,r_3\mapsto w])
    Q_w^{(3)}Q_z^{(2)}Q_y^{(1)}Q_x^{(0)}.
\end{align*}
Every summation variable ranges over $\{0,1\}$. In each line the measurement
substitutes its own result before evaluating the suffix assertion. For example,
$A_6(s)=\sum_z Q_z^{(2)}A_7(s[r_2\mapsto z])Q_z^{(2)}$.
The distinct result locations make the displayed substitutions commute.
The four remaining source transformers are unitary conjugations:
\[
 \begin{aligned}
 A_3&=C_{23}^{\dagger}A_4C_{23},&
 A_2&=C_{12}^{\dagger}A_3C_{12},\\
 A_1&=C_{01}^{\dagger}A_2C_{01},&
 A_0&=H_0^{\dagger}A_1H_0.
 \end{aligned}
\]
These equations specify every boundary in the source. We next derive the target
triples that connect the same boundaries.

\paragraph{Step 3: write every compiled scope.}
Table~\ref{tab:cat-scopes} lists the target's twelve scopes in execution order.
For row $j$, the complete command $d_j$ is
\begin{equation}
 \begin{aligned}
 d_j={}&\cmd{let}\ \nu=\cmd{ancilla}\ |0\rangle\ \cmd{in}\bigl(
       M_{X_\nu P_j}(a)\seq M_{Y_\nu}(b)\seq\\
     &\qquad\cmd{if}\ a\oplus b\ \cmd{then}\ \cmd{correct}\ P_j
       \ \cmd{else}\ \cmd{skip}\bigr).
 \end{aligned}
 \label{eq:cat-scope-program}
\end{equation}
Inside the scope, $P_j$ acts on the data and as identity on $\nu$.
Thus $X_\nu P_j=X\otimes P_j$ and $Y_\nu=Y\otimes I$.
For a negative $P_j$, the joint measurement is signed. The sign of its Pauli
correction cancels in conjugation, so the correction may use the unsigned product.
Each row prepares and discards its own fresh qubit. The bits $a,b$ are reused.

\begin{table}[htbp]
 \centering\small
 \begin{tabular}{@{}clcl@{}}
 \toprule
 Row $j$ & Source gate & $P_j$ & Derived triple\\
 \midrule
 1 & $H_0$ & $Z_0$ & $\{T_0\circ\pi\}\ d_1\ \{T_1\circ\pi\}$\\
 2 & & $X_0$ & $\{T_1\circ\pi\}\ d_2\ \{T_2\circ\pi\}$\\
 3 & & $Z_0$ & $\{T_2\circ\pi\}\ d_3\ \{T_3\circ\pi\}$\\
 \midrule
 4 & $\mathrm{CNOT}_{01}$ & $Z_0$ & $\{T_3\circ\pi\}\ d_4\ \{T_4\circ\pi\}$\\
 5 & & $X_1$ & $\{T_4\circ\pi\}\ d_5\ \{T_5\circ\pi\}$\\
 6 & & $-Z_0X_1$ & $\{T_5\circ\pi\}\ d_6\ \{T_6\circ\pi\}$\\
 \midrule
 7 & $\mathrm{CNOT}_{12}$ & $Z_1$ & $\{T_6\circ\pi\}\ d_7\ \{T_7\circ\pi\}$\\
 8 & & $X_2$ & $\{T_7\circ\pi\}\ d_8\ \{T_8\circ\pi\}$\\
 9 & & $-Z_1X_2$ & $\{T_8\circ\pi\}\ d_9\ \{T_9\circ\pi\}$\\
 \midrule
 10 & $\mathrm{CNOT}_{23}$ & $Z_2$ & $\{T_9\circ\pi\}\ d_{10}\ \{T_{10}\circ\pi\}$\\
 11 & & $X_3$ & $\{T_{10}\circ\pi\}\ d_{11}\ \{T_{11}\circ\pi\}$\\
 12 & & $-Z_2X_3$ & $\{T_{11}\circ\pi\}\ d_{12}\ \{T_{12}\circ\pi\}$\\
 \bottomrule
 \end{tabular}
 \caption{Every scope of the compiled cat-state benchmark.
 Each $d_j$ is the explicit program in Equation~\eqref{eq:cat-scope-program}
 with the displayed Pauli product. The assertions are calculated in Steps~4--5.}
 \label{tab:cat-scopes}
\end{table}

The full target is
\[
 d=d_1\seq d_2\seq\cdots\seq d_{12}\seq
 M_{Z_0}(r_0)\seq M_{Z_1}(r_1)\seq M_{Z_2}(r_2)\seq M_{Z_3}(r_3).
\]
The finite list contains 24 internal Pauli measurements and four public readouts.
It uses one temporary qubit at a time. These 28 measurements match the
original-source certificate's measurement bound. This proof starts from the
displayed typed source, as do the certificates. Parsing the benchmark file is
outside the logic.

Figure~\ref{fig:cat-ppm} draws the complete compiled circuit.
Its numbered scopes are the rows of Table~\ref{tab:cat-scopes}.
The boundary after scope $j$ carries $T_j\circ\pi$.
The displayed group boundaries connect these assertions to $A_0,\ldots,A_4$.
The four public meters then connect $A_4$ to $A_8=B$.
Step~4 proves each scope triple. Step~5 joins them in the displayed order.

\begin{figure}[p]
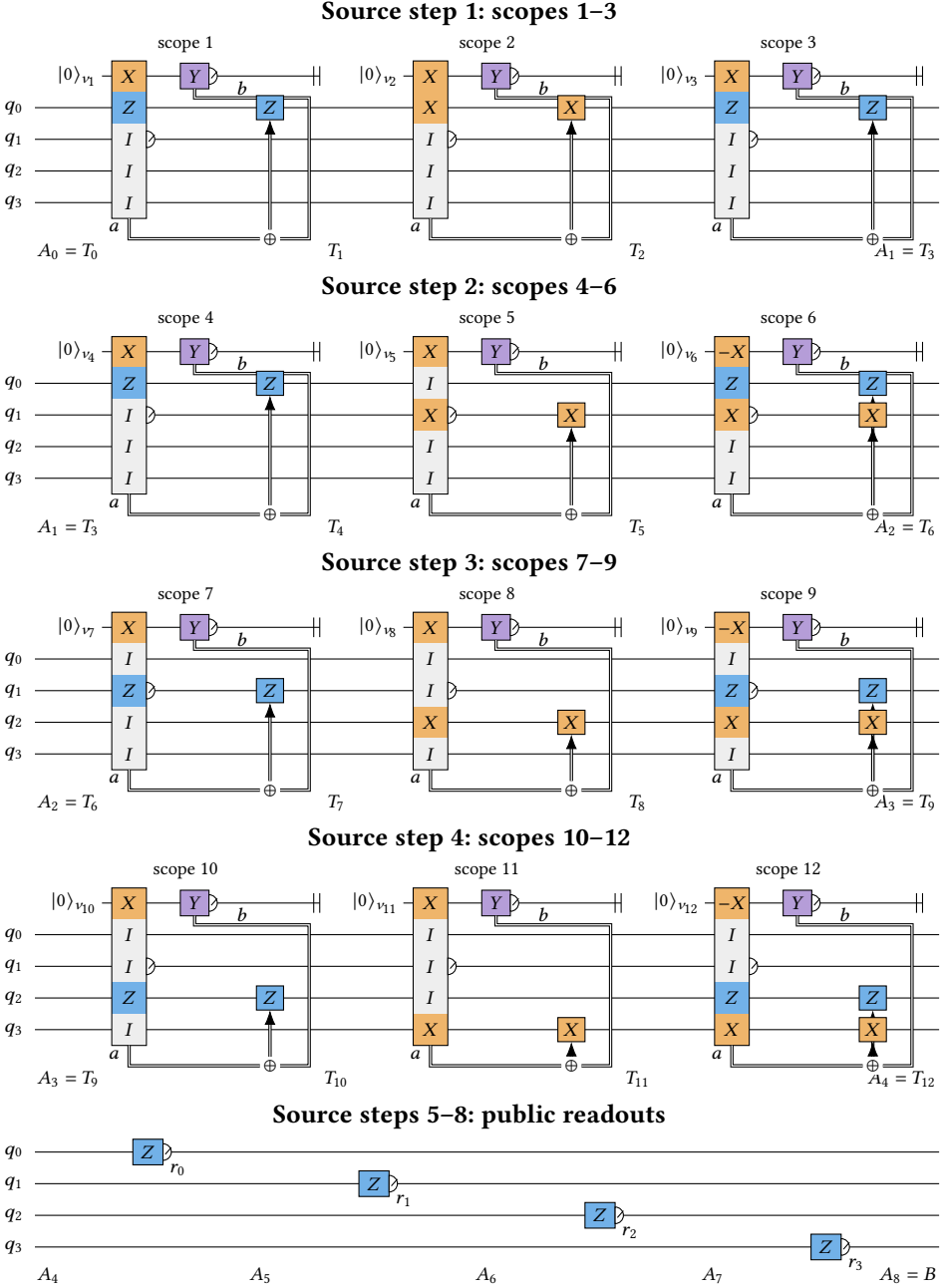

\centering
\textbf{Source step 1: scopes 1--3}\par\smallskip
\begin{ppmcompiledrow}{4}{A_0=T_0}{A_1=T_3}
 \ppmscope{0}{1}{}{Z,I,I,I}
 \ppmscope{4}{2}{}{X,I,I,I}
 \ppmscope{8}{3}{}{Z,I,I,I}
 \node at (4,\ppmcompiledassert) {$T_1$};
 \node at (8,\ppmcompiledassert) {$T_2$};
\end{ppmcompiledrow}\par\smallskip
\textbf{Source step 2: scopes 4--6}\par\smallskip
\begin{ppmcompiledrow}{4}{A_1=T_3}{A_2=T_6}
 \ppmscope{0}{4}{}{Z,I,I,I}
 \ppmscope{4}{5}{}{I,X,I,I}
 \ppmscope{8}{6}{-}{Z,X,I,I}
 \node at (4,\ppmcompiledassert) {$T_4$};
 \node at (8,\ppmcompiledassert) {$T_5$};
\end{ppmcompiledrow}\par\smallskip
\textbf{Source step 3: scopes 7--9}\par\smallskip
\begin{ppmcompiledrow}{4}{A_2=T_6}{A_3=T_9}
 \ppmscope{0}{7}{}{I,Z,I,I}
 \ppmscope{4}{8}{}{I,I,X,I}
 \ppmscope{8}{9}{-}{I,Z,X,I}
 \node at (4,\ppmcompiledassert) {$T_7$};
 \node at (8,\ppmcompiledassert) {$T_8$};
\end{ppmcompiledrow}\par\smallskip
\textbf{Source step 4: scopes 10--12}\par\smallskip
\begin{ppmcompiledrow}{4}{A_3=T_9}{A_4=T_{12}}
 \ppmscope{0}{10}{}{I,I,Z,I}
 \ppmscope{4}{11}{}{I,I,I,X}
 \ppmscope{8}{12}{-}{I,I,Z,X}
 \node at (4,\ppmcompiledassert) {$T_{10}$};
 \node at (8,\ppmcompiledassert) {$T_{11}$};
\end{ppmcompiledrow}\par\smallskip
\textbf{Source steps 5--8: public readouts}\par\smallskip
\begin{ppmcompiledrow}[.38]{4}{A_4}{A_8=B}
 \ppmreadout{1.5}{0}{Z}{r_0}
 \ppmreadout{4.5}{1}{Z}{r_1}
 \ppmreadout{7.5}{2}{Z}{r_2}
 \ppmreadout{10.5}{3}{Z}{r_3}
 \node at (3,\ppmcompiledassert) {$A_5$};
 \node at (6,\ppmcompiledassert) {$A_6$};
 \node at (9,\ppmcompiledassert) {$A_7$};
\end{ppmcompiledrow}
\caption{The complete compiled QPPM circuit for \texttt{cat\_state\_n4}.
Read left to right, then top to bottom, keeping the same data qubits.
Each scope prepares a fresh zero ancilla. Its tall meter measures the displayed
Pauli product into $a$. The purple meter measures $Y_\nu$ into $b$.
The $\oplus$ node applies the unsigned correction product iff $a\oplus b=1$.
Double bars discard the ancilla. The next scope reuses private bits $a,b$.
Minus signs belong to the joint observables. All assertions are lifted by $\pi$.
Steps~4--5 derive these boundaries. All 28 measurements are shown before fusion.}
\Description{Five consecutive circuit rows show the full compiled cat-state
program on four persistent data wires. Twelve fresh ancilla scopes each have
a joint Pauli measurement, a Y measurement, a parity-controlled correction, and
ancilla disposal. The signed joint measurements are in scopes six, nine, and
twelve. Four final Z meters write the public bits. Assertions connect every
row to the source instruction boundaries.}
\label{fig:cat-ppm}
\end{figure}

\paragraph{Step 4: derive a row using the primitive and structural rules.}
Fix a row, put $P=P_j$, and let $F$ be its public suffix assertion.
We calculate backwards through every command in that row.
Write $\sigma$ for a store containing $a,b$ and the public bits. At the body exit,
\eqref{wp:ancilla} requires the assertion
$D(\sigma)=I_2\otimes F(\pi(\sigma))$.
Set $C_{xy}=I_2\otimes P^{x\oplus y}$.
Let $f$ denote the conditional displayed in Equation~\eqref{eq:cat-scope-program}.
The four backward steps are
\begin{align}
 D_2(\sigma)
   &=C_{\sigma(a),\sigma(b)}D(\sigma)C_{\sigma(a),\sigma(b)},
       &&\text{\eqref{wp:if}, \eqref{wp:correct}, \eqref{wp:skip}},\label{eq:cat-feedback-wp}\\
 D_1(\sigma)
   &=\sum_y\Pi_y^{Y_\nu}C_{\sigma(a),y}D(\sigma)C_{\sigma(a),y}\Pi_y^{Y_\nu},
       &&\text{\eqref{wp:measure}},\label{eq:cat-y-wp}\\
 D_0(\sigma)
   &=\sum_{x,y}\Pi_x^{X_\nu P}\Pi_y^{Y_\nu}C_{xy}D(\sigma)C_{xy}
                    \Pi_y^{Y_\nu}\Pi_x^{X_\nu P},
       &&\text{\eqref{wp:measure}},\label{eq:cat-joint-wp}\\
 \WP(d_j,F\circ\pi)(\sigma)
   &=E_0^\dagger D_0(\sigma)E_0,
       &&\text{\eqref{wp:seq}, \eqref{wp:ancilla}}.\label{eq:cat-ancilla-wp}
\end{align}
The substitution for $b$ removes its old value in the second line. The substitution
for $a$ removes its old value in the third. Neither substitution changes $D$.
This proves independence from the initial scratch values.

Rules~\ref{rule:primitive} and \ref{rule:skip} supply the leaves:
\[
 \vdash\triple{(I_2\otimes P)D(I_2\otimes P)}{\cmd{correct}\ P}{D},
 \qquad\vdash\triple{D}{\cmd{skip}}{D}.
\]
Rule~\ref{rule:if} combines these leaves into $\triple{D_2}{f}{D}$.
The two measurement instances of \ref{rule:primitive} give
\[
 \dtriple{D_0}{\pmeter{X_\nu P}{a}}{D_1},
 \qquad
 \dtriple{D_1}{\pmeter[Y]{Y_\nu}{b}}{D_2}.
\]
Two applications of \ref{rule:seq} join these triples to the conditional.
Rule~\ref{rule:ancilla} then yields
\[
 \frac{\displaystyle
   \vdash\dtriple{D_0}{\pmeter{X_\nu P}{a}\!\pmeter[Y]{Y_\nu}{b}\!\pwire{f}}{D}}
 {\displaystyle
   \vdash\htriple{E_0^\dagger D_0E_0}{d_j}{F\circ\pi}}
 \quad\text{\ref{rule:ancilla}}.
\]
Thus the row triple is derived from the listed rules, with its precondition still
given by an explicit matrix expression.

To simplify that expression, expand
$D=\sum_{z=0}^1 E_zF(\pi(\sigma))E_z^\dagger$ and put
$K_{xyz}=E_z^\dagger C_{xy}\Pi_y^{Y_\nu}\Pi_x^{X_\nu P}E_0$.
The Pauli-matrix calculation in Equation~\eqref{eq:phase-kraus} gives
\[
 \begin{aligned}
 K_{xy0}&=\tfrac14P^{x\oplus y}\bigl(I-i(-1)^{x\oplus y}P\bigr)
          =\alpha_{xy0}U(P),\\
 K_{xy1}&=i(-1)^yK_{xy0},\qquad
 \alpha_{xyz}=[i(-1)^y]^z\frac{1-i(-1)^{x\oplus y}}4.
 \end{aligned}
\]
Here $U(P)=\Pi_0^P+i\Pi_1^P$ is the matrix defined in
Section~\ref{sec:gate-gallery}. Each of the eight scalars has squared magnitude
$2/16=1/8$. Therefore
\[
 E_0^\dagger D_0(\sigma)E_0
 =\sum_{x,y,z}K_{xyz}^\dagger F(\pi(\sigma))K_{xyz}
 =U(P)^\dagger F(\pi(\sigma))U(P).
\]
This calculation applies to the whole four-qubit matrix $F$.
It does not require a product state or a separable assertion.

\paragraph{Step 5: join all target triples and identify the source precondition.}
Set $T_{12}=A_4$. Read the twelve rows of Table~\ref{tab:cat-scopes} backwards,
using the calculated precondition
\[
 T_{j-1}=U(P_j)^\dagger T_jU(P_j),\qquad j=12,11,\ldots,1.
\]
For example, the last three instances are
\[
 \begin{aligned}
 T_{11}&=U(-Z_2X_3)^\dagger A_4U(-Z_2X_3),\\
 T_{10}&=U(X_3)^\dagger T_{11}U(X_3),\\
 T_9&=U(Z_2)^\dagger T_{10}U(Z_2)=C_{23}^\dagger A_4C_{23}=A_3.
 \end{aligned}
\]
The last equality uses
$U(-Z_2X_3)U(X_3)U(Z_2)=iC_{23}$.
Applying the same matrix identity to the two other specified pairs, followed by
$U(Z_0)U(X_0)U(Z_0)=e^{i\pi/4}H_0$, gives all gate boundaries:
\[
 \begin{array}{c|c|c}
 \text{rows}&\text{matrix in execution order}&\text{backward boundary}\\\hline
 10,11,12&U(-Z_2X_3)U(X_3)U(Z_2)=iC_{23}&T_9=A_3\\
 7,8,9&U(-Z_1X_2)U(X_2)U(Z_1)=iC_{12}&T_6=A_2\\
 4,5,6&U(-Z_0X_1)U(X_1)U(Z_0)=iC_{01}&T_3=A_1\\
 1,2,3&U(Z_0)U(X_0)U(Z_0)=e^{i\pi/4}H_0&T_0=A_0
 \end{array}
\]
Matrices multiply from right to left. Their scalar phases cancel in every
conjugation. Step~4 supplies the twelve row triples in the table.
Step~2 and \ref{rule:primitive} supply the four readout triples
\[
 \vdash\dtriple{A_{4+j}\circ\pi}{\pmeter[Z]{Z_j}{r_j}}
                        {A_{5+j}\circ\pi},\qquad j=0,1,2,3.
\]
Their boundaries match because $T_{12}=A_4$ and $A_8=B$.
Fifteen applications of \ref{rule:seq} join the sixteen target triples into
Equation~\eqref{eq:cat-goal}. Replacing each triple by its calculated
weakest-precondition equality and applying \eqref{wp:seq} gives its equality
statement as well. This proves the complete compilation contract.

\paragraph{Step 6: derive the benchmark's concrete readout guarantee.}
Choose $B=[r_0=r_1=r_2=r_3]$. In the following display a bracket is a scalar
Boolean indicator, and every $r_j$ in an assertion means $s(r_j)$.
Substitution into Step~2 evaluates every boundary:
\[
 \begin{array}{c|l|l}
 \text{boundary}&\text{assertion}&\text{backward step}\\\hline
 A_8&[r_0=r_1=r_2=r_3]I&\text{postcondition}\\
 A_7&[r_0=r_1=r_2]Q_{r_0}^{(3)}&\text{\eqref{wp:measure} for }r_3\\
 A_6&[r_0=r_1]Q_{r_0}^{(2)}Q_{r_0}^{(3)}&\text{\eqref{wp:measure} for }r_2\\
 A_5&Q_{r_0}^{(1)}Q_{r_0}^{(2)}Q_{r_0}^{(3)}&\text{\eqref{wp:measure} for }r_1\\
 A_4&\sum_x Q_x^{(0)}Q_x^{(1)}Q_x^{(2)}Q_x^{(3)}&\text{\eqref{wp:measure} for }r_0\\
 A_3&\sum_x Q_x^{(0)}Q_x^{(1)}Q_x^{(2)}Q_0^{(3)}&C_{23}\text{ conjugation}\\
 A_2&\sum_x Q_x^{(0)}Q_x^{(1)}Q_0^{(2)}Q_0^{(3)}&C_{12}\text{ conjugation}\\
 A_1&Q_0^{(1)}Q_0^{(2)}Q_0^{(3)}&C_{01}\text{ conjugation}\\
 A_0&A_1&H_0\text{ conjugation}
 \end{array}
\]
In particular, $A_1=I_2\otimes|000\rangle\langle000|$.
For example, $A_4$ projects onto $\operatorname{span}\{|0000\rangle,|1111\rangle\}$.
Conjugating by $C_{23}$ changes this span to
$\operatorname{span}\{|0000\rangle,|1110\rangle\}$, which is $A_3$.
The next two CNOTs similarly give the displayed $A_2$ and $A_1$.
Finally $H_0$ acts on the identity factor of $A_1$, so it leaves the assertion unchanged.

Put $Q_{0000}=|0000\rangle\langle0000|$.
The difference $A_0-Q_{0000}=|1000\rangle\langle1000|$ is positive.
Applying \ref{rule:consequence} to the triple already derived in Step~5 gives
\[
 \frac{Q_{0000}\preceq A_0\qquad
       \vdash\htriple{A_0}{d}{[r_0=r_1=r_2=r_3]}}
      {\vdash\htriple{Q_{0000}}{d}{[r_0=r_1=r_2=r_3]}}.
\]
Here the displayed assertions are lifted by $\pi$ implicitly.
On input $\rho=Q_{0000}$, duality gives a lower bound of one for the event's
probability. Trace preservation gives the matching upper bound.

For the separate events $E_x\equiv(r_0=r_1=r_2=r_3=x)$, the same backward
calculation retains just the $x$ summand at $A_4$ and gives
\[
 \WP(d,[E_x]\circ\pi)=\Pi_x^{X_0}Q_0^{(1)}Q_0^{(2)}Q_0^{(3)},\qquad
 \Pr_d[E_x]=\tr\!\left(Q_{0000}\Pi_x^{X_0}Q_0^{(1)}Q_0^{(2)}Q_0^{(3)}\right)
           =\tfrac12.
\]
Thus the compiled program returns $0000$ or $1111$, each with probability $1/2$.
The probabilities follow from exact weakest preconditions on the specified input.

\paragraph{Step 7: reuse the derivation for the benchmark's fusion.}
The certified fusion for this benchmark combines rows 3 and 4. Both have
$P_j=Z_0$. For every public suffix assertion $F$, Step~4 and \eqref{wp:seq} give
\[
 \begin{aligned}
 \WP(d_3\seq d_4,F\circ\pi)
   &=\bigl((U(Z_0)^2)^\dagger F U(Z_0)^2\bigr)\circ\pi\\
   &=(Z_0FZ_0)\circ\pi
    =\WP(\cmd{correct}\ Z_0,F\circ\pi).
 \end{aligned}
\]
The middle equality uses $U(Z_0)^2=Z_0$.
Rule~\ref{rule:primitive} supplies the replacement correction's triple.
Rule~\ref{rule:seq} joins it to the unchanged prefix and suffix.
Every suffix precondition is public by Steps~2 and 4, so the private-bit values
left by the two programs cannot affect this composition.
The fused target therefore satisfies Equation~\eqref{eq:cat-goal} and the same
concrete readout guarantees. Replacing the two scopes removes four internal
measurements, reducing the count from 28 to 24.

\section{Grover search: a complete compilation proof}
\label{app:grover-case-study}

The QASMBench input \texttt{grover\_n2} implements Grover's search algorithm on
four items~\cite{grover1996search,li2022qasmbench}. We prove that its compiled
measurement program returns the marked item with probability one on the prescribed
input. We also prove equality of source and target weakest preconditions for
every public postcondition. The proof follows all 18 original source instructions.

\paragraph{Step 1: identify the search problem and source.}
The quantum inputs are $\Gamma=(q_0,q_1)$ and the public bits are
$\Delta=(r_0,r_1)$. Basis states are written $|q_0q_1\rangle$.
Figure~\ref{fig:grover-case-study} shows the entire source.
Writing $C=\mathrm{CNOT}_{01}$ for the source matrix, its instruction list is
\[
 \begin{aligned}
 c={}&\underbrace{H_0\seq H_1}_{\text{preparation}}\seq
       \underbrace{H_1\seq C\seq H_1}_{\text{oracle}}\seq\\
     &\underbrace{H_0\seq H_1\seq X_0\seq X_1\seq H_1\seq C\seq H_1
             \seq X_0\seq X_1\seq H_0\seq H_1}_{\text{diffusion}}\seq\\
     &M_{Z_0}(r_0)\seq M_{Z_1}(r_1).
 \end{aligned}
\]
The adjacent $H_1$ instructions at positions 2 and 3 are retained.
Let $g_i$ be source instruction $i$, and let $A_i$ be the assertion after it.
The initial assertion is $A_0$.

\begin{figure}[htbp]
\centering
\begin{tikzpicture}[x=.62cm,y=1cm,font=\scriptsize,>=Latex,
 gate/.style={draw,fill=white,minimum width=.34cm,minimum height=.34cm,inner sep=1pt}]
 \foreach \y/\j in {0/0,-1/1}{
   \draw (0,\y) -- (19,\y);
   \node[anchor=east] at (-.15,\y) {$q_\j$};
 }
 \foreach \x/\y in {1/0,2/-1,3/-1,5/-1,6/0,7/-1,10/-1,12/-1,15/0,16/-1}{
   \node[gate] at (\x,\y) {$H$};
 }
 \foreach \x/\y in {8/0,9/-1,13/0,14/-1}{
   \node[gate,pauli/X] at (\x,\y) {$X$};
 }
 \foreach \x in {4,11}{
   \draw (\x,0) -- (\x,-1);
   \fill (\x,0) circle (1.6pt);
   \draw (\x,-1) circle (.12cm);
   \draw (\x,-1.12) -- (\x,-.88);
 }
 \pic at (17,0) {ppmone={Z/r_0}};
 \pic at (18,-1) {ppmone={Z/r_1}};
 \foreach \j in {1,...,18}{
   \node[text=black!65] at (\j,.42) {\j};
 }
 \foreach \a/\b/\mid/\name in {.5/2.5/1.5/preparation,2.5/5.5/4/oracle,
                                 5.5/16.5/11/diffusion,16.5/18.5/17.5/readout}{
   \draw[gray] (\a,.72) -- (\a,.81) -- (\b,.81) -- (\b,.72);
   \node[anchor=south,inner sep=2pt] at (\mid,.84) {\name};
 }
 \draw[gray] (.5,-1.94) -- (18.5,-1.94);
 \foreach \x/\j in {.5/0,2.5/2,5.5/5,16.5/16}{
   \draw[gray] (\x,-1.84) -- (\x,-2.04);
   \node[anchor=north,inner sep=2pt] at (\x,-2.08) {$A_\j$};
 }
 \draw[gray] (18.5,-1.84) -- (18.5,-2.04);
 \node[anchor=north,inner sep=2pt] at (18.5,-2.08) {$A_{18}=B$};
\end{tikzpicture}
\caption{The complete 18-step \texttt{grover\_n2} source. Steps 1--2 prepare the
search state. Steps 3--5 implement the oracle that marks $|11\rangle$ in the
$|q_0q_1\rangle$ basis. Steps 6--16 perform diffusion. Steps 17--18 measure $Z$
and record the public bits $r_0,r_1$. The adjacent $H$ gates on $q_1$ are retained
in source order. The bottom strip places full-register assertions before preparation
and after each algorithm stage. The final assertion is $A_{18}=B$.}
\Description{Two quantum wires carry all eighteen source operations in order.
The preparation has two Hadamards. The oracle has a Hadamard on qubit one, a CNOT
from qubit zero to one, and another Hadamard on qubit one. Diffusion has Hadamards
on both qubits, X on both qubits, a Hadamard on qubit one, a CNOT, another Hadamard
on qubit one, X on both qubits, and Hadamards on both qubits. Separate blue Z
meters then write r zero and r one. The assertion strip marks boundaries zero,
two, five, sixteen, and eighteen.}
\label{fig:grover-case-study}
\end{figure}
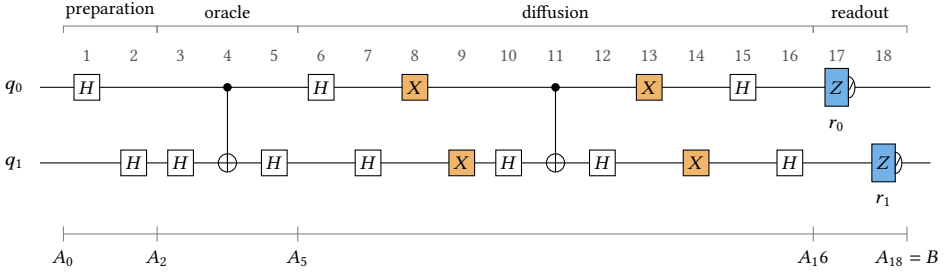

Put $|s\rangle=|++\rangle=(|00\rangle+|01\rangle+|10\rangle+|11\rangle)/2$.
The source's oracle matrix is
\[
 O=H_1CH_1=\operatorname{diag}(1,1,1,-1)
   =I-2|11\rangle\langle11|.
\]
It marks item $11$. The diffusion block has matrix
\[
 D=(H\otimes H)(X\otimes X)O(X\otimes X)(H\otimes H)
   =I-2|s\rangle\langle s|.
\]
This is the usual diffusion matrix up to a global sign.
In particular, $O|s\rangle=|s\rangle-|11\rangle$ and
$\langle s|O|s\rangle=1/2$, so $DO|s\rangle=-|11\rangle$.
The Hoare derivation below establishes this search guarantee for the compiled
Pauli-measurement program.

\paragraph{Step 2: pull back an arbitrary postcondition.}
Let $B$ be any Hermitian assertion on the two data qubits and public store.
Put $A_{18}=B$ and $Q_x^{(j)}=\Pi_x^{Z_j}$.
The two readout instances of \eqref{wp:measure} give
\[
 \begin{aligned}
 A_{17}(s)&=\sum_y Q_y^{(1)}B(s[r_1\mapsto y])Q_y^{(1)},\\
 A_{16}(s)&=\sum_{x,y}Q_x^{(0)}Q_y^{(1)}
        B(s[r_0\mapsto x,r_1\mapsto y])Q_y^{(1)}Q_x^{(0)}.
 \end{aligned}
\]
Here $x,y\in\{0,1\}$. Each measurement substitutes the bit it records before
evaluating the postcondition. For all sixteen preceding gates, calculate
\begin{equation}
 A_{i-1}=g_i^\dagger A_i g_i,\qquad i=16,15,\ldots,1.
 \label{eq:grover-source-wp}
\end{equation}
Table~\ref{tab:grover-lowering} lists every $g_i$ in this recursion.
Thus $A_0=W_c(B)$. We now derive a target triple at each of these boundaries.

\paragraph{Step 3: expand the target into Pauli measurements.}
The target uses distinct private bits $a,b$, disjoint from $\Delta$.
Let $\pi$ forget them. Its 36 ancilla scopes are the following finite commands,
with the products $P_j$ specified in Table~\ref{tab:grover-lowering}:
\begin{equation}
 \begin{aligned}
 d_j={}&\cmd{let}\ \nu=\cmd{ancilla}\ |0\rangle\ \cmd{in}\bigl(
       M_{X_\nu P_j}(a)\seq M_{Y_\nu}(b)\seq\\
     &\qquad\cmd{if}\ a\oplus b\ \cmd{then}\ \cmd{correct}\ P_j
                \ \cmd{else}\ \cmd{skip}\bigr).
 \end{aligned}
 \label{eq:grover-scope}
\end{equation}
The body context is $(\nu,q_0,q_1)$. Each scope prepares and discards a fresh
zero qubit. Products $P_j$ act as identity on $\nu$.
For $P_j=-Z_0X_1$, the first measurement is of $-X_\nu Z_0X_1$.
The sign of its conditional Pauli correction cancels in conjugation.

\begin{table}[htbp]
 \centering\small
 \begin{tabular}{@{}clll@{}}
 \toprule
 Source step $i$ & $g_i$ & Target command $c_i$ & Ordered products $P_j$\\
 \midrule
 1 & $H_0$ & $d_1\seq d_2\seq d_3$ & $(Z_0,X_0,Z_0)$\\
 2 & $H_1$ & $d_4\seq d_5\seq d_6$ & $(Z_1,X_1,Z_1)$\\
 3 & $H_1$ & $d_7\seq d_8\seq d_9$ & $(Z_1,X_1,Z_1)$\\
 4 & $C$ & $d_{10}\seq d_{11}\seq d_{12}$ & $(Z_0,X_1,-Z_0X_1)$\\
 5 & $H_1$ & $d_{13}\seq d_{14}\seq d_{15}$ & $(Z_1,X_1,Z_1)$\\
 6 & $H_0$ & $d_{16}\seq d_{17}\seq d_{18}$ & $(Z_0,X_0,Z_0)$\\
 7 & $H_1$ & $d_{19}\seq d_{20}\seq d_{21}$ & $(Z_1,X_1,Z_1)$\\
 8 & $X_0$ & $\cmd{correct}\ X_0$ & ---\\
 9 & $X_1$ & $\cmd{correct}\ X_1$ & ---\\
 10 & $H_1$ & $d_{22}\seq d_{23}\seq d_{24}$ & $(Z_1,X_1,Z_1)$\\
 11 & $C$ & $d_{25}\seq d_{26}\seq d_{27}$ & $(Z_0,X_1,-Z_0X_1)$\\
 12 & $H_1$ & $d_{28}\seq d_{29}\seq d_{30}$ & $(Z_1,X_1,Z_1)$\\
 13 & $X_0$ & $\cmd{correct}\ X_0$ & ---\\
 14 & $X_1$ & $\cmd{correct}\ X_1$ & ---\\
 15 & $H_0$ & $d_{31}\seq d_{32}\seq d_{33}$ & $(Z_0,X_0,Z_0)$\\
 16 & $H_1$ & $d_{34}\seq d_{35}\seq d_{36}$ & $(Z_1,X_1,Z_1)$\\
 17 & $M_{Z_0}(r_0)$ & $M_{Z_0}(r_0)$ & ---\\
 18 & $M_{Z_1}(r_1)$ & $M_{Z_1}(r_1)$ & ---\\
 \bottomrule
 \end{tabular}
 \caption{Complete lowering of the original Grover benchmark.
 Each listed $d_j$ is the explicit Pauli-measurement scope in
 Equation~\eqref{eq:grover-scope}. The four direct corrections retain their
 positions among the scopes.}
 \label{tab:grover-lowering}
\end{table}

The target is $d=c_1\seq c_2\seq\cdots\seq c_{18}$.
There are ten Hadamards, two CNOTs, four direct corrections, and two readouts.
The twelve non-Pauli gates expand into 36 scopes, giving
$36\cdot2+2=74$ measurements. This agrees with the original-source certificate.
The target needs one temporary ancilla at a time and no magic states.
The depth experiment in Table~\ref{tab:qasmbench} starts from the shorter
Qiskit output. Here we follow the original typed source, whose parsing remains
outside the Hoare logic.

Figures~\ref{fig:grover-ppm-a}--\ref{fig:grover-ppm-c} draw the full compiled
circuit in source order. Each source step $i$ connects $A_{i-1}\circ\pi$ to
$A_i\circ\pi$. The two internal boundaries of each three-scope row carry
$F_1\circ\pi$ and $F_2\circ\pi$ from Step~4.
The four direct corrections and the two public readouts retain their original
positions. The diagrams therefore display the same eighteen target triples
that Step~4 composes.

\clearpage
\begin{figure}[p]
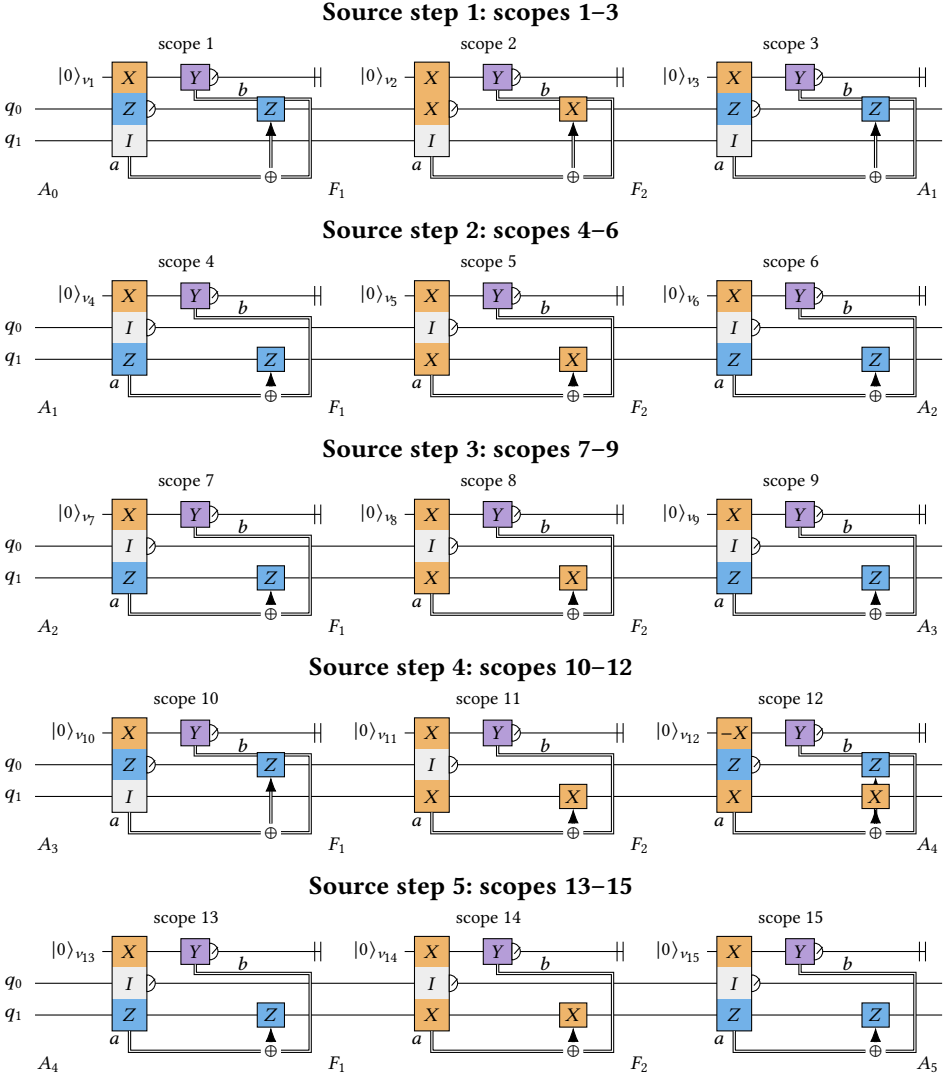

\centering
\textbf{Source step 1: scopes 1--3}\par\smallskip
\begin{ppmcompiledrow}{2}{A_0}{A_1}
 \ppmscope{0}{1}{}{Z,I}
 \ppmscope{4}{2}{}{X,I}
 \ppmscope{8}{3}{}{Z,I}
 \node at (4,\ppmcompiledassert) {$F_1$};
 \node at (8,\ppmcompiledassert) {$F_2$};
\end{ppmcompiledrow}\par\medskip
\textbf{Source step 2: scopes 4--6}\par\smallskip
\begin{ppmcompiledrow}{2}{A_1}{A_2}
 \ppmscope{0}{4}{}{I,Z}
 \ppmscope{4}{5}{}{I,X}
 \ppmscope{8}{6}{}{I,Z}
 \node at (4,\ppmcompiledassert) {$F_1$};
 \node at (8,\ppmcompiledassert) {$F_2$};
\end{ppmcompiledrow}\par\medskip
\textbf{Source step 3: scopes 7--9}\par\smallskip
\begin{ppmcompiledrow}{2}{A_2}{A_3}
 \ppmscope{0}{7}{}{I,Z}
 \ppmscope{4}{8}{}{I,X}
 \ppmscope{8}{9}{}{I,Z}
 \node at (4,\ppmcompiledassert) {$F_1$};
 \node at (8,\ppmcompiledassert) {$F_2$};
\end{ppmcompiledrow}\par\medskip
\textbf{Source step 4: scopes 10--12}\par\smallskip
\begin{ppmcompiledrow}{2}{A_3}{A_4}
 \ppmscope{0}{10}{}{Z,I}
 \ppmscope{4}{11}{}{I,X}
 \ppmscope{8}{12}{-}{Z,X}
 \node at (4,\ppmcompiledassert) {$F_1$};
 \node at (8,\ppmcompiledassert) {$F_2$};
\end{ppmcompiledrow}\par\medskip
\textbf{Source step 5: scopes 13--15}\par\smallskip
\begin{ppmcompiledrow}{2}{A_4}{A_5}
 \ppmscope{0}{13}{}{I,Z}
 \ppmscope{4}{14}{}{I,X}
 \ppmscope{8}{15}{}{I,Z}
 \node at (4,\ppmcompiledassert) {$F_1$};
 \node at (8,\ppmcompiledassert) {$F_2$};
\end{ppmcompiledrow}
\caption{The compiled Grover circuit, part 1 of 3: preparation and oracle.
Read rows from left to right and top to bottom. The two data qubits persist
between rows. Every numbered scope prepares a fresh $|0\rangle_\nu$,
measures the displayed joint Pauli product into $a$, measures $Y_\nu$ into $b$,
and applies the unsigned correction product exactly when $a\oplus b=1$.
Double end bars discard the ancilla. Each following scope reuses the private
locations $a,b$. The minus sign in scope 12 is part of the measured observable.
The row assertions are lifted by $\pi$. They are the boundaries proved in
Step~4. The labels $F_1,F_2$ are local to each row.
Continue at $A_5$ in Figure~\ref{fig:grover-ppm-b}.}
\Description{Five consecutive rows expand the first five Grover source
instructions into fifteen explicit Pauli-measurement scopes. Each scope shows
the joint meter, ancilla Y meter, parity feedback, and ancilla disposal.
The source preparation and oracle have no remaining Hadamard or CNOT boxes.
Assertions run from A zero to A five.}
\label{fig:grover-ppm-a}
\end{figure}

\begin{figure}[p]
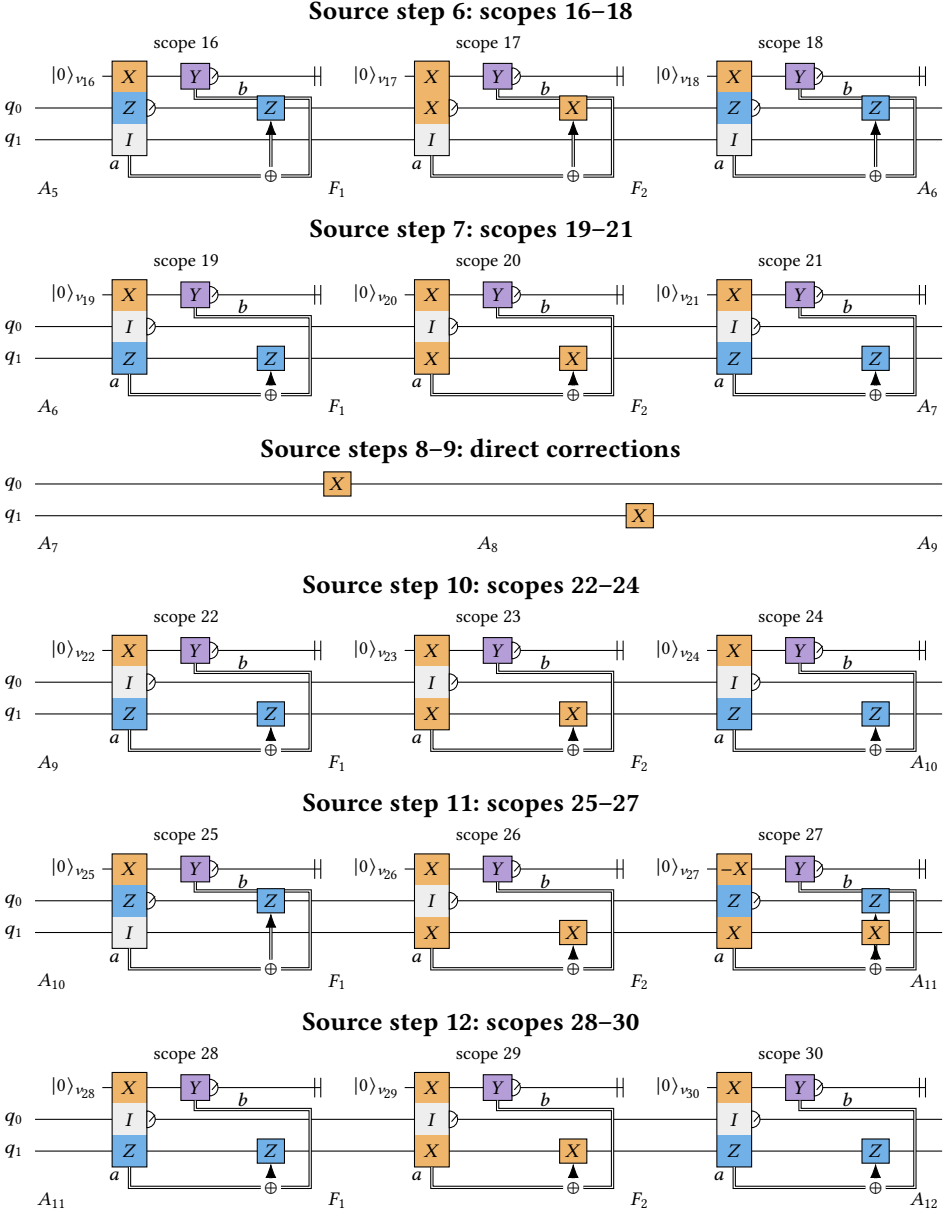

\centering
\textbf{Source step 6: scopes 16--18}\par\smallskip
\begin{ppmcompiledrow}{2}{A_5}{A_6}
 \ppmscope{0}{16}{}{Z,I}
 \ppmscope{4}{17}{}{X,I}
 \ppmscope{8}{18}{}{Z,I}
 \node at (4,\ppmcompiledassert) {$F_1$};
 \node at (8,\ppmcompiledassert) {$F_2$};
\end{ppmcompiledrow}\par\medskip
\textbf{Source step 7: scopes 19--21}\par\smallskip
\begin{ppmcompiledrow}{2}{A_6}{A_7}
 \ppmscope{0}{19}{}{I,Z}
 \ppmscope{4}{20}{}{I,X}
 \ppmscope{8}{21}{}{I,Z}
 \node at (4,\ppmcompiledassert) {$F_1$};
 \node at (8,\ppmcompiledassert) {$F_2$};
\end{ppmcompiledrow}\par\medskip
\textbf{Source steps 8--9: direct corrections}\par\smallskip
\begin{ppmcompiledrow}[.38]{2}{A_7}{A_9}
 \ppmcorrection{4}{X,I}
 \ppmcorrection{8}{I,X}
 \node at (6,\ppmcompiledassert) {$A_8$};
\end{ppmcompiledrow}\par\medskip
\textbf{Source step 10: scopes 22--24}\par\smallskip
\begin{ppmcompiledrow}{2}{A_9}{A_{10}}
 \ppmscope{0}{22}{}{I,Z}
 \ppmscope{4}{23}{}{I,X}
 \ppmscope{8}{24}{}{I,Z}
 \node at (4,\ppmcompiledassert) {$F_1$};
 \node at (8,\ppmcompiledassert) {$F_2$};
\end{ppmcompiledrow}\par\medskip
\textbf{Source step 11: scopes 25--27}\par\smallskip
\begin{ppmcompiledrow}{2}{A_{10}}{A_{11}}
 \ppmscope{0}{25}{}{Z,I}
 \ppmscope{4}{26}{}{I,X}
 \ppmscope{8}{27}{-}{Z,X}
 \node at (4,\ppmcompiledassert) {$F_1$};
 \node at (8,\ppmcompiledassert) {$F_2$};
\end{ppmcompiledrow}\par\medskip
\textbf{Source step 12: scopes 28--30}\par\smallskip
\begin{ppmcompiledrow}{2}{A_{11}}{A_{12}}
 \ppmscope{0}{28}{}{I,Z}
 \ppmscope{4}{29}{}{I,X}
 \ppmscope{8}{30}{}{I,Z}
 \node at (4,\ppmcompiledassert) {$F_1$};
 \node at (8,\ppmcompiledassert) {$F_2$};
\end{ppmcompiledrow}
\caption{The compiled Grover circuit, part 2 of 3: the first diffusion steps.
Continue from Figure~\ref{fig:grover-ppm-a} at $A_5$.
The first direct $X_0$ and $X_1$ corrections occur after scope 21 and before
scope 22. These two boxes have no classical control.
Their intermediate assertion is $A_8\circ\pi$.
All other correction columns retain their local parity controls.
The signed joint meter in scope 27 measures $-X_\nu Z_0X_1$.
Continue at $A_{12}$ in Figure~\ref{fig:grover-ppm-c}.}
\Description{Six consecutive rows show scopes sixteen through thirty and the
two direct X corrections between scopes twenty-one and twenty-two.
Every measurement scope includes its fresh ancilla and feedback control.
The final signed joint measurement in this part is scope twenty-seven.
The proof boundaries run from A five to A twelve.}
\label{fig:grover-ppm-b}
\end{figure}

\begin{figure}[p]
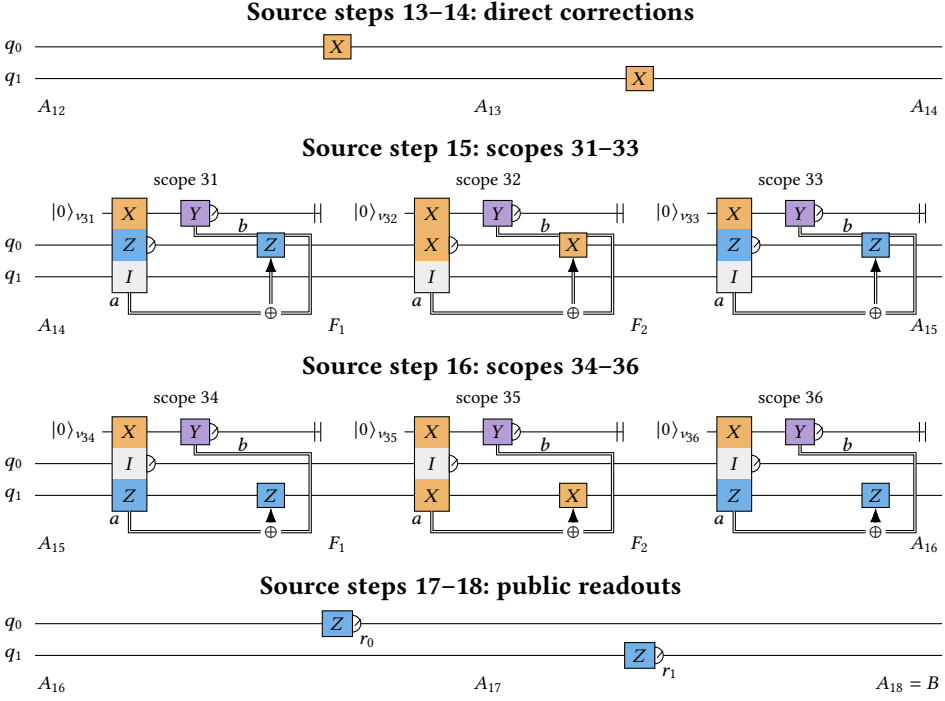

\centering
\textbf{Source steps 13--14: direct corrections}\par\smallskip
\begin{ppmcompiledrow}[.38]{2}{A_{12}}{A_{14}}
 \ppmcorrection{4}{X,I}
 \ppmcorrection{8}{I,X}
 \node at (6,\ppmcompiledassert) {$A_{13}$};
\end{ppmcompiledrow}\par\medskip
\textbf{Source step 15: scopes 31--33}\par\smallskip
\begin{ppmcompiledrow}{2}{A_{14}}{A_{15}}
 \ppmscope{0}{31}{}{Z,I}
 \ppmscope{4}{32}{}{X,I}
 \ppmscope{8}{33}{}{Z,I}
 \node at (4,\ppmcompiledassert) {$F_1$};
 \node at (8,\ppmcompiledassert) {$F_2$};
\end{ppmcompiledrow}\par\medskip
\textbf{Source step 16: scopes 34--36}\par\smallskip
\begin{ppmcompiledrow}{2}{A_{15}}{A_{16}}
 \ppmscope{0}{34}{}{I,Z}
 \ppmscope{4}{35}{}{I,X}
 \ppmscope{8}{36}{}{I,Z}
 \node at (4,\ppmcompiledassert) {$F_1$};
 \node at (8,\ppmcompiledassert) {$F_2$};
\end{ppmcompiledrow}\par\medskip
\textbf{Source steps 17--18: public readouts}\par\smallskip
\begin{ppmcompiledrow}[.38]{2}{A_{16}}{A_{18}=B}
 \ppmreadout{4}{0}{Z}{r_0}
 \ppmreadout{8}{1}{Z}{r_1}
 \node at (6,\ppmcompiledassert) {$A_{17}$};
\end{ppmcompiledrow}
\caption{The compiled Grover circuit, part 3 of 3: final diffusion steps and
readout. Continue from Figure~\ref{fig:grover-ppm-b} at $A_{12}$.
The remaining two direct corrections occur before scope 31.
Their intermediate assertion is $A_{13}\circ\pi$.
The final meters record $Z_0$ in $r_0$, then $Z_1$ in $r_1$.
Their intermediate assertion is $A_{17}\circ\pi$.
The three figures contain all 36 scopes, four direct corrections, and two
public readouts. Step~4 joins their triples to prove compilation correctness.
Step~5 instantiates $B=[r_0=r_1=1]$ and derives the input assertion
$|00\rangle\langle00|$.}
\Description{Four consecutive rows show two direct X corrections, six remaining
measurement scopes, and two final public Z readouts. The direct corrections
precede scope thirty-one. Scope thirty-six is followed by measurements into
bits r zero and r one. Assertions connect A twelve to A eighteen equal to B.}
\label{fig:grover-ppm-c}
\end{figure}

\clearpage

\paragraph{Step 4: derive each target boundary by the Hoare rules.}
The scope calculation in Appendix~\ref{app:cat-case-study}, Step~4, uses
$P^2=I$ and the Pauli projector laws. It applies to this two-qubit data register.
For a public suffix assertion $F$, the conditional is proved by
\ref{rule:primitive}, \ref{rule:skip}, and \ref{rule:if}.
Two measurement instances of \ref{rule:primitive} and two uses of
\ref{rule:seq} prove the scope body. Rule~\ref{rule:ancilla} compresses its
precondition by the zero-ancilla inclusion $E_0$.
The resulting weakest-precondition calculation is
\[
 \begin{aligned}
 \WP(d_j,F\circ\pi)(\sigma)
 &=\sum_{x,y,z}K_{xyz}^\dagger F(\pi(\sigma))K_{xyz}\\
 &=U(P_j)^\dagger F(\pi(\sigma))U(P_j),\\
 K_{xyz}&=E_z^\dagger(I_2\otimes P_j^{x\oplus y})
                 \Pi_y^{Y_\nu}\Pi_x^{X_\nu P_j}E_0.
 \end{aligned}
\]
As calculated there, $K_{xyz}=\alpha_{xyz}U(P_j)$ and all eight
$|\alpha_{xyz}|^2$ equal $1/8$. The two measurements overwrite $a,b$ before
feedback reads them. Therefore this derived scope triple holds for every
initial value of the private bits.

For a row containing three scopes, write their products as $P,Q,R$ in execution
order and calculate the three intermediate assertions
\[
 F_3=A_i,\qquad F_2=U(R)^\dagger F_3U(R),\qquad
 F_1=U(Q)^\dagger F_2U(Q),\qquad F_0=U(P)^\dagger F_1U(P).
\]
The scope rule derivation gives the three triples with matching boundaries
$F_0,F_1,F_2,F_3$, all lifted by $\pi$. Two uses of \ref{rule:seq} join them.
Their precondition is
$F_0=(U(R)U(Q)U(P))^\dagger A_i(U(R)U(Q)U(P))$.
For every Hadamard row and both CNOT rows in the table, respectively,
\[
 \begin{aligned}
 U(Z_j)U(X_j)U(Z_j)&=e^{i\pi/4}H_j,\qquad j=0,1,\\
 U(-Z_0X_1)U(X_1)U(Z_0)&=iC.
 \end{aligned}
\]
The phases cancel, so $F_0=A_{i-1}$ from
Equation~\eqref{eq:grover-source-wp}.
For the four direct corrections, \eqref{wp:correct} and \ref{rule:primitive}
give the same equality and triple without an ancilla.
For the last two rows, Step~2 and \ref{rule:primitive} give
\[
 \vdash\dtriple{A_{16}\circ\pi}{\pmeter[Z]{Z_0}{r_0}}{A_{17}\circ\pi},
 \qquad
 \vdash\dtriple{A_{17}\circ\pi}{\pmeter[Z]{Z_1}{r_1}}{B\circ\pi}.
\]
We have now derived $\vdash\triple{A_{i-1}\circ\pi}{c_i}{A_i\circ\pi}$
for every row $i=1,\ldots,18$.
Seventeen uses of \ref{rule:seq} join these row triples. The corresponding
uses of \eqref{wp:seq} compose their exact equalities. Thus
\begin{equation}
 \WP(d,B\circ\pi)=W_c(B)\circ\pi,
 \qquad\vdash\triple{W_c(B)\circ\pi}{d}{B\circ\pi}.
 \label{eq:grover-compile}
\end{equation}
This proves preservation of every public postcondition, including assertions
that relate the measured bits to the quantum output.

\paragraph{Step 5: evaluate every precondition for successful search.}
Let $E(s)$ mean $s(r_0)=s(r_1)=1$, and take $B=[E]$.
The last two steps become
\[
 A_{18}=[E],\qquad
 A_{17}(s)=[s(r_0)=1]Q_1^{(1)},\qquad
 A_{16}=Q_1^{(0)}Q_1^{(1)}=|11\rangle\langle11|.
\]
The first equality for a measurement follows by selecting $y=1$ in the
$A_{17}$ sum. The next selects $x=1$. These are the two
\eqref{wp:measure} steps, including their classical substitutions.

Every earlier assertion is a rank-one projector.
Define $|\pm\rangle=(|0\rangle\pm|1\rangle)/\sqrt2$ and
\[
 \begin{aligned}
 |\Phi^+\rangle&=(|00\rangle+|11\rangle)/\sqrt2,&
 |\Psi^-\rangle&=(|01\rangle-|10\rangle)/\sqrt2,\\
 |\phi\rangle&=(|00\rangle+|01\rangle+|10\rangle-|11\rangle)/2,\\
 |\psi\rangle&=(|00\rangle-|01\rangle+|10\rangle+|11\rangle)/2,\\
 |\chi\rangle&=(|00\rangle-|01\rangle-|10\rangle-|11\rangle)/2.
 \end{aligned}
\]
Table~\ref{tab:grover-wp} evaluates all sixteen conjugations explicitly.
An entry $|v_i\rangle$ denotes $A_i=|v_i\rangle\langle v_i|$.
Each row is checked by $g_i^\dagger|v_i\rangle\langle v_i|g_i
=|v_{i-1}\rangle\langle v_{i-1}|$.
Vector phases may differ between rows because they cancel in these projectors.

\begin{table}[htbp]
 \centering\small
 \begin{tabular}{@{}clll@{}}
 \toprule
 Step $i$ & Gate $g_i$ & Postcondition vector $|v_i\rangle$
                                  & Precondition vector $|v_{i-1}\rangle$\\
 \midrule
 16 & $H_1$ & $|11\rangle$ & $|1-\rangle$\\
 15 & $H_0$ & $|1-\rangle$ & $|--\rangle$\\
 14 & $X_1$ & $|--\rangle$ & $|--\rangle$\\
 13 & $X_0$ & $|--\rangle$ & $|--\rangle$\\
 12 & $H_1$ & $|--\rangle$ & $|-1\rangle$\\
 11 & $C$ & $|-1\rangle$ & $|\Psi^-\rangle$\\
 10 & $H_1$ & $|\Psi^-\rangle$ & $|\chi\rangle$\\
 9 & $X_1$ & $|\chi\rangle$ & $|\psi\rangle$\\
 8 & $X_0$ & $|\psi\rangle$ & $|\phi\rangle$\\
 7 & $H_1$ & $|\phi\rangle$ & $|\Phi^+\rangle$\\
 6 & $H_0$ & $|\Phi^+\rangle$ & $|\phi\rangle$\\
 5 & $H_1$ & $|\phi\rangle$ & $|\Phi^+\rangle$\\
 4 & $C$ & $|\Phi^+\rangle$ & $|+0\rangle$\\
 3 & $H_1$ & $|+0\rangle$ & $|++\rangle$\\
 2 & $H_1$ & $|++\rangle$ & $|+0\rangle$\\
 1 & $H_0$ & $|+0\rangle$ & $|00\rangle$\\
 \bottomrule
 \end{tabular}
 \caption{The complete backward proof of Grover's success event.
 Read from top to bottom. Every entry is a normalized vector specifying the
 displayed boundary's projector. Step~4 supplies the corresponding target
 Hoare triple for each source gate.}
 \label{tab:grover-wp}
\end{table}

For example, pulling $|\phi\rangle$ back through gate 5 uses
$H_1|\phi\rangle=|\Phi^+\rangle$. Pulling that projector through gate 4
uses $C|\Phi^+\rangle=|+0\rangle$.
For gate 9, $X_1|\chi\rangle=-|\psi\rangle$, so the negative sign cancels.
These calculations illustrate the Hadamard, CNOT, and direct-correction rows.

The final precondition is exactly $A_0=|00\rangle\langle00|$.
Instantiating Equation~\eqref{eq:grover-compile} gives the derived judgment
\[
 \vdash\htriple{|00\rangle\langle00|}{d}{[r_0=r_1=1]}.
\]
The constant precondition and the event are lifted to the full store implicitly.
For the caller-supplied input $\rho=|00\rangle\langle00|$, forward/backward
duality gives
\[
 \Pr_d[r_0=r_1=1]
 =\tr\!\left(\rho\,\WP(d,[E]\circ\pi)\right)
 =\tr(\rho^2)=1.
\]
Thus the complete measurement implementation finds the marked item with
certainty. Equation~\eqref{eq:grover-compile} additionally preserves every
public observation for arbitrary input states and stores.

\section{A ten-qubit Ising circuit: a complete derivation}
\label{app:ising-case-study}

The QASMBench input \texttt{ising\_n10} has 490 instructions: 110 Hadamards,
90 CNOTs, 280 rotations, and ten final measurements. We use it to derive a guarantee
for a concrete observation: whether the first two output bits agree. This example
uses the same measurement, sequencing, approximation, and resource rules as the
small examples. Its size changes the number of compositions.

\paragraph{Step 1: identify the program and its interface.}
The data inputs are $\Gamma=(q_0,\ldots,q_9)$ and the public locations are
$\Delta=(r_0,\ldots,r_9)$. Figure~\ref{fig:ising-case-study} groups the finite source
into five consecutive blocks $L_1,\ldots,L_5$. Each contains nine two-qubit
interactions and ten one-qubit updates. The notation $R_\theta=\operatorname{diag}(1,e^{i\theta})$
represents the source's $R_z(\theta)$ under conjugation. Their global phase has no
observable effect. All angles are exact rational radians.

\begin{figure}[htbp]
\centering
\begin{tikzpicture}[x=.96cm,y=.78cm,font=\small,>=Latex,
 gate/.style={draw,fill=white,minimum height=.46cm,minimum width=.65cm,inner sep=3pt},
 phase/.style={gate,fill=ppmz,minimum width=1.02cm}]
 \node[anchor=east] at (0,0) {$q_0,\ldots,q_9$};
 \draw (0,0) -- (10.4,0);
 \draw (.14,-.12) -- (.32,.12);
 \node[font=\scriptsize,anchor=south] at (.22,.15) {$10$};
 \node[gate] at (1,0) {$H^{\otimes10}$};
 \foreach \x/\k in {2.5/1,3.8/2,5.1/3,6.4/4,7.7/5}{
   \node[gate] at (\x,0) {$L_\k$};
   \node[font=\scriptsize] at (\x,-.57) {$94$};
 }
 \node[gate,fill=ppmz] at (9.15,0) {$M_Z^{\otimes10}$};
 \node[anchor=west,fill=ppmclassical,inner sep=3pt] at (10.45,0) {$r_0,\ldots,r_9$};
 \node[anchor=west,font=\scriptsize] at (.3,.75) {(a) All 490 source instructions};
 \node[anchor=west,font=\scriptsize] at (.3,-1.3)
   {(b) The first six-instruction interaction, $J_{01}(-3/10)$};
 \foreach \y/\j in {-2.15/0,-3.25/1}{
   \node[anchor=east] at (0,\y) {$q_\j$};
   \draw (0,\y) -- (10.2,\y);
 }
 \node[phase] at (1,-2.15) {$R_{-3/10}$};
 \node[phase] at (2.7,-3.25) {$R_{3/10}$};
 \node[phase] at (4.2,-3.25) {$R_{3/10}$};
 \node[phase] at (7.4,-3.25) {$R_{-3/10}$};
 \foreach \x in {5.8,9}{
   \draw (\x,-2.15) -- (\x,-3.25);
   \fill (\x,-2.15) circle (.055);
   \draw[fill=ppmx] (\x,-3.25) circle (.18);
   \draw (\x-.18,-3.25) -- (\x+.18,-3.25) (\x,-3.43) -- (\x,-3.07);
 }
 \foreach \x/\e in {1/\eta,2.7/\eta,4.2/\eta,5.8/0,7.4/\eta,9/0}{
   \node[font=\scriptsize,fill=ppmclassical,inner sep=2pt] at (\x,-4.05) {$\e$};
 }
 \node[anchor=west,font=\scriptsize] at (10.1,-4.05) {allowance};
\end{tikzpicture}
\caption{The ten-qubit Ising case study. Each $L_\ell$ contains nine interactions
and ten local updates. The lower circuit expands the first interaction. All omitted
qubits are idle during it. Blue gates rotate about $Z$, and orange targets mark $X$.
Green annotations record the four synthesis allowances and two exact CNOT steps.}
\Description{A ten-qubit bus passes through initial Hadamards, five blocks of 94
instructions, and ten final Z measurements. The expanded first interaction has one
negative three-tenths phase on qubit zero, two positive three-tenths phases on qubit
one, a CNOT from zero to one, another negative phase on qubit one, and a second CNOT.
The four phases each contribute eta. The CNOTs contribute zero.}
\label{fig:ising-case-study}
\end{figure}
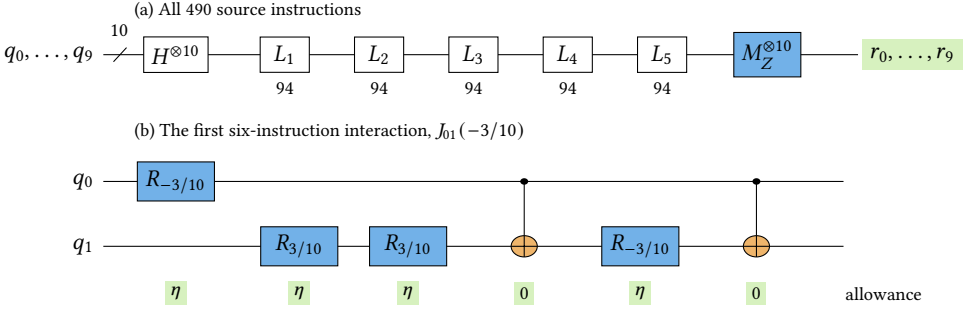

In each $L_\ell$, the interactions visit pairs $(0,1),(2,3),(4,5),(6,7),(8,9)$,
then $(1,2),(3,4),(5,6),(7,8)$. An interaction on $(i,j)$ with angle $\theta$ is
\[
 J_{ij}(\theta)=R_\theta(q_i)\seq R_{-\theta}(q_j)\seq R_{-\theta}(q_j)\seq
 \mathrm{CNOT}_{ij}\seq R_\theta(q_j)\seq\mathrm{CNOT}_{ij}.
\]
The subsequent update of $q_j$ is $H_j\seq R_{\beta_\ell}(q_j)\seq H_j\seq R_{\gamma_{\ell j}}(q_j)$,
with parameters from the source. Thus a block has $9\cdot6+10\cdot4=94$ instructions,
and the whole circuit has $10+5\cdot94+10=490$. These are abbreviations for finite
instruction lists, not loops added to QPPM.

\paragraph{Step 2: pull back a concrete output event.}
Let $E(s)\equiv s(r_0)=s(r_1)$, and put $B=[E]$. This is an effect: it is either
$I$ or zero. Write $Q_b^{(j)}=\Pi^{Z_{q_j}}_b$. Backward calculation first crosses
measurements into $r_9,\ldots,r_2$. None changes $E$, so each leaves $B$ unchanged:
\[
 \WP(M_{Z_{q_j}}(r_j),B)(s)
 =B(s)\sum_{b=0}^1 Q_b^{(j)}=B(s),\qquad 2\leq j\leq9.
\]
The remaining two measurements expose the quantum condition for agreement:
\[
 \begin{array}{c@{\qquad}l}
 \pmeter[Z]{Z_{q_1}}{r_1}
 &\displaystyle B_1(s)=\sum_b[s(r_0)=b]\,Q_b^{(1)},\\[10pt]
 \pmeter[Z]{Z_{q_0}}{r_0}
 &\displaystyle B_0=\sum_b Q_b^{(0)}Q_b^{(1)}Q_b^{(0)}.
 \end{array}
\]
Here each bracket is a Boolean indicator multiplying the matrix to its right.
For example, the $b=0$ term in the second line selects the branch in which both
measurements record zero. The two terms account for all successful branches.
Pulling $B_0$ backwards through $L_5,\ldots,L_1$ and the initial Hadamards defines
$A=W_c(B)$. This is an exact precondition for the chosen event in the source $c$.

\paragraph{Step 3: replace one interaction, with its allowance.}
The first interaction is $J_{01}(-3/10)$, shown in the figure. Its four phase
occurrences use only two checked words, $W_{-3/10}$ and $W_{3/10}$. The two CNOTs
are exact. For either sign and either wire, the local proof establishes, for every
effect $F$ on all ten qubits,
\[
 -\eta I\preceq R_{\pm3/10}^{\dagger}FR_{\pm3/10}
                  -W_{\pm3/10}^{\dagger}FW_{\pm3/10}\preceq\eta I,
 \qquad \eta=\frac{210000000001}{6250000000000000000}.
\]
Other tensor factors are implicit identities. In particular, $F$ need not separate
$q_0,q_1$ from the other eight qubits. Write $g_1,\ldots,g_6$ for this interaction's
six instructions, $\widehat g_i$ for their replacements, and set
$F_6=F$, $F_{i-1}=W_{g_i}(F_i)$. The backward assertions and allowance increments are
\[
 \begin{array}{c|cccccc}
 i&1&2&3&4&5&6\\\hline
 g_i&R_{-3/10}(q_0)&R_{3/10}(q_1)&R_{3/10}(q_1)&\mathrm{CNOT}_{01}
       &R_{-3/10}(q_1)&\mathrm{CNOT}_{01}\\
 \varepsilon_i&\eta&\eta&\eta&0&\eta&0
 \end{array}
\]
Define $\widehat F_6=F$ and $\widehat F_{i-1}=W_{\widehat g_i}(\widehat F_i)$.
The sequencing argument is
\[
 \begin{aligned}
 W_{g_i}(F_i)-W_{\widehat g_i}(\widehat F_i)
 &=\bigl(W_{g_i}(F_i)-W_{\widehat g_i}(F_i)\bigr)
   +W_{\widehat g_i}(F_i-\widehat F_i).
 \end{aligned}
\]
The first term is bounded by $\varepsilon_i I$. Positivity and preservation of $I$
carry the suffix allowance through the second. Reading the row backwards gives
$0,\eta,\eta,2\eta,3\eta,4\eta$. Therefore
\[
 \pwire{J_{01}(-3/10)}
 \xrightarrow{\quad4\eta\quad}\pwire{\widehat J_{01}},
 \qquad 4\eta=\frac{210000000001}{1562500000000000000}<1.345\times10^{-7}.
\]
Reusing a local proof does not erase an occurrence from this sum.

\paragraph{Step 4: compose all source instructions and lower to measurements.}
The 280 rotations use 101 distinct local angle certificates. Each occurrence
contributes its allowance. The Hadamards, CNOTs, and readouts contribute zero.
The complete sum is
\[
 \varepsilon_{\mathrm{Ising}}
 =\sum_{j=1}^{280}\varepsilon_j
 =\frac{78280510032873}{3125000000000000000}<2.505\times10^{-5}.
\]
The resulting Clifford+$T$ circuit has 55,027 instructions. Its QPPM target $d$
uses $k=21{,}548$ supplied ideal magic states and two private classical locations
$u,v$, disjoint from $\Delta$. Gadget ancillas are scoped. Let $V_k$ insert those
resources and let $D_k$ lift public assertions, as in Section~\ref{sec:compilation}.
Exact lowering preserves the composed allowance. Instantiating the existing Hoare
transfer theorem with the effect $B=[E]$ and $A=W_c(B)$ gives
\[
 \vdash\dtriple{Q}{\pwire{d}}{D_k([E])},\qquad
 V_k^\dagger Q(\sigma)V_k=A(\pi(\sigma))-\varepsilon_{\mathrm{Ising}}I.
\]
This is the ordinary Hoare judgment, with the magic-state premise stated separately.

\paragraph{Step 5: read the guarantee as a probability bound.}
For any positive trace-one input $\rho$ and initial public store $s$, let $p_c$ be
the source probability of $E$, and $p_d$ the target probability, with input
$V_k\rho V_k^\dagger$ and two zero private bits. Forward/backward duality and the
two matrix inequalities of the certificate give
\[
 |p_d-p_c|\leq\varepsilon_{\mathrm{Ising}}<2.505\times10^{-5}.
\]
For the measured nearest-neighbor spin correlation, equal bits contribute $+1$
and unequal bits contribute $-1$, so its expectation is $2p-1$. Consequently
\[
 |(2p_d-1)-(2p_c-1)|\leq2\varepsilon_{\mathrm{Ising}}<5.010\times10^{-5}.
\]
This bounds the compilation error in a concrete observable of the complete simulation.
It makes no assumption about the observable's ideal value and introduces no new
command or Hoare rule.

\section{Kraus paths and the reset example}
\label{app:paths}

The instrument is constructed from a finite list
of pairs $(t,K)$, where $t$ is a final
store and $K$ is the matrix obtained by composing the operations along one execution path.
Such a matrix is called a Kraus operator. Call this list $\mathcal{B}(c,s)$. For skip
and assignment it contains one pair, with matrix $I$ and the appropriate unchanged or
updated store. For a measurement it contains
\[
 \mathcal{B}(M_P(r),s)=
 [(s[r\mapsto0],\Pi^P_0),(s[r\mapsto1],\Pi^P_1)].
\]
Correction gives the single pair $(s,P)$. A conditional chooses its branch using the
initial store for that command. Sequencing composes paths in execution order:
\[
 \mathcal{B}(c\seq d,s)=
 [\,(t,LK)\mid(u,K)\in\mathcal{B}(c,s),\ (t,L)\in\mathcal{B}(d,u)\,].
\]
The intermediate store $u$ is passed to the continuation, which is where a recorded
measurement result influences a later guard. When a program overwrites a bit, distinct paths can finish with the same store.

For a scope initialized to $b$, each body path $(t,K)$ contributes two paths
$(t,E_0^\dagger K E_b)$ and $(t,E_1^\dagger K E_b)$. These account for discarding
the first qubit. The final instrument groups paths by their retained classical store:
\begin{equation}
 \Inst_c(s,\rho,t)=
 \sum_{(u,K)\in\mathcal{B}(c,s)}[u=t]\,K\rho K^\dagger.
 \label{eq:path-instrument}
\end{equation}
The bracket here is a scalar indicator. Paths that finish with the same store contribute to the same sum. An impossible path
contributes the zero matrix.

To connect this construction to the logic, structural induction first establishes
\[
 \WP(c,A)(s)=\sum_{(t,K)\in\mathcal{B}(c,s)}K^\dagger A(t)K.
\]
The sequence case distributes matrix multiplication over the finite path lists.
The scope case sums the two ancilla blocks. Cyclicity of trace then gives
Equation~\eqref{eq:duality}. Positivity follows because every term in
Equation~\eqref{eq:path-instrument} has the form $K\rho K^\dagger$ with $\rho\succeq0$. Taking the constant
identity assertion in the duality equation and using $\WP(c,I)=I$ proves preservation
of total trace. The same path construction therefore supplies both the output matrices and the weakest
preconditions used by the logic.

\paragraph{The two reset paths.}
Use the reset inputs $\Gamma=(q)$, $\Delta=(r,u)$ and let $s_b=s[r\mapsto b][u\mapsto b]$.
Following the two rows of the forward reset
diagram gives
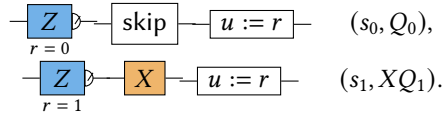
\[
 \begin{array}{c@{\quad}c@{\quad}c}
 \pmeter[Z]{Z}{r=0}\!\pwire{\cmd{skip}}\!\pwire{u:=r}
   &(s_0,Q_0),\\[8pt]
 \pmeter[Z]{Z}{r=1}\!\pwire[X]{X}\!\pwire{u:=r}
   &(s_1,XQ_1).
 \end{array}
\]
Assignment contributes an identity matrix, so
$\mathcal B(c_{\mathrm{reset}},s)=[(s_0,Q_0),(s_1,XQ_1)]$.
For an arbitrary assertion $A$, the path formula reads
\[
 \WP(c_{\mathrm{reset}},A)(s)
 =Q_0 A(s_0)Q_0+Q_1 X A(s_1)XQ_1.
\]
For $A=B_{\mathrm{reset}}$, both $A(s_b)=Q_0$.
The second term is $Q_1(XQ_0X)Q_1=Q_1$, and the first is $Q_0$.
This recovers the precondition $I$ calculated from the syntax.

\paragraph{Four paths, two records.}
With inputs $\Gamma=()$, $\Delta=(r)$, take
$f=\cmd{let}\ \nu=\cmd{ancilla}\ |0\rangle\ \cmd{in}\ M_{X_\nu}(r)$. The scope construction gives
$K_{bj}=E_j^\dagger\Pi^X_bE_0$.
Multiplying the matrices in Appendix~\ref{app:scope-example} yields
\[
 \begin{array}{c|cc}
   &j=0&j=1\\ \hline
   b=0&1/2&1/2\\
   b=1&1/2&-1/2
 \end{array}
 \qquad
 \underbrace{\pwire{E_0}\!\pwire{\Pi^X_b}\!\pwire{E_j^\dagger}}
             _{K_{bj}}.
\]
Each path contributes $K_{bj}\,1\,K_{bj}^\dagger=1/4$.
Grouping the two paths with record $b$ gives $1/2$. The index $j$ belongs to the
discard operation and is not a classical result available to the program.

\section{Repeated measurement and a classical copy}
\label{app:repetition}
\label{sec:derive-measure}

We first specify the state produced by a measurement and then use that specification
to reason about a second measurement of the same product. Fix input contexts $\Gamma,\Delta$,
a signed product $P$ on $\Gamma$, and a destination $r\in\Delta$.
Write $A_P(s)=\Pi^P_{s(r)}$.

\paragraph{Steps 1--2: substitute the record and simplify.}
On outcome $b$, the updated store reads $r=b$, so $A_P(s[r\mapsto b])=\Pi^P_b$.
Equation~\eqref{wp:measure}, idempotence, and \ref{rule:primitive} give the derivation over the
measurement diagram:
\begin{equation}
 \frac{\WP(M_P(r),A_P)(s)
       =\sum_b\Pi^P_b\underbrace{A_P(s[r\mapsto b])}_{\Pi^P_b}\Pi^P_b
       =\sum_b\Pi^P_b=I}
      {\dtriple{I}{\pmeter{P}{r}}{A_P}}.
 \label{eq:measure-records}
\end{equation}
For every initial state, the output lies in the eigenspace indexed by the recorded bit.
A continuation can therefore refer to that eigenspace by reading $r$.

\paragraph{Step 3: read the established eigenspace again.}
Let $u\in\Delta$, $u\ne r$, and define
\[
 B_P(s)=\begin{cases}\Pi^P_{s(r)}&s(u)=s(r),\\0&\text{otherwise}.\end{cases}
\]
Writing $a=s(r)$, Equation~\eqref{wp:measure} for the second measurement gives
\begin{align*}
 \WP(u:=\cmd{measure}\ P,B_P)(s)
  &=\sum_b\Pi^P_b
       \bigl(\text{if }b=a\text{ then }\Pi^P_a\text{ else }0\bigr)\Pi^P_b\\
  &=\Pi^P_a\Pi^P_a\Pi^P_a=\Pi^P_a=A_P(s).
\end{align*}
The condition $u\ne r$ ensures that updating $u$ leaves the read of $r$ unchanged.
Rule \ref{rule:seq} joins the matching assertion boundaries:
\[
 \{I\}\,\pmeter{P}{r}\,\{A_P\}\,\pmeter{P}{u}\,\{B_P\}.
\]

\paragraph{Step 4: replace the second measurement by an assignment.}
By \eqref{wp:assign}, substituting $r$ for $u$ makes the agreement test true. The same assertion
therefore supports a classical copy:
\[
 \{I\}\,\pmeter{P}{r}\,\{A_P\}\,\pwire{u:=r}\,\{B_P\}.
\]
The optimizer proves the stronger assertion-independent equivalence. For any postcondition
$Q$, expanding both measurements produces terms
\[
 \Pi^P_a\Pi^P_b\,
 Q(s[r\mapsto a][u\mapsto b])\,
 \Pi^P_b\Pi^P_a.
\]
Every term with $a\ne b$ is zero. The remaining terms are precisely those obtained by
measuring once and copying the result. When $u=r$, two successive updates
to the same destination collapse to the last update, and the surviving outcomes agree.
Hence the rewrite also applies when $u=r$, although the earlier two-record specification
used distinct destinations.

For a concrete instance, take $\Gamma=(q)$, $\Delta=(r,u)$, put $P=Z$, and start in $\rho_+$.
The first measurement gives $Q_b/2$ with record $r=b$. Orthogonality makes the second
measurement deterministic within each branch:
\[
 \underbrace{Q_0/2}_{r=0}\ \pmeter[Z]{Z}{u}\ %
 \underbrace{Q_0/2}_{r=u=0}
 \qquad
 \underbrace{Q_1/2}_{r=1}\ \pmeter[Z]{Z}{u}\ %
 \underbrace{Q_1/2}_{r=u=1}.
\]
Indeed, $Q_b(Q_b/2)Q_b=Q_b/2$, while $Q_{1-b}(Q_b/2)Q_{1-b}=0$.
Replacing either displayed meter by the assignment $u:=r$ therefore preserves its
store and output matrix. The assertion-independent proof above extends this calculation
from $\rho_+$ to every input.

\section{Measurement elimination: the scanner}
\label{app:scanner}

Fix input contexts $\Gamma,\Delta$, products $P,Q$ on $\Gamma$, locations $r,u\in\Delta$,
and a command $d$ on these registers. The scanner takes $P$, its saved destination $r$,
and the remaining command $d$. Its invariant is
\begin{equation}
 \pwire{\mathsf{scan}(P,r,d)}\ \equiv\ \pmeter{P}{r}\!\pwire{d}.
 \label{eq:scanner-invariant}
\end{equation}
For a measurement prefix, define
\[
\begin{aligned}
 &\mathsf{scan}(P,r,M_Q(u)\seq d)={}\\
 &\quad\begin{cases}
 M_P(r)\seq u:=r\seq d & P=Q,\\
 M_Q(u)\seq \mathsf{scan}(P,r,d)
    & P\ne Q,\ PQ=QP,\ r\ne u,\\
 M_P(r)\seq M_Q(u)\seq d & \text{otherwise}.
 \end{cases}
\end{aligned}
\]
For every other command $d$, put $\mathsf{scan}(P,r,d)=M_P(r)\seq d$.
A terminal measurement is treated as a prefix with continuation $\cmd{skip}$.
Product equality includes the sign. Immediate repetition may copy into its own destination.
An intervening write to that destination blocks motion. Induction follows the scanner's
cases, using projector orthogonality, the swap theorem, and the unchanged barrier case.
Sequence joining reassociates the command tree and invokes the scanner at leading measurements.
Its correctness and count inequalities compose with the scanner's local proof.

\paragraph{Follow the recursive calls.}
In Figure~\ref{fig:measurement-rewrite}, take $P=Z\otimes I$, $Q=I\otimes X$,
$R=I\otimes Z$, with distinct destinations $r_0,r_1,r_2,r_3$.
Writing the continuation after the last measurement as $\cmd{skip}$, the calls are
\[
 \begin{aligned}
 &\mathsf{scan}(P,r_0,M_Q(r_1)\seq M_R(r_2)\seq M_P(r_3))\\
 &\quad=M_Q(r_1)\seq \mathsf{scan}(P,r_0,M_R(r_2)\seq M_P(r_3))\\
 &\quad=M_Q(r_1)\seq M_R(r_2)\seq \mathsf{scan}(P,r_0,M_P(r_3))\\
 &\quad=M_Q(r_1)\seq M_R(r_2)\seq M_P(r_0)\seq r_3:=r_0.
 \end{aligned}
\]
The first two lines use the commuting case. The last uses product equality.
They are the two motions and final copy drawn in the figure.
For a noncommuting prefix on $\Gamma=(q)$, $\Delta=(r_0,r_1)$, take $P=Z$, $Q=X$,
and any continuation $d$ on these registers. The scanner returns
$\pmeter[Z]{Z}{r_0}\!\pmeter[X]{X}{r_1}\!\pwire{d}$ unchanged.

\subsection{Elimination across commuting measurements}
\label{sec:derive-run}

Write $P=Z\otimes I$, $Q=I\otimes X$, and $R=I\otimes Z$ for the four-measurement
example. Moving a measurement requires both commuting quantum operations and independent
classical updates. Tensor factors show that $P$ commutes with both $Q$ and $R$. In contrast,
$QR=-RQ$. The compiler needs only the first two commutation facts.

\paragraph{Step 1: move the pending measurement past $Q$.}
The destinations $r_0$ and $r_1$ differ, so their store updates commute. Quantum projectors
also commute because their measured products do. For outcomes $a,b$ and an arbitrary
postcondition matrix $A$, the relevant equality is
\[
 \Pi^P_a(\Pi^Q_b A\Pi^Q_b)\Pi^P_a
 =\Pi^Q_b(\Pi^P_a A\Pi^P_a)\Pi^Q_b.
\]
Combining these equalities with \eqref{wp:measure} and \eqref{wp:seq} proves
$M_P(r_0)\seq M_Q(r_1)\equiv M_Q(r_1)\seq M_P(r_0)$, including every final-store block.

\paragraph{Steps 2--3: move past $R$ and eliminate repetition.}
Applying the swap theorem to $R,r_2$ gives the middle circuit of
Figure~\ref{fig:measurement-rewrite}. The order of the noncommuting pair $Q,R$ is unchanged.
Projector idempotence then replaces the final measurement by a copy, giving the rightmost
circuit. It has three measurements and preserves all four records on every input matrix.
Both associations of the input sequence reduce to this result.

\paragraph{Step 4: lift the rewrite through the program.}
The equation holds for all postconditions. Substituting a continuation's weakest
precondition lifts it through sequencing by \eqref{wp:seq}. Equations~\eqref{wp:if} and \eqref{wp:ancilla}
lift it through conditionals and ancilla scopes. Thus the rewrite preserves an enclosing program's weakest precondition.

The recursive calls above perform these two swaps and the final copy.

\section{A single magic-state \texorpdfstring{$T$}{T} gate}
\label{app:resource-example}

Take data input $\Gamma=(d)$, private locations $\Delta=(a,b)$, $P=Z$, and $c=[T]$
in the magic-state contract of Section~\ref{sec:magic-resource}.
The compiled program has quantum inputs $(\mu,d)$, with $|M\rangle$ supplied on $\mu$.
There is one supplied magic state: $k=1$, $T=\operatorname{diag}(1,\omega)$, and
\[
 V_1\psi=|M\rangle\otimes\psi,\qquad
 D_1(B)(s)=I_2\otimes B(\pi(s)).
\]
The map $\pi$ forgets the two private outcome bits and retains every public bit.

\paragraph{Step 1: choose an assertion that observes the phase.}
Use constant assertions
\[
 A=|+\rangle\langle+|=\tfrac12\begin{pmatrix}1&1\\1&1\end{pmatrix},\qquad
 B=|M\rangle\langle M|=\tfrac12\begin{pmatrix}1&\overline\omega\\\omega&1\end{pmatrix}.
\]
Both are effects. Pulling $B$ backwards through the source gate gives
\[
 T^\dagger B T
 =\tfrac12\begin{pmatrix}1&0\\0&\overline\omega\end{pmatrix}
   \begin{pmatrix}1&\overline\omega\\\omega&1\end{pmatrix}
   \begin{pmatrix}1&0\\0&\omega\end{pmatrix}=A,
 \qquad |\omega|^2=1.
\]

\paragraph{Step 2: derive the compiled triple.}
Put $Q=\WP(\mathcal L([T]),D_1(B))$. The weakest-precondition derivation and the
injection contract give
\[
 \begin{gathered}
 \htriple{Q}{\mathcal L([T])}{D_1(B)},\qquad
 V_1^\dagger Q(s)V_1=T^\dagger BT=A.
 \end{gathered}
\]
Thus $Q$ witnesses the circuit judgment for $A,B$. The postcondition ignores the
consumed magic state and private outcomes. Compression supplies the input assumption.

\paragraph{Step 3: check the magic-state premise.}
For data input $A$, the expected value of $B$ is one. Replacing the supplied
magic state by $|0\rangle$ instead gives
$E_0^\dagger Q(s)E_0=Q_0BQ_0+Q_1BQ_1=I/2$.
Duality calculates the expected value of $B$ from these preconditions:
\[
 \begin{array}{c@{\;\longrightarrow\;}c@{\;\longrightarrow\;}l}
 |M\rangle\otimes|+\rangle&\pwire{\mathcal L([T])}&\tr(A^2)=1\\[5pt]
 |0\rangle\otimes|+\rangle&\pwire{\mathcal L([T])}&\tr(A\,I/2)=1/2.
 \end{array}
\]
The same compiled program therefore fails the certainty guarantee with the wrong magic state.

\section{Ancilla discard and an idle entangled qubit}
\label{sec:derive-entanglement}

Preparing a fresh qubit and immediately discarding it is the identity on the original register.
Discarding an existing qubit and preparing a replacement is a different operation. For the
input context $(d,e)$, Bell state $|\Phi\rangle=(|00\rangle+|11\rangle)/\sqrt2$, and
$\rho_\Phi=|\Phi\rangle\langle\Phi|$, we prove positivity and unit trace. The two
orders of preparation and discard have different results:
\[
\begin{aligned}
 \rho\;&\xrightarrow{\ \pwire{\mathsf{prep}_0}\!\pwire{\mathsf{drop}}\ }\;\rho,\\
 \rho_\Phi\;&\xrightarrow{\ \pwire{\mathsf{drop}}\ }\;I/2
 \;\xrightarrow{\ \pwire{\mathsf{prep}_0}\ }\;|0\rangle\langle0|\otimes I/2\ne\rho_\Phi.
\end{aligned}
\]
Partial trace removes the correlations with the discarded qubit. Preparing a new qubit
does not restore them.

For a command $c$ on $\Gamma,\Delta$ and a fresh name $e\notin\Gamma$, we derive the idle-qubit rule
\begin{equation}
 \frac{\htriple{A}{c}{B}}
 {\{A\otimes I\}\quad
 \begin{tikzpicture}[baseline=(current bounding box.center),font=\small]
  \draw (0,0) -- (2,0) (0,-.55) -- (2,-.55);
  \node[draw,fill=white,inner sep=4pt] at (1,0) {$c$};
  \node[fill=white,inner sep=1pt,font=\scriptsize] at (1,-.55) {idle qubit};
 \end{tikzpicture}
 \quad\{B\otimes I\}}.
 \label{eq:idle-proof}
\end{equation}

Here $c^\uparrow$ has inputs $(\Gamma,e),\Delta$ and leaves $e$ unchanged. The proof retains the
off-diagonal blocks of joint assertions, so it applies to entangled inputs. Exact
equivalence and unitary component specifications extend as well. Applying the compiled phase
component to the first half of this Bell state, with the last qubit idle, gives the lifted
unitary action $(S\otimes I)\rho_\Phi(S^\dagger\otimes I)$. Its $(00,11)$ entry is
$-i/2$. This entry records the phase coherence between $|00\rangle$ and $|11\rangle$.
The extension theorem transports the application Hoare proof as well, so its outcome
agreement remains certain for the joint Bell input. Instantiating the rule with the
application triple from Equation~\eqref{eq:running}, with inputs $(d,e)$ and $(a,b,r,t)$, gives
\[
 \frac{\htriple{I}{c_{\mathrm{app}}}{B}}
      {\htriple{I\otimes I}{c_{\mathrm{app}}^\uparrow}{B\otimes I}}.
\]
Here $(B\otimes I)(s)=[s(t)=s(r)](\Pi^Y_{s(r)}\otimes I)$ ignores the last qubit
while still checking the public records and the first qubit's eigenspace.
This exact rule does not by itself extend approximation allowances to a larger register.

\section{Numerical instances of the approximation rules}
\label{app:error-examples}

With inputs $\Gamma=()$, $\Delta=(r)$, the commands $f$ and $z$
in Figure~\ref{fig:fair-bit} satisfy $f\approx_{1/2}z$.
For $G(s)\equiv s(r)=0$, \ref{rule:primitive} and \eqref{wp:assign} prove $\triple{1}{z}{[G]}$.
Symmetry and Theorem~\ref{thm:approximation} give
\[
 \frac{\htriple{1}{r:=0}{[G]}\qquad z\approx_{1/2}f}
      {\dtriple{\tfrac12}{
        \pwire{|0\rangle}\!\pmeter[X]{X}{r}\!\pwire{\mathsf{drop}}}{[G]}}.
\]
The exact forward calculation gives $\Pr_f[G]=1/2$, so the transferred guarantee is
attained. For $F(s)\equiv s(r)=1$, the exact precondition of $z$ is zero.
The same theorem yields $\triple{-1/2}{f}{[F]}$.
This valid but weak bound illustrates why the assertion language allows preconditions
outside the interval $[0,I]$. The exact precondition of $f$ is $1/2$.

\paragraph{A scope instance.}
Keep external inputs $\Gamma=()$, $\Delta=(r)$. Inside each scope, the body context is
$(\nu)$. Take $d_0=(r:=0)$ and $d_1=(r:=1)$ on that qubit and location $r$.
For any effect $A$, their preconditions are $A(s[r\mapsto0])$ and
$A(s[r\mapsto1])$. Both lie between $0$ and $I$, so $d_0\approx_1 d_1$.
The scope case of Theorem~\ref{thm:approximation}, using \eqref{wp:ancilla}, keeps allowance one:
\[
 \underbrace{\pwire{|0\rangle}\!\pwire{r:=0}\!\pwire{\mathsf{drop}}}
             _{\cmd{let}\ \nu=\cmd{ancilla}\ |0\rangle\ \cmd{in}\ d_0}
 \quad\approx_1\quad
 \underbrace{\pwire{|0\rangle}\!\pwire{r:=1}\!\pwire{\mathsf{drop}}}
             _{\cmd{let}\ \nu=\cmd{ancilla}\ |0\rangle\ \cmd{in}\ d_1}.
\]
The event $[r=1]$ has probability zero on the left and one on the right, so a smaller
allowance is impossible. This uses the same preparation on both sides, as required.

\paragraph{Positive and negative output blocks.}
For $f$ minus $z$ on their normalized scalar input, the two blocks are
$\Delta^{(0)}=-1/2$ and $\Delta^{(1)}=1/2$.
Their spectral decompositions and positive-eigenspace projectors are
\[
 \begin{array}{c|cc|c}
  t&\Delta^{(t)+}&\Delta^{(t)-}&\Pi_t^+\\ \hline
  r=0&0&1/2&0\\
  r=1&1/2&0&1
 \end{array}
 \qquad
 \sum_t\tr(\Pi_t^+\Delta^{(t)})
 =0(-1/2)+1(1/2)
 =\tfrac12.
\]
The maximizing assertion is therefore the event $[r=1]$ already used above.
In this example the trace norm is absolute value, giving
$\frac12\sum_t\|\Delta^{(t)}\|_1=\frac12(|-1/2|+|1/2|)=1/2$.
The total trace difference is zero, although neither individual block has trace zero.

\section{The phase measurement sequence: normalization and scope}
\label{app:phase-branches}

Section~\ref{sec:derive-phase} computes the eight Kraus operators
$K_{xyj}=\alpha_{xyj}U(P)$ of the phase measurement sequence from
Equation~\eqref{eq:phase-kraus}. There $E_j\psi=|j\rangle\otimes\psi$ includes the data
register under ancilla basis value $j$, and $E_j^\dagger$ takes that ancilla block.
The feedback acts after both outcomes are available. For example, the path with $x=y=j=0$
reduces to $K_{000}=(I-iP)/4$. The calculation includes negative products and identity
factors. For $P=Z$, we have $U(P)=S$ and $K_{000}=(1-i)S/4$.

\paragraph{Step 2: prove normalization.}
The displayed coefficients satisfy
$\alpha_{xyj}\overline{\alpha_{xyj}}=1/8$ for every $x,y,j$. It follows that
\[
 K_{xyj}\rho K_{xyj}^\dagger=\tfrac18 U(P)\rho U(P)^\dagger.
\]
For a fixed recorded pair $(x,y)$, discarding the ancilla sums the two $j$ terms and yields
$U(P)\rho U(P)^\dagger/4$. Summing the four outcome pairs yields $U(P)\rho U(P)^\dagger$.
Every measurement outcome contributes to this sum. The calculation holds for every input matrix.

\paragraph{Step 3: derive the component specification.}
For a public assertion $A$, the corresponding calculation pulls back each outcome pair to
$U(P)^\dagger A U(P)/4$. Since only the private bits change, the public store is
preserved. With $\pi$ forgetting the two private bits, the contract is
\[
\begin{aligned}
 &\WP\bigl(\cmd{let}\ \nu=\cmd{ancilla}\ |0\rangle\ \cmd{in}\bigl(
 M_{X\otimes P}(a)\seq M_{Y\otimes I}(b)\seq \\
 &\qquad\cmd{if}\ a\oplus b\ \cmd{then}\ \cmd{correct}\ (I\otimes P)
 \ \cmd{else}\ \cmd{skip}\bigr),A\circ\pi\bigr)(s)\\
 &\qquad=U(P)^\dagger A(\pi(s))U(P).
\end{aligned}
\]
Here $X_\nu P=X\otimes P$ and $Y_\nu=Y\otimes I$ in the body context.
The private bits are overwritten by the measurements before feedback reads them.
Consequently, the specification holds for every initial value of those bits. Their
locations must be distinct: if they alias, the second measurement overwrites the first
outcome and the exclusive-or guard is always false. Taking $P=Z$ gives the application
transformer $S^\dagger A S$, which maps the $Y$ projector back to the corresponding
$X$ projector.

\section{Certificate for the two-qubit circuit}
\label{app:circuit-certificate}

This appendix supplies the numerical and rewriting steps for
Section~\ref{sec:two-qubit-compilation}. The example uses the existing syntax,
gadget contracts, approximation rule, and scheduling checker.

\paragraph{The phase allowance.}
For $\theta=\pi/2+\delta$ and $\delta=\pi/2^{26}$, the phase function is
Lipschitz with constant one:
\[
 |e^{i\theta}-i|
 =|e^{i\theta}-e^{i\pi/2}|
 \leq|\theta-\pi/2|
 =\frac{\pi}{2^{26}}<\frac4{2^{26}}=2^{-24}.
\]
Thus the diagram's replacement $\pwire{R_\theta}\rightsquigarrow\pwire[Z]{S}$
has a checked scalar certificate. To convert it into the effect allowance used by
the logic, set $V=R_\theta$, $\widetilde V=V\otimes I$, and
$\widetilde S=S\otimes I$. Let
$E_{ij}=|i\rangle\langle j|\otimes I$, for $i,j\in\{0,1\}$.
Expanding conjugation in these matrix units gives
\[
 \widetilde S^\dagger B\widetilde S-\widetilde V^\dagger B\widetilde V
 =\sum_{i,j,k,l=0}^1
   (\overline{S_{ki}}S_{lj}-\overline{V_{ki}}V_{lj})E_{ik}B E_{lj}.
\]
The only possibly nonzero coefficient differences are $i-e^{i\theta}$ and
$-i-e^{-i\theta}$. Their absolute values agree and are at most $2^{-24}$.
For $0\preceq B\preceq I$, the operator norms of $B$ and all $E_{ij}$ are at most one.
Applying the uniform bound to the sixteen terms, as in the synthesis checker, gives
\[
 \|\widetilde S^\dagger B\widetilde S-\widetilde V^\dagger B\widetilde V\|_{\mathrm{op}}
 \leq16\cdot2^{-24}=2^{-20}.
\]
The matrix difference is Hermitian, so it lies between $-2^{-20}I$ and $2^{-20}I$.
This calculation acts on the full two-qubit assertion $B$ and makes no separability
assumption. Exact gates before and after this replacement preserve the allowance:
the suffix pulls back an effect, and the prefix preserves the two matrix inequalities.
The five source blocks therefore contribute allowances $0,0,0,2^{-20},0$.

\paragraph{Each exact rewrite.}
Use $P=Z\otimes I$, $Q=I\otimes Z$, $R=I\otimes X$ as in the main text.
For the purpose of fusion, write the final correction as $U(R)^2=R$.
The checked sequence of edits is
\[
\begin{array}{c@{\quad}l}
 \text{input}&U(P)\seq U(Q)\seq U(P)\seq U(R)\seq U(PR)^3\seq U(P)\seq U(R)^2\\
 \text{swap}&U(Q)\seq U(P)\seq U(P)\seq U(R)\seq U(PR)^3\seq U(P)\seq U(R)^2\\
 \text{fuse}&U(Q)\seq U(P)^2\seq U(R)\seq U(PR)^3\seq U(P)\seq U(R)^2\\
 \text{swap}&U(Q)\seq U(P)^2\seq U(R)\seq U(P)\seq U(PR)^3\seq U(R)^2\\
 \text{swap}&U(Q)\seq U(P)^2\seq U(P)\seq U(R)\seq U(PR)^3\seq U(R)^2\\
 \text{fuse}&U(Q)\seq U(P)^3\seq U(R)\seq U(PR)^3\seq U(R)^2\\
 \text{swap}&U(Q)\seq U(P)^3\seq U(R)\seq U(R)^2\seq U(PR)^3\\
 \text{fuse}&U(Q)\seq U(P)^3\seq U(R)^3\seq U(PR)^3 .
\end{array}
\]
The swaps use $PQ=QP$, $P(PR)=(PR)P$, $PR=RP$, and $R(PR)=(PR)R$,
respectively. Each follows by multiplying the displayed Pauli factors.
The fusions add the two exponents modulo four. No swap crosses the noncommuting
pair $Q,R$, and no approximation is used in these edits.

\paragraph{From macro layers to logical rounds.}
The first optimized layer uses distinct private bits $(r_Q,s_Q)$ and $(r_P,s_P)$.
Its serial primitive sequence is $J_Q\seq Y_Q\seq F_Q\seq J_P\seq Y_P\seq F_P$, where $J$ denotes
the joint measurement, $Y$ the ancilla readout, and $F$ the parity-controlled correction.
These are the following two paths through the first part of
Figure~\ref{fig:two-qubit-ppm}:
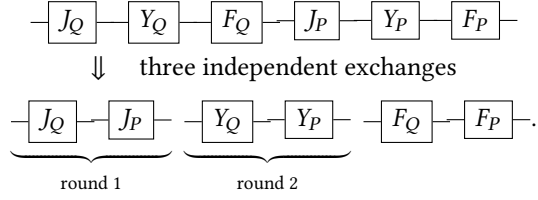
\[
 \begin{array}{c}
  \pwire{J_Q}\!\pwire{Y_Q}\!\pwire{F_Q}\!
  \pwire{J_P}\!\pwire{Y_P}\!\pwire{F_P}\\[6pt]
  \Downarrow\quad\text{three independent exchanges}\\[6pt]
  \underbrace{\pwire{J_Q}\!\pwire{J_P}}_{\text{round 1}}\;
  \underbrace{\pwire{Y_Q}\!\pwire{Y_P}}_{\text{round 2}}\;
  \pwire{F_Q}\!\pwire{F_P}.
 \end{array}
\]
Move $J_P$ past $F_Q$ and $Y_Q$, then move $Y_P$ past $F_Q$.
The allocated quantum supports are disjoint and the private bits differ,
so all three primitive exchanges are checked. The remaining two layers use one
ancilla each, with separate preparation and discard, and reuse $(r_Q,s_Q)$.
The parallel baseline uses the same checks on its first two layers.
Its four pairs of logical rounds become three, while the peak ancilla count
remains two. Both concrete QPPM programs have the rounded circuit's public
weakest precondition, so the synthesized allowance survives allocation and scheduling.

\section{The error-limit experiment: cost and selection}
\label{app:budget-cost}

\paragraph{The proposal algorithm for fusion and parallel scheduling.}
The proposal algorithm for Section~\ref{sec:rotation-schedule} expands gates into signed
Pauli turns, scans backwards from each turn until a readout or a noncommuting turn, and
moves an equal axis across the intervening commuting turns, adding exponents modulo four
and erasing zero turns. It then builds precedence edges for pairs that cannot be
exchanged (the baseline also orders overlapping supports) and repeatedly takes ready
operations in source order, draining zero-cost corrections first, into layers of
disjoint supports until the ancilla budget is full, recording the adjacent swaps that
realize the chosen order. The checker replays those swaps and fusions, verifies the
layer supports and private resources, and checks the primitive schedule and its public
weakest precondition.

\paragraph{Serial depth.}
The targets in Table~\ref{tab:budget-sweep} contain conditional $T$ injections.
For their serial schedule, let $\mathsf{sd}(c)$ count logical rounds along the longest
syntactic branch. Measurement costs one. Skip, assignment, and correction cost zero.
The remaining cases are
\begin{equation}
 \begin{aligned}
 \mathsf{sd}(c\seq d)&=\mathsf{sd}(c)+\mathsf{sd}(d),\\
 \mathsf{sd}(\cmd{if}\ e\ \cmd{then}\ c\ \cmd{else}\ d)&=\max(\mathsf{sd}(c),\mathsf{sd}(d)),\\
 \mathsf{sd}(\cmd{let}\ \nu=\cmd{ancilla}\ |b\rangle\ \cmd{in}\ c)&=\mathsf{sd}(c).
 \end{aligned}
 \tag{D-Serial}\label{eq:serial-depth}
\end{equation}
This cost ignores preparation latency and Pauli feedback. Induction shows that
$\mathsf{sd}(c)$ is at most the syntactic measurement count. The sequence case adds the induction
hypotheses, and the conditional uses $\max(a,b)\leq a+b$. The existing lowering
bound therefore also bounds $D$ after exact measurement cleanup.

For example, the injection diagram in Section~\ref{sec:compilation} has one joint
measurement, a conditional two-measurement phase gadget, and one magic-state readout:
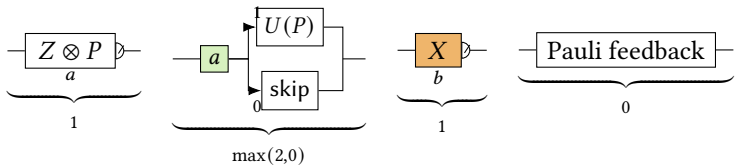
\[
 \underbrace{\pmeter{Z\otimes P}{a}}_{1}\quad
 \underbrace{\pchoice{a}{U(P)}{\cmd{skip}}}_{\max(2,0)}\quad
 \underbrace{\pmeter[X]{X}{b}}_{1}\quad
 \underbrace{\pwire{\text{Pauli feedback}}}_{0}.
\]
Equation~\eqref{eq:serial-depth} gives a bound of four. Each H or CNOT expands to
three two-measurement gadgets, S or its inverse to one, and a public readout to
one measurement. Thus, with gate counts $n_H,n_{\mathrm{CX}},n_S,n_T,n_M$,
\[
 \mathsf{sd}(d)\leq\overline{\mathsf{sd}}
   =6n_H+6n_{\mathrm{CX}}+2n_S+4n_T+n_M.
\]
For the $10^{-2}$ row, the emitted circuit has 472 H gates, 12 CNOTs, 222 S-class
gates, 471 T-class gates, and four readouts. Substitution gives
\[
 \overline{\mathsf{sd}}=6(472)+6(12)+2(222)+4(471)+4=5{,}236.
\]
The counts include inverse S and T gates. Pauli corrections add zero. Each
whole-circuit certificate checks this sum for the emitted instruction list and
instantiates the depth bound for its concrete compiled program. It also proves
the source-effect allowance for the row's requested error limit. The quantum semantics
still includes all ideal supplied magic inputs, even though their preparation is
outside this cost model.

\paragraph{Checking a selection.}
For a requested limit of $10^{-3}$, the coarsest candidate's certified allowance is
too large. The other two allowances meet the request, and $7{,}540<9{,}184$ selects
the middle candidate (Figure~\ref{fig:budget-selection}). Thus selection minimizes
the certified depth bound among feasible proposals. It need not minimize their error allowances.
A complete circuit certificate checks the selected target against this requested limit.

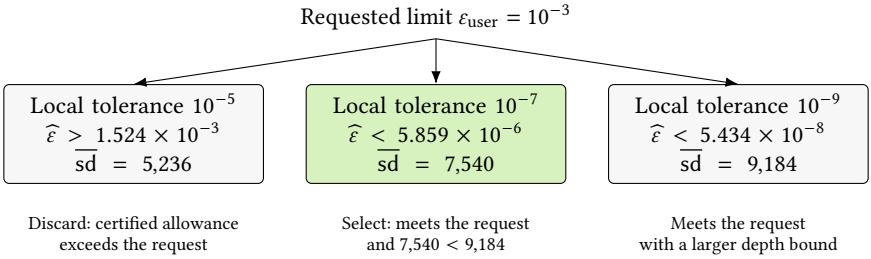
\begin{figure}[htbp]
\centering
\begin{tikzpicture}[font=\small,>=Latex,
 candidate/.style={draw,rounded corners=2pt,minimum height=1.25cm,
   text width=3.2cm,align=center,fill=black!3}]
 \node (request) at (4.0,1.55) {Requested limit $\varepsilon_{\mathrm{user}}=10^{-3}$};
 \node[candidate] (coarse) at (0,0)
   {Local tolerance $10^{-5}$\\
    $\widehat\varepsilon>1.524\times10^{-3}$\\$\overline{\mathsf{sd}}=5{,}236$};
 \node[candidate,fill=ppmclassical] (middle) at (4.0,0)
   {Local tolerance $10^{-7}$\\
    $\widehat\varepsilon<5.859\times10^{-6}$\\$\overline{\mathsf{sd}}=7{,}540$};
 \node[candidate] (fine) at (8.0,0)
   {Local tolerance $10^{-9}$\\
    $\widehat\varepsilon<5.434\times10^{-8}$\\$\overline{\mathsf{sd}}=9{,}184$};
 \draw[->] (request.south) -- (coarse.north);
 \draw[->] (request.south) -- (middle.north);
 \draw[->] (request.south) -- (fine.north);
 \node[anchor=north,align=center,font=\scriptsize] at ([yshift=-3mm]coarse.south)
   {Discard: certified allowance\\exceeds the request};
 \node[anchor=north,align=center,font=\scriptsize] at ([yshift=-3mm]middle.south)
   {Select: meets the request\\and $7{,}540<9{,}184$};
 \node[anchor=north,align=center,font=\scriptsize] at ([yshift=-3mm]fine.south)
   {Meets the request\\with a larger depth bound};
\end{tikzpicture}
\caption{Selection for a requested limit between the limits of
Table~\ref{tab:budget-sweep}. Comparisons use the exact rational certified allowances.
The smallest allowance belongs to the finest candidate. The smallest feasible
serial-depth bound belongs to the middle candidate.}
\Description{A requested error limit of one thousandth is tested against three certified
QFT candidates. The coarse candidate's allowance exceeds the request. The middle and
fine candidates meet it, with serial depth bounds 7540 and 9184. The middle candidate
is selected.}
\label{fig:budget-selection}
\end{figure}

The boundary tests accept a requested limit equal to the certificate's exact
rational allowance and reject a limit just below it. They also reject unchecked
candidates and return no result when the feasible set is empty. Three deliberately
understated allowance claims are rejected by the kernel. Source and target instruction
lists are compared with the exported certificate records and the emitted QASM.
The reported limits, the $10^{-3}$ example, and the limits at and just below each
candidate's allowance have checked finite selection certificates. Each checks
feasibility and the least depth bound, or verifies that every candidate exceeds the request.
An infeasible choice, a feasible choice with a larger bound, and an incorrect empty
result are rejected. Nine further mutations alter the stored evidence, including
error and depth claims, output gates, and the candidate collection. All nine mutations are rejected.

\noindent\begin{minipage}{\linewidth}
\paragraph{Reading the selected allowance.}
Let $c$ be the typed QFT source, with inputs $(q_0,q_1,q_2,q_3)$ and $(r_0,r_1,r_2,r_3)$,
and $d$ the selected middle target. Its classical inputs prepend private locations $a,b$.
Its certificate
uses $k=672$ ideal magic inputs and the allowance
$\widehat\varepsilon=\frac{18308100000009}{3125000000000000000}$.
We derive a bound for the event $E(t)\equiv t(r_0)=\cdots=t(r_3)=0$.

First choose the effect $[E]$, which is identity on that public record and zero
on the others. For a positive trace-one data state $\rho$ and initial store $s$,
extend $s$ by two zero private bits to obtain $\widetilde s$. The two executions are
\[
 \begin{array}{rcl}
 (s,\rho)&\xrightarrow{\ \pwire{c}\ }&p_c=\Pr_c[E],\\[3pt]
 (\widetilde s,V_k\rho V_k^\dagger)&\xrightarrow{\ \pwire{d}\ }&
 p_d=\Pr_d[E\circ\pi].
 \end{array}
\]
Second, pull the event back through each program. Duality gives
$p_c=\operatorname{Re}\tr(\rho W_c([E])(s))$ and
$p_d=\operatorname{Re}\tr(\rho C_E(s))$, where
$C_E(s)=V_k^\dagger\WP(d,D_k([E]))(\widetilde s)V_k$.
The latter probability sums over all private outcomes.

Finally, pair the certificate's two matrix inequalities with $\rho$:
\[
 \begin{aligned}
 |p_c-p_d|
 &=\left|\operatorname{Re}\tr\bigl(\rho(W_c([E])(s)-C_E(s))\bigr)\right|\\
 &\leq\widehat\varepsilon\tr(\rho)
   =\widehat\varepsilon<5.859\times10^{-6}.
 \end{aligned}
\]
Changing $E$ gives the same bound for any event on the four public output bits.
\end{minipage}

\section{A tight effect allowance}
\label{app:effect-bounds}
\subsection{A tight approximation allowance}
\label{sec:derive-fair}

With inputs $\Gamma=()$, $\Delta=(r)$, let $f$ prepare $|0\rangle$, measure $X$ into $r$,
and discard the qubit. Let $z$ assign
zero to $r$. Both have zero surviving quantum wires and one classical bit. A zero-qubit
operator is a $1\times1$ matrix, so an effect can be identified with a scalar between zero
and one for each final store. Figure~\ref{fig:fair-bit} places the circuit beside its syntax and output probabilities.

\begin{figure}[htbp]
\centering
\begin{tikzpicture}[font=\small,>=Latex]
 \node[anchor=east] at (0,0) {$f:$};
 \node[anchor=east] at (.65,0) {$|0\rangle_\nu$};
 \draw (.7,0) -- (2.7,0);
 \pic at (1.65,0) {ppmone={X/r}};
 \draw (2.7,-.14) -- (2.7,.14) (2.77,-.14) -- (2.77,.14);
 \node[font=\scriptsize] at (2.8,.38) {discard};
 \node[anchor=west,align=left] at (4.3,0)
   {$\cmd{let}\ \nu=\cmd{ancilla}\ |0\rangle\ \cmd{in}\ (r:=\cmd{measure}\ X_\nu)$\\[5pt]
    $\Pr[r=0]=\tfrac12,\qquad\Pr[r=1]=\tfrac12$};
 \node[anchor=east] at (0,-1.35) {$z:$};
 \node[draw,fill=ppmclassical,minimum width=2cm,minimum height=.45cm] at (1.65,-1.35) {$r:=0$};
 \node[anchor=west,align=left] at (4.3,-1.35)
   {$\Pr[r=0]=1,\qquad\Pr[r=1]=0$\\[5pt]
    $\tfrac12\sum_t\|\Delta_t\|_1=\tfrac12(|-\tfrac12|+|\tfrac12|)=\tfrac12$};
\end{tikzpicture}
\caption{Inputs: $\Gamma=()$, $\Delta=(r)$. The measurement circuit generates
a fair bit. The assignment produces zero. Their quantum outputs are both empty, but the
result distributions differ by $1/2$ in trace distance.}
\Description{The first program prepares zero, measures X into r, and discards the qubit,
producing a fair bit. The second assigns zero to r. Their output probabilities and their
distance of one half are written beside the programs.}
\label{fig:fair-bit}
\end{figure}
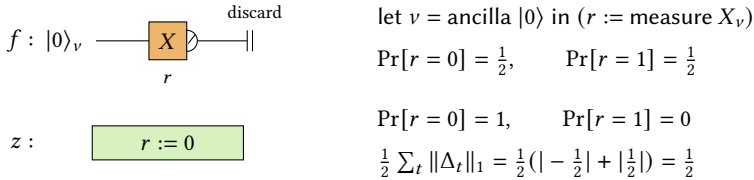

\paragraph{Step 1: compute the branch weights.}
The projectors applied to the prepared state are
\[
 \Pi^X_0|0\rangle=\tfrac12(|0\rangle+|1\rangle),\qquad
 \Pi^X_1|0\rangle=\tfrac12(|0\rangle-|1\rangle).
\]
Each branch has squared norm $1/2$. After discarding the qubit, the output consists of
two scalar blocks of weight $1/2$, one for each value of $r$. For any assertion $A$, write $A_b(s)=A(s[r\mapsto b])$. The backward calculation is
\begin{equation}
 \htriple{\tfrac12(A_0+A_1)}{f}{A},\qquad \htriple{A_0}{z}{A}.
 \label{eq:fair-wp}
\end{equation}
The preconditions in these judgments are exactly $\WP(f,A)$ and $\WP(z,A)$,
respectively. Equations~\eqref{wp:ancilla}, \eqref{wp:measure}, and \eqref{wp:assign} calculate them for
arbitrary assertions before imposing the inequalities required of effects.

\paragraph{Step 2: prove a bound for every effect.}
Subtracting the two expressions in Equation~\eqref{eq:fair-wp} gives
$(A_1-A_0)/2$. For an effect, $0\preceq A_b\preceq I$, so
\[
 \tfrac12 I-\tfrac12(A_1-A_0)
   =\tfrac12\bigl((I-A_1)+A_0\bigr)\succeq0,
\]
and the same calculation with $0$ and $1$ exchanged bounds the opposite difference.
Both approximation inequalities hold with $\varepsilon=1/2$. This proof uses only
positivity, closure under addition, and multiplication by a nonnegative scalar. It applies to every effect.

\paragraph{Step 3: exhibit a matching lower bound.}
Choose the event $E(s)=s(r)$, whose effect has $A_0=0$ and $A_1=I$. Its probability is
$1/2$ under $f$ and zero under $z$. Any uniform allowance must bound this event difference,
so it must be at least $1/2$. The upper and lower theorems together establish sharpness.

\paragraph{Step 4: retain the classical blocks in the distance.}
The false and true blocks differ by $-1/2$ and $1/2$, giving the distance shown in
Figure~\ref{fig:fair-bit}. Summing over the records first would cancel these differences,
although a continuation can distinguish the programs by reading $r$.

\subsection{Composition of approximation allowances}

The same local certificate can be used in two different program contexts. Sequentially
performing two fair-bit commands and replacing both by assignments gives
\[
 \pwire{f}\!\pwire{f}\quad\approx_{1/2+1/2}\quad\pwire{z}\!\pwire{z}.
\]
Theorem~\ref{thm:approximation} gives allowance $1$ without requiring distinct destinations.
Here the exact allowance is smaller because the second command overwrites the first result.
For any assertion $A$, let $F=\WP(f,A)=\tfrac12(A_0+A_1)$.
Substituting either value for the initial $r$ leaves $F$ unchanged, so \eqref{wp:seq} gives
\[
 \WP(f\seq f,A)(s)
 =\tfrac12\bigl(F(s[r\mapsto0])+F(s[r\mapsto1])\bigr)=F(s).
\]
Thus $f\seq f\equiv f$. Likewise, $\WP(z\seq z,A)=\WP(z,A)$, so $z\seq z\equiv z$.
The two sequences therefore satisfy $f\seq f\approx_{1/2}z\seq z$. The event $[r=1]$ again
witnesses a difference of $1/2$. The additive rule supplies a valid bound, while these
equalities establish the tight one.

\paragraph{Step 1: select the branch before overwriting its guard.}
Consider $c$, which runs $f$ if $r=1$ and $z$ otherwise. Replace only its true arm to
obtain $d$. Theorem~\ref{thm:approximation} gives the maximum allowance, $\max(1/2,0)=1/2$:
\[
 \frac{f\approx_{1/2}z\qquad z\approx_0 z}
 {\underbrace{\pchoice{r}{f}{z}}_{c}
   \quad\approx_{1/2}\quad
  \underbrace{\pchoice{r}{z}{z}}_{d}}.
\]
Both arms overwrite $r$, but the fork reads its initial value only once.
The two runs therefore select corresponding arms.

\paragraph{Step 2: calculate the weakest preconditions.}
Using \eqref{wp:if} and Equation~\eqref{eq:fair-wp}, for any assertion $A$,
\[
 \WP(c,A)(s)=
 \begin{cases}
  \tfrac12(A_0(s)+A_1(s)) & s(r)=1,\\
  A_0(s) & s(r)=0,
 \end{cases}
 \qquad \WP(d,A)(s)=A_0(s).
\]
For an effect, the difference lies between $-\tfrac12 I$ and $\tfrac12 I$ in the first
case and is zero in the second. The maximum covers both initial stores. Taking $A=[r=1]$ and $s(r)=1$
gives probabilities $1/2$ and zero, so no smaller uniform allowance suffices.

\paragraph{Step 3: check why the guard must agree.}
Negating only the guard gives a different program:
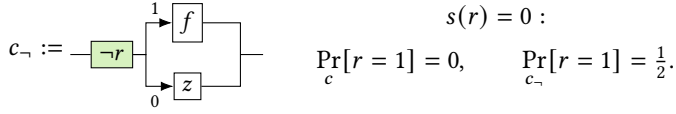
\[
 c_{\neg}:=\pchoice{\neg r}{f}{z}
 \qquad
 \begin{gathered}
  s(r)=0:\\
  \Pr_c[r=1]=0,\qquad \Pr_{c_{\neg}}[r=1]=\tfrac12.
 \end{gathered}
\]
The premises $f\approx_0 f$ and $z\approx_0 z$ still hold, but
$c\not\approx_0 c_{\neg}$: the displayed event forces an allowance of at least $1/2$.
Thus branchwise equivalence alone cannot justify changing a guard.

\paragraph{The general composition argument.}
Let $c,d,c',d'$ share input contexts, with $c\approx_\varepsilon c'$ and
$d\approx_\delta d'$, and let $B$ be an effect. The sequential transformer uses
\eqref{wp:seq}. Each vertical link below is a local approximation premise.
Effect preservation makes $\WP(d,B)$ an admissible observation for the first link, which
contributes $\varepsilon I$. Monotonicity of $\WP(c',\cdot)$ applies to the second link.
Linearity and identity preservation pull its $\delta I$ slack through unchanged.
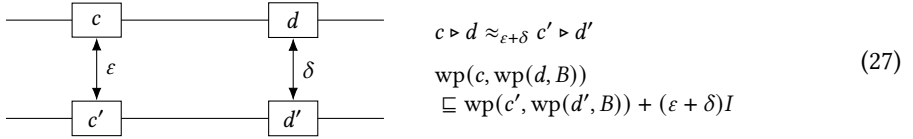
\begin{equation}
\begin{tikzpicture}[baseline=(current bounding box.center),font=\small,>=Latex]
 \foreach \y in {0,-1.3} {\draw (0,\y) -- (5,\y);}
 \foreach \x/\y/\t/\id in {1.2/0/c/c,3.8/0/d/d,1.2/-1.3/c'/cp,3.8/-1.3/d'/dp} {
  \node[draw,fill=white,minimum width=.65cm,minimum height=.45cm] (\id) at (\x,\y) {$\t$};
 }
 \draw[<->] (c.south) -- node[right] {$\varepsilon$} (cp.north);
 \draw[<->] (d.south) -- node[right] {$\delta$} (dp.north);
 \node[anchor=west,align=left] at (5.55,-.65)
  {$c\seq d\approx_{\varepsilon+\delta}c'\seq d'$\\[5pt]
   $\WP(c,\WP(d,B))$\\
   ${}\sqsubseteq\WP(c',\WP(d',B))+(\varepsilon+\delta)I$};
\end{tikzpicture}
\label{eq:approximation-proof}
\end{equation}
Transitivity adds the slacks. Reversing the inequalities gives the second half of
the approximation relation. By \eqref{wp:ancilla}, we lift the effect by $I_2\otimes B$, apply the body certificate, and compress
by the fixed preparation $E_b$. Because $E_b^\dagger I E_b=I$, the slack remains
$\varepsilon I$.

\subsection{Trace-distance characterization}

The distance theorem connects these local reasoning steps to an operational observation
on all final records. For a fixed positive input, form the Hermitian differences
$\Delta^{(t)}=\Inst_c(s,\rho,t)-\Inst_d(s,\rho,t)$. Total trace preservation gives
$\sum_t\tr(\Delta^{(t)})=0$. It does not give $\tr(\Delta^{(t)})=0$ for each $t$.
The fair-bit example is an explicit counterexample to that stronger premise.
Appendix~\ref{app:error-examples} gives its two numerical blocks, their spectral
projectors, and concrete instances of specification transfer and scope congruence.

The spectral proof chooses the assertion that projects each block onto its positive
eigenspaces. With $\Delta^{(t)}=\Delta^{(t)+}-\Delta^{(t)-}$ its positive/negative decomposition,
the observation diagram is
\[
 \Delta^{(t)}\;\xrightarrow{\ \pwire{\text{observe }\Pi_t^+}\ }\;
 \tr(\Pi_t^+\Delta^{(t)})=\tr(\Delta^{(t)+}),\qquad
 \sum_t\tr(\Delta^{(t)+})=\tfrac12\sum_t\|\Delta^{(t)}\|_1.
\]
The last equality uses equality of the total positive and negative masses.
Every other effect lies between zero and identity and cannot yield a larger pairing.
Applying the same argument to $-\Delta^{(t)}$ controls the absolute difference. Thus the
half-sum is characterized by the largest difference over store-indexed effects.
The artifact proves this characterization, then applies the two approximation inequalities
and the trace-pairing order lemma to obtain the normalized output-distance bound.

The local effect allowances therefore determine both the output trace-distance bound and
the weakening of a Hoare precondition.

\section{A local logical-fault contract}
\label{app:noise-example}

Take input registers $\Gamma=(q)$, $\Delta=(r)$, with $r$ reserved for the final observation.
The qubit suffers an $X$ fault with probability $p$, where
$0\leq p\leq1$. Its ideal transformer is the identity. The noisy transformer is
$N_p(B)=(1-p)B+pXBX$. For an arbitrary effect, the two required positive differences are
\[
 \begin{aligned}
 N_p(B)+pI-B&=p\bigl((I-B)+XBX\bigr)\succeq0,\\
 B+pI-N_p(B)&=p\bigl(B+(I-XBX)\bigr)\succeq0.
 \end{aligned}
\]
Here conjugation preserves effects. The same calculation works for the
weakest-precondition transformer of any QPPM command in place of $XBX$.
For $B=|0\rangle\langle0|$, this contract has a directly observable instance:
\[
 \begin{array}{c@{\quad}c@{\quad}l}
 |0\rangle\langle0| & \pwire{N_p} &
 (1-p)|0\rangle\langle0|+p|1\rangle\langle1|\\[5pt]
 &&\Pr[\,\pmeter[Z]{Z}{r=0}\,]=1-p.
 \end{array}
\]
This is an example of a local contract, not an assumption that all logical noise
is a stochastic Pauli fault.

The theorem sums local \emph{effect} bounds. When the same bound is expressed as
trace norm on trace-one states, its allowance is twice as large:
\[
 2\varepsilon+2(p_LN_L+p_TN_T+p_IN_I).
\]
A trace-norm synthesis allowance
$E_{\mathrm{synth}}$ therefore appears as
$E_{\mathrm{synth}}+2(p_LN_L+p_TN_T+p_IN_I)$, while a probability bound uses half
that quantity. Converting units does not supply a missing noise contract.

For surface-code or lifted-product surgery, the implementation must additionally
specify the encoding, extraction circuit, decoder, record interpretation, and schedule.
After encoding inputs and interpreting outputs, each implemented logical operation
must satisfy the stated local contract. Logical equivalence of two QPPM programs
then preserves their synthesis allowance. Their noise allowances must be recalculated
from their respective implementations. In particular, disjoint logical supports
need not yield disjoint routing regions or independent decoder resources.

\section{Additional verified benchmark results}
\label{app:large-coverage}

The large-suite corpus comprises all 56 canonical QASMBench inputs at the pinned
revision, of which 55 parse successfully. The valid inputs span 28--433 data qubits.
This appendix reports completed source-error certificates for 23 of those 55 inputs
and depth certificates for all eleven selected static Clifford inputs. These are
the certified subsets of the fixed corpus. The tables below list only completed results.

% Successful certificates from the large-suite exact-macro report.
Table~\ref{tab:large-coverage} gives fifteen complete exact-source certificates.
\begin{table}[htbp]
\centering\small
\begin{tabular}{lrrrr}
\toprule
Input & Qubits & Instructions & PPM bound & $\varepsilon$\\
\midrule
\texttt{adder\_n28} & 28 & 452 & 2,158 & 0\\
\texttt{adder\_n64} & 64 & 1,052 & 5,034 & 0\\
\texttt{adder\_n118} & 118 & 1,952 & 9,348 & 0\\
\texttt{bv\_n30} & 30 & 107 & 491 & 0\\
\texttt{bv\_n70} & 70 & 245 & 1,119 & 0\\
\texttt{bv\_n140} & 140 & 491 & 2,245 & 0\\
\texttt{cat\_n35} & 35 & 70 & 245 & 0\\
\texttt{cat\_n65} & 65 & 130 & 455 & 0\\
\texttt{cat\_n130} & 130 & 260 & 910 & 0\\
\texttt{cat\_n260} & 260 & 520 & 1,820 & 0\\
\texttt{ghz\_n40} & 40 & 80 & 280 & 0\\
\texttt{ghz\_n78} & 78 & 156 & 546 & 0\\
\texttt{ghz\_n127} & 127 & 254 & 889 & 0\\
\texttt{multiplier\_n45} & 45 & 5,990 & 30,573 & 0\\
\texttt{multiplier\_n75} & 75 & 17,092 & 87,315 & 0\\
\bottomrule
\end{tabular}
\caption{Verified exact-source compilations. Instructions count
the Clifford+$T$ source after macro expansion and barrier erasure. The PPM bound counts
logical measurements, not parallel depth. The complete 433-qubit adder and 280-qubit
Bernstein--Vazirani certificates are described below.}
\label{tab:large-coverage}
\end{table}

The checked adders and multipliers contain $T$ gates. Their expansion equalities are
checked in blocks of 128 source steps and composed over the complete list.
Each certificate proves the typed expansion, PPM lowering, public Hoare
specification, and zero approximation error, assuming ideal magic states for
$T$ instructions. Exported input records match the independently parsed source
expansion. Parsing and macro inlining remain outside the proof. A direct cost calculation
is proved equal to the measurement bound, avoiding construction of every gadget.

\paragraph{Register-controlled sources.}
The four \texttt{cc} inputs compare a public $n$-bit register with zero or
$2^{n-1}$. Both tests constrain every register bit. They expand into nested core
conditionals as in Appendix~\ref{app:guarded-source}, for every initial store.
Shared predicates avoid repeating the tests. Validity and measurement bounds are
checked in blocks of 32 steps.

% Generated from the complete guarded-source supplement.
\begin{table}[htbp]
\centering\small
\begin{tabular}{lrrrrl}
\toprule
Input & Source steps & Guarded gates & PPM bound & $\varepsilon$ & Result\\
\midrule
\texttt{cc\_n32} & 161 & 65 & 788 & 0 & checked\\
\texttt{cc\_n64} & 321 & 129 & 1,588 & 0 & checked\\
\texttt{cc\_n151} & 756 & 303 & 3,763 & 0 & checked\\
\texttt{cc\_n301} & 1,506 & 603 & 7,513 & 0 & checked\\
\bottomrule
\end{tabular}
\caption{Every large-suite input whose only missing frontend feature was classical control.
Source steps include barriers. Each successful check preserves every register condition
and proves a bound on primitive measurements. These are not parallel-depth results.}
\label{tab:large-guarded}
\end{table}

% Generated from validated complete Ising source evidence.
\paragraph{Decimal phases.}
The inputs \texttt{ising\_n34} and \texttt{ising\_n42} contain 67 and 83 distinct angles, including
zero. Each nonzero angle occurs twice. Decimal literals denote exact rational radians.
For example,
$1.6751132=4187783/2500000$, without a factor of $\pi$.
Ross--Selinger synthesis proposes each gate sequence at tolerance $10^{-8}$.
Exact coefficient calculations and Taylor remainders establish its allowance.
For $n\in\{34,42\}$, the complete compilation satisfies
\[
 \pwire{\mathrm{Ising}_{n}}
 \xrightarrow{\ \varepsilon_n=\sum_j\varepsilon_j\ }
 \pwire{\text{Clifford}+T}
 \xrightarrow{\ 0\ }\pwire{\text{QPPM}}.
\]
The sum includes every source rotation occurrence. The table reports source
steps, rotations, emitted Clifford+$T$ instructions, ideal magic inputs, and the
proved effect-error allowance.
\begin{center}
\small
\begin{tabular}{rrrrrl}
\toprule
$n$ & Steps & Rotations & Instructions & Magic inputs & $\varepsilon_n$ \\
\midrule
34 & 403 & 200 & 27,940 & 10,920 & $<1.285\times10^{-5}$ \\
42 & 499 & 248 & 34,880 & 13,584 & $<1.692\times10^{-5}$ \\
\bottomrule
\end{tabular}
\end{center}
Each proof retains the $2n$-bit public interface, including the unused register,
and exports every source step and output instruction for comparison with the parsed circuits.
Each angle certificate has a 120-second checking limit.
Complete composition and export have a separate 600-second limit. Checks use
two workers and a 6,000\,MB allocation limit, with at most two running concurrently.
These are source-error allowances. The parallel-depth experiment remains separate.

The fifteen exact-source rows, four register-controlled sources, two Ising sources,
433-qubit adder, and 280-qubit Bernstein--Vazirani circuit give the 23 distinct
source-error certificates reported above.

\paragraph{Parallel measurement schedules.}
The depth experiment includes all eleven static Clifford inputs selected before
scheduling. Each typed input retains every gate and readout. Parsing and removal
of identity barriers are external to the proof. The support-disjoint baseline and
commutation/fusion optimizer each use four scoped ancillas and eight private bits.
Replay checks every commutation and classical-store guard. Composition proves
resource validity, primitive counts, logical depth, and public weakest-precondition
equality for the complete program.

% Final accepted certificates for every preselected large Clifford depth input.
\begin{table}[htbp]
\centering\small
\begin{tabular}{@{}lrrrl@{}}
\toprule
Input & $D_0$ & $D_1$ & Reduction & Proof composition\\
\midrule
\texttt{bv\_n30} & 127 & 81 & 36.2\% & direct\\
\texttt{bv\_n70} & 267 & 179 & 33.0\% & direct\\
\texttt{bv\_n140} & 529 & 355 & 32.9\% & modular\\
\texttt{bv\_n280} & 1,069 & 715 & 33.1\% & modular\\
\texttt{cat\_n35} & 143 & 139 & 2.8\% & direct\\
\texttt{cat\_n65} & 263 & 259 & 1.5\% & direct\\
\texttt{cat\_n130} & 523 & 519 & 0.8\% & modular\\
\texttt{cat\_n260} & 1,043 & 1,039 & 0.4\% & modular\\
\texttt{ghz\_n40} & 163 & 159 & 2.5\% & direct\\
\texttt{ghz\_n78} & 315 & 311 & 1.3\% & direct\\
\texttt{ghz\_n127} & 511 & 507 & 0.8\% & modular\\
\bottomrule
\end{tabular}
\caption{Complete depth certificates for all eleven selected large Clifford inputs.
$D_0$ schedules disjoint supports. $D_1$ adds commutation and fusion. Both use four
ancillas and eight private bits. Every row certifies the allocated programs, the
displayed depths, preservation of public specifications, and zero source error under
ideal logical operations. Direct and modular checking establish the same contracts.}
\label{tab:large-depth}
\end{table}

% Complete larger source and scheduling certificates.
\paragraph{Composing larger certificates.}
Modular checking splits rewrite replay and readout proofs into separately checked
parts. Their composition establishes the complete source correspondence and public
weakest-precondition equality. The five modular depth results in
Table~\ref{tab:large-depth} use this construction without changing the circuits or
resource budgets. For the 280-qubit Bernstein--Vazirani circuit, fusion and scheduling
reduce logical rounds from 1,069 to 715 and measurements from 4,545 to 3,125.

The same method checks the complete 433-qubit adder source: 7,203 typed steps,
including one identity barrier, emit 7,202 Clifford+$T$ instructions and require
2,688 ideal magic inputs. Its PPM target has a proved bound of 34,513 measurements
and zero approximation error. This is a source-compilation result. It does not
assert a parallel schedule for the adder.

These counts describe ideal logical operations. Physical placement, routing, and
decoding require a separate implementation satisfying Theorem~\ref{thm:logical-noise}.
The artifact records the checking procedures and resource limits.

\end{document}